\RequirePackage{etoolbox}
\csdef{input@path}{%
 {sty/}% cls, sty files
 {img/}% eps files
}%
\csgdef{bibdir}{bib/}% bst, bib files

 \RequirePackage{setspace}
 
\documentclass[ba]{imsart}
\pubyear{0000}
\volume{00}
\issue{0}
\doi{0000}
\firstpage{1}
\lastpage{1}

\usepackage{amsthm}
\usepackage{amsmath,amssymb}
\usepackage{natbib}
\usepackage[colorlinks,citecolor=blue,urlcolor=blue,filecolor=blue,backref=page]{hyperref}
\usepackage{graphicx}

\startlocaldefs
\numberwithin{equation}{section}
\theoremstyle{plain}

\endlocaldefs

\usepackage{subfigure}
\usepackage{fancyhdr}
\usepackage{longtable}
\usepackage{natbib}
\usepackage[table]{xcolor}
\usepackage{tikz}
\usepackage{graphicx}
\usepackage{caption}
\usepackage{setspace}
\usepackage{amssymb}
\usepackage{subfigure}
\usepackage{a4wide}
\usepackage{setspace}
\usepackage{graphics}
\usepackage{ifpdf}
\usepackage{epic}
\usepackage{eepic}
\usepackage{subfigure}
\usepackage{courier}
\usepackage{natbib}
\usepackage{rotating}
\usepackage{titlesec,color}
\usepackage{titletoc}
\usepackage[colorlinks,citecolor=blue,urlcolor=blue,filecolor=blue,backref=page]{hyperref}

\usepackage[utf8]{inputenc}

\bibpunct{(}{)}{;}{a}{}{,} 

\def\a{\mbox{\boldmath$a$}}

\def\H{\mbox{\boldmath$H$}}
\def\bR{\mbox{\boldmath$R$}}
\def\A{\mbox{\boldmath$A$}}
\def\B{\mbox{\boldmath$B$}}
\def\bG{\bm{\Gamma}}
\def\bC{\mbox{\boldmath$C$}}

\def\x{\mbox{\boldmath$x$}}
\def\y{\mbox{\boldmath$y$}}
\def\b{\mbox{\boldmath$b$}}

\def\d{\mbox{\boldmath$d$}}

\def\B{\mbox{\boldmath$B$}}

\def\W{\mbox{\boldmath$W$}}

\def\m{\mbox{\boldmath$m$}}
\def\v{\mbox{\boldmath$v$}}

\def\bbeta{\mbox{\boldmath$\beta$}}

\def\seq#1#2{#1{:}#2}

\newcommand{\bm}[1]{\boldsymbol{#1}}
\newcommand{\E}{\mathsf{E\,}}

\renewcommand{\d}{\mathrm{d}}
\renewcommand{\P}{\mathsf{P}}

\newcommand{\mF}{\mathcal{F}}

\newcommand{\iid}{\stackrel{iid}{\sim}}

\newcommand{\wh}[1]{\smash{\widehat{#1}}}

 \usepackage{amssymb,tikz}
 
\newcommand{\var}{\mathsf{Var}\,}

\def\C {\,|\:}

\def\C {\,|\:}

\def\B{\bm{B}}

\def\A{\bm{A}}

\def\W{\bm{W}}
\def\y{\bm{y}}

\def\bg{\bm{\gamma}}

\def\b{\bm{\beta}}

\renewcommand{\d}{\mathrm{d}\,}
\newcommand{\e}{\mathrm{e}}

\usepackage[affil-it]{authblk}

 \usepackage{amssymb,tikz}

\newcommand{\mysetminusD}{\hbox{\tikz{\draw[line width=0.6pt,line cap=round] (3pt,0) -- (0,6pt);}}}
\newcommand{\mysetminusT}{\mysetminusD}
\newcommand{\mysetminusS}{\hbox{\tikz{\draw[line width=0.45pt,line cap=round] (2pt,0) -- (0,4pt);}}}
\newcommand{\mysetminusSS}{\hbox{\tikz{\draw[line width=0.4pt,line cap=round] (1.5pt,0) -- (0,3pt);}}}

\newcommand{\mysetminus}{\mathbin{\mathchoice{\mysetminusD}{\mysetminusT}{\mysetminusS}{\mysetminusSS}}}

\def\sls{{\mysetminus}}

\newtheorem{theorem}{Theorem}
\newtheorem{lemma}{Lemma}
\newtheorem{remark}{Remark}

\newtheorem{definition}{Definition}

\begin{document}

\begin{frontmatter}
\title{Dynamic Variable Selection \\with Spike-and-Slab Process Priors}
\runtitle{Dynamic Variable Selection with Spike-and-Slab Process Priors}

\begin{aug}
\author{VR}{Veronika  Rockova\thanksref{addr1,t1}\ead[label=e1]{veronika.rockova@chicagobooth.edu}},
\author{KM}{Kenichiro McAlinn\thanksref{addr1}\ead[label=e2]{kenichiro.mcalinn@chicagobooth.edu}}

\runauthor{Rockova and McAlinn}

\address[addr1]{Booth School of Business, University of Chicago,  5807 S Woodlawn Ave, Chicago, IL 60637
    \printead{e1} % print email address of "e1"
    \printead*{e2}
}

\thankstext{t1}{The authors gratefully acknowledge the support from the James S. Kemper Foundation Faculty Research Fund at the University of Chicago Booth School of Business.}

\end{aug}

\begin{abstract}
{We address the problem of dynamic variable selection in time series regression with unknown  residual variances, where the set of active predictors is allowed to evolve over time.}
%The spike-and-slab methodology for variable selection has been widely used in statistical data analysis since its inception more than two decades ago. However,  developments for varying coefficient models have not yet received much attention.
To capture  time-varying variable selection uncertainty, we introduce new dynamic shrinkage priors for the time series of regression coefficients. 
These priors are characterized by two main ingredients: smooth parameter evolutions and intermittent zeroes for modeling predictive breaks. 
More formally,  our proposed {\em Dynamic Spike-and-Slab (DSS)}  priors are  constructed as mixtures of two processes: a spike process for the irrelevant coefficients and a slab autoregressive process for the active coefficients. 
The mixing weights are themselves time-varying and depend on lagged values of the series. 
Our $DSS$ priors are probabilistically coherent in the sense that their stationary distribution is fully known and characterized by {\em spike-and-slab  marginals}.   %selection and smoothing and leads to a practically appealing  MAP estimation algorithm. 
For  posterior sampling  over dynamic regression coefficients, model selection indicators as well as unknown dynamic residual variances, we propose a {\em Dynamic SSVS}    algorithm based on
forward-filtering and backward-sampling. To scale our method to large data sets, we  develop a {\em Dynamic EMVS} algorithm for MAP smoothing. 
We demonstrate,  through simulation and a topical  macroeconomic dataset, that $DSS$ priors are very effective at separating active and noisy coefficients.
Our fast implementation significantly extends the reach of spike-and-slab methods to big time series data.

\end{abstract}


\end{frontmatter}

%\title{Dynamic Variable Selection \\with Spike-and-Slab Process Priors}  
%\author{Veronika Ro\v{c}kov\'{a} \& Kenichiro McAlinn\thanks{5807 S Woodlawn Ave, Chicago, IL 60637\\ (773) 702-7743\\
%\{veronika.rockova,kenichiro.mcalinn\}@chicagobooth.edu}\\
%\small Booth School of Business, University of Chicago\\
%This work was supported by the James S. Kemper Foundation Faculty Research Fund at the University of Chicago Booth School of Business. }

\section{Dynamic Sparsity}

For dynamic linear modeling with many potential predictors, the assumption of a static generative model with a fixed subset of regressors (albeit with time-varying regressor effects) may be misleadingly  restrictive.  By obscuring variable selection uncertainty over time, confinement to a single inferential model may lead to poorer predictive performance, especially when the actual effective subset at each time is sparse. The potential for dynamic model selection techniques in time series modeling has been  recognized \citep{fruhwirth_wagner,groen, nakajima_west, kalli_griffin,chan_koop}. In  inflation forecasting, for example, {large sets of predictors are available  and it is expected that} the forecasting model changes over time, not only its coefficients \citep{KoopKorobilis2012,groen,kalli_griffin,wright2009forecasting}.
%\citep{banbura2010large,nakajima_west,koop,groen,koop2013forecasting,korobilis2013var,giannone2014short,kalli_griffin,gefang2014bayesian,pettenuzzo2016bayesian}.
In particular, in recessions {we might see distress related factors be effective, while having no predictive power}  in expansions \citep{KoopKorobilis2012}. Motivated by such contexts, we develop a new dynamic shrinkage approach for time series models that exploits time-varying predictive subset sparsity.

We present our approach in the context of  dynamic linear models \citep{WestHarrison1997book2} (or varying coefficient models with a time effect modifier \citep{hastie}) that link a scalar response $y_t$ at time $t$ to a set of $p$ known regressors $\x_t=(x_{t1},\dots,x_{tp})'$ through the relation
\begin{equation}\label{model}
y_t=\x_t'\b_{t}^0+\varepsilon_t, \quad t=1,\dots, T,
\end{equation}
where $\b_{t}^0=(\beta_{t1}^0,\dots,\beta_{tp}^0)'$  is a time-varying vector of regression coefficients and where the innovations $\varepsilon_t$  come  from $\mathcal{N}(0,v_t)$. The observational variances $v_t$ are assumed to be unknown, where the precisions $\nu_t=1/v_t$ arise from the following Markov evolution model  (Chapter 10.8.2 of \cite{WestHarrison1997book2})
\begin{equation}\label{eq:dsvt}
\nu_t= c_t\nu_{t-1}/\delta,\quad\text{where}\quad c_{t}\sim\mathcal{B}(\delta n_{t-1}/2,(1-\delta)n_{t-1}/2)\quad\text{and}\quad n_t=\delta n_{t-1}+1
\end{equation}
with a discount parameter $\delta\in(0,1]$.
%However, the outreach of our proposed method is far broader than just this framework.
%The primary focus of this work will be on designing dynamic shrinkage priors on the regression coefficients, assuming  $\sigma^2_t=\sigma^2$. However, our framework can be combined with more flexible stochastic volatility priors.  %The regression coefficients are assumed to evolve according to some stochastic process. 

{The challenge of estimating the $T\times p$ coefficients in \eqref{model}, with merely $T$ observations, is typically made feasible with a smoothness inducing state-space model that} treats  $\{\b^0_t\}_{t=1}^T$ as realizations from a (vector autoregressive) stochastic process $
\b_{t}^0=f(\b_{t-1}^0)+\bm{e}_t$ with $\bm{e}_t\sim\mathcal{N}(0,\bm{\Lambda}_t)
$
for some $\bm{\Lambda}_t$ and $f(\cdot)$.
% where $f(\cdot)$ is a
%This makes the regression coefficients closely dynamically intertwined 
{Nevertheless, any regression model with a large number of potential predictors will still be vulnerable to overfitting.  This phenomenon is perhaps even more pronounced  here,}  where the regression coefficients are forced to be dynamically intertwined. The major concern is 
that overfitted coefficient evolutions disguise true underlying dynamics and provide misleading representations  with poor out-of-sample predictive performance.   
For long term forecasts, this concern is exacerbated by the proliferation of the state space.
As the model propagates forward, the non-sparse state innovation accumulates noise, further hindering the out-of-sample forecast ability.
With many potentially irrelevant predictors,  seeking sparsity is a natural remedy against the  loss of statistical efficiency and forecast ability.

We shall assume that $p$ is potentially very large, where possibly only a small portion of predictors is relevant for the outcome at any given time. 
{Besides  time-varying regressor effects,} we  adopt the point of view that the regressors are allowed to enter and leave the model as time progresses, rendering the subset selection problem ultimately dynamic. This anticipation can be reflected by the following sparsity manifestations in the matrix of regression coefficients $\B^0_{p\times T}=[\b^0_1,\dots,\b^0_T]$: (a) {\sl horizontal sparsity}, where each individual time series $\{\beta_{tj}^0\}_{t=1}^T$ (for $j=1,\dots,p$) allows for intermittent zeroes for when $j^{th}$ predictor is not a persisting predictor at all times, (b) {\sl vertical sparsity}, where only a subset of coefficients  $\b_t^0=(\beta_{t1}^0,\dots,\beta_{tp}^0)'$ (for $t=1,\dots,T$)  will be active at the $t^{th}$ snapshot in time.

%The main thrust of this work is to propose new stochastic process priors for the time series of regression coefficients that acknowledge both types of sparsity mentioned above and smoothness, at the same time. While our framework can be extended to  time-varying variances $\sigma^2_t$, throughout the paper we shall focus on the simpler case $\sigma^2_t=\sigma^2$. 

%{We are certainly not alone in proposing dynamic regularization for  time-varying subset selection and  shrinkage.} 
This problem has been addressed in the literature by multiple authors including, for example,  \cite{groen,belmonte, koop, kalli_griffin, nakajima_west}. 
 We should like to draw particular attention to the latent threshold process of  \cite{nakajima_west}, %which served as a springboard for our approach.
%We will denote by $\B=[\b_1,\dots,\b_T]$ the matrix of all regression coefficients in all $p$ dimensions at all times $t=1,\dots,T$. and by $\bG=[\bg_1,\dots,\bg_T]$ the corresponding  binary indicators.  The goal will be to develop dynamic spike-and-slab priors that account for temporal dependencies between coefficients in two distinct ways. First, an underlying
%variable selection process is assumed, where neighboring coefficients are likely to be selected/discarded jointly (a stochastic process for $\bg_t$). The second component is the smoothing process reflecting the similarity of neighboring active coefficients (a stochastic process for $\b_t$).
 a related regime switching scheme for either shrinking coefficients exactly to zero or for leaving them alone on their autoregressive path:
\begin{align}
\beta_{tj}&=b_{tj}\gamma_{tj},\quad\text{where}\quad \gamma_{tj}=\mathrm{I}(|b_{tj}|>d_j),\label{nw1}\\
b_{tj}&=\phi_{0j}+\phi_{1j}(b_{t-1j}-\phi_{0j})+e_t,\quad |\phi_{1j}|<1,\quad e_t\iid\mathcal{N}(0,\lambda_1).\label{nw2}
\end{align}
The model assumes a latent autoregressive process $\{b_{tj}\}_{t=1}^T$, giving rise to the actual coefficients $\{\beta_{tj}\}_{t=1}^T$ only when it meanders away from a latent basin around zero $[-d_j,d_j]$. 
%Note that, while \eqref{nw2} is a stationary process, the methodology is more general allowing for more general processes \citep[see, e.g.,][]{larsen2017}. 
This process is reminiscent of a dynamic extension of point-mass mixture priors that exhibit exact zeros \citep{mitchell_beauchamp}. 
%The probability of  thresholding  is a deterministic function of the most recent latent value  value $\beta_{t-1j}$ and the selection threshold $d_j$, which is subject to estimation. 
%This process can be regarded as a dynamic extension of point-mass mixture priors that exhibit exact zeros \citep{mitchell_beauchamp}.  
%The latent threshold approach has, so far, relied on rather laborious MCMC implementations, for which optimization variants are far from obvious. In this work, we propose new dynamic {\sl continuous} spike-and-slab alternatives which have appeal not only from a methodological but also a practical (optimization) viewpoint. 
 Other related works include shrinkage approaches towards static coefficients in time-varying models  \citep{fruhwirth_wagner, bitto,lopes_mcc_tsay}. We approach the dynamic sparsity problem through the {lens of}  Bayesian variable selection  and develop it  further for varying coefficient models. Namely, we assume the traditional spike-and-slab setup  by assigning  each regression coefficient $\beta_{tj}$ a mixture prior underpinned by a binary latent indicator $\gamma_{tj}$, which flags {the coefficient} as being either active or inert.  While static variable selection with spike-and-slab priors 
  has received a considerable attention (\cite{carlin_chib, Clyde, GM93,GM97,mitchell_beauchamp,RG14}, to name a few), dynamic incarnations are yet to be fully explored \citep{george_sun,fruhwirth_wagner, nakajima_west,groen}. 
  To narrow this gap, this work proposes several new dynamic extensions of popular  spike-and-slab priors.

The main thrust of this work is to introduce {\sl Dynamic Spike-and-Slab ($DSS$) priors}, a new class of time series priors, which induce either smoothness or shrinkage towards zero. 
These processes are formed as mixtures of two (stationary) time series: one for the active  and {another} for the negligible coefficients. The $DSS$ priors pertain closely to the broader framework of mixture autoregressive ($MAR$) processes  with a given lag, where the mixing weights are allowed to depend on time. Despite the reported success of $MAR$ processes (and variants thereof)  for modeling non-linear time series \citep{wong_li, wong_li2, kalli, wood_kohn}, their potential as dynamic sparsity inducing priors has been unexplored. Here, we harness this potential within a dynamic variable selection framework. 
%The process is formed as a mixture of  two stationary processes, one for the active coefficients (slab) and one for the negligible (spike). The proposal pertains to existing mixture autoregressive (MAR) processes with time-varying mixing weights, originally proposed for modeling non-linear time series. %To fix ideas, we suppress the subscript $j$ from the notation. This simpler case is a fundamental building block
One feature of stationary variants of our $DSS$ priors, {that}  sets it apart from the latent threshold model, is that it yields  benchmark continuous spike-and-slab priors (such as the Spike-and-Slab LASSO of \cite{rockova15}) as its {\sl marginal stationary distribution}. This property guarantees marginal stability  in the selection/shrinkage dynamics and probabilistic coherence. {Non-stationary variants with a  random walk slab process are also possible within our framework.}

%Another key distinguishing feature our process is the way it is deployed. Rather than implementing a full-blown MCMC apparatus,  we regard the ASSP priors as MAP  smoothing media
%within a penalized likelihood framework.
%The other important difference is that the ASSP process achieves sparsity through sparse posterior modes rather than latent  thresholding. 
%The second major point of contrast between the latent threshold approach and our ASSP priors is the way these priors are deployed.
{For efficient posterior sampling under the Gaussian spike-and-slab process, we develop {\em Dynamic SSVS}, a new extension of SSVS of \cite{GM93} for time series regression with closed-form forward-smoothing and backward-sampling updates \citep{Schnatter1994}.
To scale our method to big data settings, we then develop a MAP smoother  called {\em Dynamic EMVS}, a time series incarnation of EMVS originally conceived for static regression
 \citep{RG14}. 
Dynamic EMVS is very fast and uses closed-form updates for both the mean and variance parameters. }
We also consider Laplace spike distributions and turn these mixture processes  into dynamic penalty constructs. We formalize the notion of prospective and retrospective shrinkage through doubly adaptive shrinkage terms that pull together past, current, and future information. We introduce asymmetric dynamic thresholding rules --extensions of existing rules for static symmetric regularizers \citep{fanli,antoniadis_fan}-- to characterize the behavior of joint posterior modes for MAP smoothing.  For calculations under the Laplace spike, we implement a one-step-late EM algorithm of \citep{green_OSL}, that capitalizes on fast closed-form one-site updates. 
Our dynamic penalties can be regarded as natural extensions of the spike-and-slab penalty functions  introduced by \cite{rockova15} and further developed by \cite{SSL}. 

{We demonstrate the effectiveness of our introduced $DSS$ priors with a thorough simulation study and a topical macroeconomic application.
Both studies highlight the comparative improvements --in terms of inference, forecasting, and computational time-- of $DSS$ priors over conventional and recent methods in the literature.
In particular, the macroeconomic application, using a large number of economic indicators to forecast inflation and infer on underlying economic structures, serves as a motivating example as to why dynamic sparsity is  effective, and even necessary, in these contexts.
}
%The major contribution of this paper is the proposal of new dynamic sparsity priors rather than a cure-all comprehensive framework for capturing complex dynamic phenomena. 
% While our framework can be naturally extended to  time-varying variances $\sigma^2_t$,  we shall focus on the simpler case of fixed variances $\sigma^2_t=\sigma^2$. 

%We illustrate our approach on portfolio allocation models that attempt to balance optimality and simplicity (along the lines of the penalized utility approach of \cite{hahn_carvalho} and \cite{puelz}), where our new regularization approach is shown to yield stable and simple portfolios.

 The paper is structured as follows: Section \ref{sec:assp} and Section \ref{sec:other} introduce the $DSS$ processes and their variants. 
 Sections \ref{sec:dynamic_SSVS} and \ref{sec:dynamic_EMVS} introduce Dynamic SSVS and EMVS, respectively. Section \ref{sec:pen} develops the penalized likelihood perspective, introducing the prospective and retrospective shrinkage terms. Section \ref{sec:mode} develops the one-step-late EM algorithm for Spike-and-Slab Fused LASSO MAP smoothing. Section \ref{sec:simul} illustrates the MAP smoothing deployment of $DSS$ on simulated examples and Section \ref{sec:high}  on a macroeconomic dataset. Section \ref{sec:dis} concludes with a discussion.
 %Section \ref{sec:analysis} applies our framework to sparse portfolio models and Section \ref{sec:dis} wraps up with discussion.
 
  %When $b_{tj}$ drops within the range $[-d_j,d_j]$, the coefficient $\beta_{tj}$ is thresholded to zero.
%The inference rests on the posterior sampling of $\{b_{tj}\}_{t=1}^T$ as well as ${d}_j$. The formulation \eqref{nw1} and \eqref{nw2} can be viewed as a particular point-mass spike-and-slab formulation, where the latent indicators $\gamma_{tj}$ are deterministic functions of $b_{t-1j}$ and $d_j$. Despite amenable to posterior simulation, the point-mass mixture priors are less friendly for optimization, a perspective we want to develop here.

%The conditional prior $\pi(\beta_{tj}|d_j,b_{t-1j})$ implied by   \eqref{nw1} and \eqref{nw2} erases posterior mass in a narrow box around zero. The goals is to maintain only coefficients that are worthwhile. A similar effect can be achieved with continuous spike-and-slab priors that operate with the notion of practical significance. While the latent threshold approach achieves shrinkage to zero through a deletion of posterior mass, continuous spike-and-slab priors exert automatic thresholding through sparse posterior modes. Here, we exploit this property by proposing continuous alternatives to \eqref{nw1} and \eqref{nw2} that provide a viable route towards fast optimization.
\section{Dynamic Spike-and-Slab  Priors}\label{sec:assp}
%The spike-and-slab methodology for variable selection has been widely used in statistical data analysis since its inception more than two decades ago. 
%Extensions to varying coefficient models have not yet received much attention.
In this section, we introduce the class of Dynamic Spike-and-Slab ($DSS$) priors that constitute a coherent extension of benchmark spike-and-slab priors  for dynamic selection/shrinkage.
%Under  Gaussian linear state-space systems, . Standard Gaussian transition models assume uni-modal transition kernels, yielding marginally and conditionally Gaussian. 
%In this section, we introduce the class of {\sl Autoregressive Spike-and-Slab Process (ASSP)} priors,   dynamic elaborations of spike-and-slab priors for time-domain linear models. 
%Rather than treating all regression coefficients as one group with a single Gaussian autoregressive process, we treat them differentially with two different processes, one for the active and one for the negligible. 
We will assume that the $p$ time series $\{\beta_{tj}\}_{t=1}^T$ (for $j=1,\dots, p$) in \eqref{model} follow independent and identical $DSS$ priors and thereby we suppress the subscript $j$ (for notational simplicity).

We start with a conditional specification of the $DSS$ prior. Given a binary indicator $\gamma_t\in\{0,1\}$, which encodes the spike/slab membership at time $t$, and a lagged value $\beta_{t-1}$, we assume that $\beta_t$  arises from a mixture of the form
\begin{equation}\label{betas}
\pi(\beta_{t}\C\gamma_{t},\beta_{t-1})=(1-\gamma_{t})\psi_0(\beta_{t}\C\lambda_0)+\gamma_{t}\psi_1\left(\beta_{t}\,|\,\mu_t,\lambda_1\right),
\end{equation}
where
\begin{equation}\label{mu}
\mu_t=\phi_{0}+\phi_{1}(\beta_{t-1}-\phi_0)\quad\text{with}\quad |\phi_1|<1
\end{equation}
and
\begin{equation}\label{gammas}
\P(\gamma_{t}=1\C\beta_{t-1})=\theta_t.
\end{equation}
For Bayesian variable selection, it has been customary to specify a zero-mean spike density $\psi_0(\beta\C\lambda_0)$, such that it concentrates at (or in a narrow vicinity of) zero.
% One purposeful choice is the  Laplace density $\psi_0(\beta\C\lambda_0)=\frac{\lambda_0}{2}\e^{-|\beta|\lambda_0}$ (with a relatively large penalty parameter $\lambda_0>0$) due to its ability to threshold  via sparse posterior modes \citep{rockova15}.  
Regarding the slab distribution $\psi_1(\beta_{t}\C\mu_t,\lambda_1)$, we require that it be moderately peaked around its mean $\mu_t$, where the amount of spread is regulated by a concentration parameter $\lambda_1>0$.  
%While our framework encompasses a broad range of possible choices of $\psi_0(\cdot)$ and $\psi_1(\cdot)$ (as  elaborated on in Section \ref{sec:other}),
%we will focus primarily on the Gaussian slab $\psi_1(\beta_{t}\C\mu_t,\lambda_1)$ (with  mean $\mu_t$ and variance $\lambda_1$) due to its ability to smooth over past/future values. %While our framework allows for a broader range of popular continuous spike and slab priors (George and McCulloch (1993), Ishwaran and Rao (2005)), Rockova and George (2016), we will be mostly focusing on the Laplace-spike  and Gaussian-slab mixture for reasons mentioned above. We discuss extensions with other spike and slab densities in Section  
%We should also like to point out that our framework can be naturally extended to higher-order autoregression in \eqref{mu}, where $\mu_t$ is allowed to depend on older values than just the previous one. 
%Our framework can be extended to higher-order autoregressive polynomials where $\mu_t$ may also depend on values older than $\beta_{t-1}$. However, here we focus on the first-order autoregression due to its {practicality and} ubiquity  in practice \citep{fused_lasso,WestHarrison1997book2,Prado2010}. 
%The choice of the Laplace spike is due its ability to threshold. The choice of the slab is due to its ability to smooth over previous values.
%These processes yield  conditionally Gaussian dynamic linear models (DLM), whose posteriors can be explored efficiently using sampling techniques (Fruhwirth-Schnatter (1994), Carter and Kohn (1994)).
The conditional $DSS$ prior formulation \eqref{betas} generalizes existing {\sl continuous} spike-and-slab priors \citep{GM93,ishrawan_rao_annals,rockova15} in two important ways. First, rather than centering the slab around zero, the $DSS$ prior anchors it around an actual model for the {\sl time-varying} mean  $\mu_t$. The non-central mean is defined as an autoregressive lag polynomial of the first order with  hyper-parameters $(\phi_0,\phi_1)$. While our framework can be extended to higher-order autoregressive polynomials where $\mu_t$ may also depend on values older than $\beta_{t-1}$, we outline our method for the first-order autoregression with $\phi_0=0$ due to its ubiquity  in practice \citep{fused_lasso,WestHarrison1997book2,Prado2010}. {The autoregressive parameter $\phi_1$ will be treated as unknown and estimated.}
% Throughout the manuscript, however, we remark on how to suitably modify our framework for higher orders. 
%Although estimable (subject to stationary restrictions), $(\phi_0,\phi_1)$ will be treated as fixed. Assuming a fixed $\phi_1$ is not too far from the common practice in the Bayesian literature \citep[][to name a few]{omori2007stochastic,nakajima_west} which consists of imposing a tight prior around, but below, 1 for $\phi_1$, to ensure stable, stationary estimation.  We discuss extensions of our framework for random $(\phi_0,\phi_1)$ in Section \ref{sec:dis}.
%Assuming a fixed $\phi_1$ is thus not too far from imposing a strong informative prior. 

It is illuminating to view the conditional prior \eqref{betas}  as a ``multiple shrinkage" prior \citep{george86a,george86b} with two shrinkage targets: (1) zero (for the gravitational pull of the spike), and (2) $\mu_{t}$ (for the gravitational pull of  the slab). It is also worthwhile to emphasize that the spike distribution $\psi_0(\beta_{t}\C\lambda_0)$ {\sl does not depend on $\beta_{t-1}$}, only the slab does.  The $DSS$ formulation thus  induces separation of regression coefficients  into two groups, where only the {\sl active} ones are assumed to walk on an autoregressive path.  
%The $DSS$ prior is seen as a variant of the $GMAR$ process of Kalliovirta et al. (2015). 

 The second important  generalization is implicitly hidden in the  hierarchical formulation of the mixing weights $\theta_t$ in \eqref{gammas}, which casts them as a smoothly evolving process (as will be seen in Section \ref{sec:weights}  below).
%As will be seen in Section \ref{sec:weights}, these weights} are also time-varying and depend deterministically on the lagged values of the $\beta_{t}$ process. 
\iffalse
Before turning to this formulation, however, it is worthwhile to point out the time-varying conditional moments of $DSS$,   given the past values $\mF_{t-1}$ after margining over $\gamma_t$:
$
\E[\beta_{t}\C\beta_{t-1}]=\theta_t[\phi_0+\sum_{k=1}^q\phi_k\beta_{t-k}]
$
and
$\var[\beta_{t}\C\beta_{t-1}]=(1-\theta_t)\frac{2}{\lambda_0^2}+\theta_t\lambda_1+(1-\theta_t)\theta_t\mu_t^2. $
\fi
%Hence, not only the mean but also the conditional variance  changes over time, {\color{red}allowing the $DSS$ to capture conditional heteroscedasticity of the} regression coefficients. {\color{red}(I'm not sure why this is useful.  The AR process is not heteroscedastic so this is just about marginal variation created by the spike.)}
Before turning to this formulation, we discuss several special cases of $DSS$ priors.

\subsection{Spike and Slab Pairings}\label{sec:other}
One possible choice of the spike distribution is the  Laplace density $\psi_0(\beta\C\lambda_0)=\frac{\lambda_0}{2}\e^{-|\beta|\lambda_0}$ (with a relatively large penalty parameter $\lambda_0>0$) due to its ability to threshold  via sparse posterior modes, as will be elaborated on in Section \ref{sec:mode}. {Under the Laplace spike distribution (i.e. conditionally on $\gamma_t=0$) the series $\{\beta_t\}_{t=1}^T$ is  stationary, iid with a marginal density $\psi_0(\beta\C\lambda_0)$.}  
Another natural choice, a Gaussian spike, would impose no new computational challenges due to its conditional conjugacy. However,  additional thresholding would be required to obtain a sparse representation. %We will primarily focus on the Laplace spike due to its automatic thresholding property \citep{rockova15}.

Regarding the slab distribution, we will focus primarily on the Gaussian slab $\psi_1(\beta_{t}\C\mu_t,\lambda_1)$ (with  mean $\mu_t$ and variance $\lambda_1$) due to its ability to smooth over past/future values. %Note that the conditional form \eqref{betas}-\eqref{gammas} is a mixture of two processes. 
Under the Gaussian slab distribution, $\{\beta_t\}_{t=1}^T$  follow a stationary Gaussian $AR(1)$ process 
%defined through the following auxiliary equations:
\begin{equation}\label{slab_process}
\beta_{t}=\phi_{0}+\phi_{1}( \beta_{t-1}-\phi_0)+e_{t},\quad |\phi_1|<1,\quad e_{t}\iid \mathcal{N}\left(0,\lambda_1\right),
\end{equation}
whose stationary distribution is characterized by  univariate marginals
\begin{equation}\label{stat_marginal}
\psi_1^{ST}(\beta_{t}\C \lambda_1,\phi_0,\phi_1)\equiv\psi_1\left(\beta_{t}\,\Big|\,\phi_0,\frac{\lambda_1}{1-\phi_1^2}\right);
\end{equation}
%where $\psi_1\left(\beta\C\mu,\lambda\right)$ is 
a Gaussian density with mean $\phi_0$ and variance $\frac{\lambda_1}{1-\phi_1^2}$.
The availability of this tractable stationary distribution \eqref{stat_marginal} is another appeal of the conditional Gaussian slab distribution. 
%However, the $DSS$ construction is not confined to the Gaussian slab (Laplace spike). We elaborate on alternative choices in Section \ref{sec:other}.
%However, the $DSS$ construction is not confined to the Gaussian slab (Laplace spike). We elaborate on alternative choices in Section \ref{sec:other}.

{Rather than shrinking  to the {\em vicinity} of the past value, one might like to entertain the possibility of shrinking {\em exactly} to the past value \citep{fused_lasso} to obtain piece-wise constant reconstructions. Such a property  would be appreciated, for instance, in dynamic  sparse  portfolio allocation models to mitigate transaction costs associated with negligible shifts in the portfolio weights \citep{kaoru, brodie, jagan, puelz}. 
This extension has also desirable consequences for $h$-step ahead  forecasting, where $\beta_{j t+h}$ would be prevented from decaying (albeit slowly) over time.
One way of attaining the desired effect would be replacing the Gaussian slab  $\psi_1(\cdot)$ in \eqref{betas} with a Laplace distribution centered at $\mu_t$, i.e.
\begin{equation}\label{conditional_laplace}
\psi_1(\beta_t\C\mu_t,\lambda_1)=\frac{\lambda_1}{2}\e^{-|\beta_t-\mu_t|\lambda_1}
\end{equation}
and by considering $\phi_0=0$ and $\phi_1=1$.
While both the Gaussian and Laplace slab will lead to a conditional posterior mean which shrinks towards the past value, the conditional posterior mode will shrink exactly to the past value for the Laplace (and not the Gaussian). 
This relates  the non-stationary extensions discussed further in Remark \ref{remark:nonstat}.  A similar effect could be achieved with coefficient specific-autoregressive parameters by allowing for $\phi_{j0}\neq 0$ and $\phi_{j1}=0$ for $1\leq j\leq p$ \citep{lopes_mcc_tsay}. 
}

The stationary Laplace conditional construction \eqref{conditional_laplace} (with $|\phi_1|<1$), however, does not  imply the Laplace distribution marginally. The univariate marginals are defined through the characteristic function given in (2.7) of \cite{andel}. The lack of availability of the marginal density in a simple form thwarts the specification of transition weights in our $DSS$ framework. There are, however, avenues for constructing an autoregressive process with Laplace marginals, e.g., through the normal-gamma-autoregressive  ($NGAR$) process by \cite{kalli_griffin}. We define  the following Laplace autoregressive ($LAR$) process as a special case.

\begin{definition}
We define the Laplace autoregressive ($LAR$) process by
$$
\beta_t=\sqrt{\frac{\psi_t}{\psi_{t-1}}}\phi_1\beta_{t-1}+\eta_t,\quad \eta_t\sim\mathcal{N}\left(0,(1-\phi_1^2)\psi_t\right),
$$
where $\{\psi_t\}_{t=1}^T$ follow an exponential autoregressive process specified through $\psi_t\C\kappa_{t-1}\sim \mathrm{Gamma}(1+\kappa_{t-1},\lambda_1^2/[2(1-\rho)])$
 and $\kappa_{t-1}\C\psi_{t-1}\sim\mathrm{Poisson}\left(\frac{\rho}{2(1-\rho)}\lambda_1^2\psi_{t-1}\right)$ with a marginal distribution $Exp(\lambda_1^2/2)$.
 \end{definition}
The $LAR$ process exploits the scale-normal-mixture representation of the Laplace distribution, yielding  Laplace marginals  $\beta_t\sim \widetilde{\psi}_{ST}(\beta_t\C\lambda_1)\equiv Laplace(\lambda_1)$.  This coherence property can be leveraged within our $DSS$ framework as follows. If we replace the slab Gaussian $AR(1)$ process in \eqref{betas} with the $LAR$ process and 
deploy $\widetilde{\psi}_{ST}(\beta_t\C\lambda_1)$ instead of ${\psi}_{ST}(\beta_t\C\lambda_1)$ in  \eqref{weights}, we obtain a Laplace $DSS$ variant with the {\sl Spike-and-Slab LASSO} prior of \cite{rockova15} as its marginal distribution (according to Theorem \ref{thm1}).

It is worth pointing out an alternative autoregressive construction with Laplace marginals proposed  by \cite{andel}, {where  the following $AR(1)$ scheme is considered.}
\begin{equation}\label{andel}
\beta_t=\begin{cases}
&\phi_1\beta_{t-1} \quad \quad\quad\text{with probability}\quad \phi_1^2,\\
&\phi_1\beta_{t-1}+\eta_t\quad \text{with probability}\quad 1-\phi_1^2, \quad\text{where}\quad\eta_t\sim Laplace(\lambda_1).
\end{cases}
\end{equation}
The innovations in \eqref{andel} come from a mixture of a point mass at zero, providing an opportunity to settle at the previous value, and a Laplace distribution. Again, {by} deploying this process in  the slab, we  obtain the {\sl Spike-and-Slab LASSO} marginal distribution \citep{rockova15}.  While MCMC implementations can be obtained for the dynamic Spike-and-Slab LASSO method (e.g. embedding the sampler of \cite{kalli_griffin} within our MCMC approach outlined in Section \ref{sec:dynamic_SSVS}),
the  slab extensions with Laplace marginals are  more challenging for optimization. 
%One can, however, specify the mixing weights without the stationary distribution (as we now point out in Remark \ref{remark:nonstat}) and optimize under the  Laplace versions as well.
Throughout the rest of the paper, we thereby focus primarily on the Gaussian $AR(1)$ slab process. We will, however, consider both Gaussian and Laplace spike distributions.

%Before proceeding, let us note that the spike distribution $\psi_0(\beta\C\lambda_0)$ can be replaced by any (continuous) density without disturbing the validity of Theorem \ref{thm1}. 

\subsection{Evolving Inclusion Probabilities}\label{sec:weights}

%A very appealing feature of the sequence of slab probabilities $\{\theta_t\}_{t=1}^T$ defined by \eqref{weights}, is how it evolves smoothly over time, allowing for changes  in variable importance as time progresses and, at the same time, avoiding erratic regime switching. 
   %Note that the inclusion is required to depend on the previous value of the process, not  on the previous inclusion probability! We discuss these discrepancies below.

A very appealing feature of $DSS$ priors that makes them suitable for dynamic subset selection is the opportunity they afford  for obtaining ``smooth" spike/slab memberships.
%So far, we  focused on the evolution of regression coefficients $\beta_t$ in relation to the past value $\beta_{t-1}$ and the binary indicator $\gamma_t$. 
%In this section we detail how the subset selection process induced by the latent binary indicators $\{\gamma_{t}\}$.
Recall that  the binary indicators in \eqref{gammas} determine which of the spike or slab regimes is switched on at time $t$, where $\P(\gamma_t=1\C\beta_{t-1})=\theta_t$. 
   It is desirable that the sequence of slab probabilities $\{\theta_t\}_{t=1}^T$ evolves smoothly over time, allowing for changes  in variable importance as time progresses and, at the same time, avoiding erratic regime switching. 
%A viable strategy might be to treat $\{\theta_t\}_{t=1}^T$ as random with an autoregressive process that relates $\theta_t$ to the previous value $\theta_{t-1}$ with a beta autoregressive process  \citep{casarin,mckensie}. 
Because the series $\{\theta_t\}_{t=1}^T$ is a key driver of the sparsity pattern, it is important that  it be (marginally) stable and that it reflects all relevant information, including not only the previous value $\theta_{t-1}$, but {\sl also} the previous value $\beta_{t-1}$. Many possible constructions of $\theta_t$ could be considered. We turn to the implied stationary distribution as a guide for a principled construction of  $\theta_t$.
%The beta autoregressive  constructions mentioned above can be  modified to include $\beta_{t-1}$.  However, characterizing stationarity becomes a far more delicate task due to the interplay between the two autoregressive processes (on $\{\beta_t\}_{t=1}^T$ and $\{\theta_t\}_{t=1}^T$). 

% It would further be desirable to obtain a process $\{\beta_t\}_{t=1}^T$ with a completely characterized marginal stationary distribution.  This feature is obtained with our deterministic construction of mixing weights in \eqref{weights}. %Similar deterministic definition was \eqref{nw1} and \eqref{nw2}.

%   Although the $\{\beta_{t}\}_{t=1}^T$ process will be stationary under each of the spike and slab distributions separately, it is not immediately obvious that it will be stationary under the spike-and-slab mixture where the $\beta_t$'s can transition between these distributions and where $\theta_t$ depends on $\beta_{t-1}$.  However, with a suitable formulation for the $\{\theta_t\}_{t=1}^T$ sequence, the stability and coherence of the $DSS$ can be maintained.  Under this formulation (introduced below), the $\{\beta_{t}\}_{t=1}^T$ process will not only be stationary, but will have  spike-and-slab marginals.  Such a $\theta_t$ sequence is obtained with a deterministic  transition function of the lagged $\beta_t$ values $\theta_t=\theta(\beta_{t-1})$. 
   %(similarly as in \eqref{nw1} and \eqref{nw2} of \cite{nakajima_west}).  
 For our formulation, we introduce a {\sl marginal} importance weight $0<\Theta<1$, a scalar parameter which controls the overall balance between the spike and the slab distributions. Given $(\Theta,\lambda_0,\lambda_1,\phi_0,\phi_1)$, the conditional inclusion probability $\theta_t$ (or a transition function $\theta(\beta_{t-1})$)  is  defined as
\begin{equation}\label{weights}
\theta_t\equiv\theta(\beta_{t-1})=\frac{\Theta\psi_1^{ST}\left(\beta_{t-1}|\lambda_1,\phi_0,\phi_1\right)}{\Theta\psi_1^{ST}\left(\beta_{t-1}|\lambda_1,\phi_0,\phi_1\right)+(1-\Theta)\psi_0\left(\beta_{t-1}|\lambda_0\right)}.
\end{equation}

\iffalse
Before turning to stationarity properties of the full $DSS$ priors, we pause  to appreciate the probabilistic meaning of \eqref{weights}. 
\fi
The conditional mixing weight $\theta_t$ can be interpreted as the posterior probability of classifying the past coefficient $\beta_{t-1}$ as arriving from the {\sl stationary} slab {distribution} as opposed to the (stationary) spike {distribution}. This interpretation {reveals} how the weights $\{\theta_t\}_{t=1}^T$ proliferate  parsimony throughout the process $\{\beta_t\}_{t=1}^T$. 
Suppose that the past value $|\beta_{t-1}|$ was large, then $\theta(\beta_{t-1})$ will be close to one, signaling that the {\sl current} observation $\beta_t$  is more likely to be in the slab.  The contrary occurs when  $|\beta_{t-1}|$ is small, where $\beta_t$  will be discouraged from the slab because the inclusion weight $\theta(\beta_{t-1})$ will be small (close to zero). {Let us also note that the weights in \eqref{weights}  are {\sl different from} the conditional  probabilities for classifying  $\beta_{t-1}$ as arising from the {\sl conditional slab} in \eqref{betas}.} These weights will be introduced later in Section \ref{sec:pen}.

\smallskip

Now that we have elaborated on all the layers of the hierarchical model, we are ready to formally define the Dynamic Spike-and-Slab Process.

\begin{definition}
Equations \eqref{betas}, \eqref{mu}, \eqref{gammas} and \eqref{weights} define a
Dynamic Spike-and-Slab  Process (DSS) 
with parameters
$(\Theta,\lambda_0,\lambda_1,\phi_{0},\phi_{1})$.
We will write
$$
\{\beta_{t}\}_{t=1}^T\sim DSS(\Theta,\lambda_0,\lambda_1,\phi_{0},\phi_{1}).
$$
\end{definition}
The $DSS$ process relates to the Gaussian mixture of autoregressive ($GMAR$) process of \cite{kalli}, which was conceived as a model for time series data with regime switches. Here, we deploy it as a prior on time-varying regression coefficients within the spike-and-slab framework, allowing for {distributions} other than Gaussian. The $DSS$, being an instance/elaboration of the $GMAR$ process,  inherits elegant marginal characterizations (as will be seen below)

The $DSS$ construction has a strong conceptual appeal  in the sense that its marginal probabilistic structure is fully known. This property is  rarely available with conditionally defined non-Gaussian time series models, where not much is known about the stationary distribution  beyond just the mere fact that it exists. The $DSS$ process, on the other hand, guarantees well behaved stable marginals that can be described through benchmark spike-and-slab priors. The marginal distribution can be used as a prior for the initial vector at time $t=0$, which is typically estimated with the remaining coefficients.  %As such, this process is a natural extension of the spike-and-slab priors which have been ubiquous throughout the variable selection literature. 
%  from alternative formulations (using e.g. autoregressive processes on $\theta_t$ described in Section \ref{sec:weights})
%Being inherently a mixture of stationary processes,  the $DSS$ process ought to be stationary. In the context of $MAR$ models, \cite{wong_li} describe stationarity restrictions on the autoregressive polynomial parameters as well as mixing weights that are {\sl not time varying}.  If the mixing weights $\theta_t$ were fixed, the stationarity would be then inherited from the slab Gaussian $AR(1)$ process when $|\phi_1|<1$. Going further, \cite{kalli} characterize the stationary distribution of the $GMAR$ process with time varying weights, which aligns closely with the $DSS$ process. 
The following theorem  is an elaboration of  Theorem 1 of \cite{kalli}.

\begin{theorem}\label{thm1}
Assume $\{\beta_{t}\}_{t=1}^T\sim DSS(\Theta,\lambda_0,\lambda_1,\phi_0,\phi_1)$ with $|\phi_1|<1$. Then $\{\beta_{t}\}_{t=1}^T$ has a stationary distribution characterized by the following univariate marginal distributions:
\begin{equation}\label{stat_dist}
\pi^{ST}(\beta|\Theta,\lambda_0,\lambda_1,\phi_0,\phi_1)=\Theta\,\psi_1^{ST}(\beta\C \lambda_1,\phi_0,\phi_1)+(1-\Theta)\psi_0\left(\beta\C\lambda_0\right),
\end{equation}
where $\psi_1^{ST}(\beta\C \lambda_1,\phi_0,\phi_1)$ is the stationary slab distribution \eqref{stat_marginal}.
\end{theorem}
\begin{proof}
We assume an initial condition $\beta_{t=0}\sim \pi^{ST}(\beta_0|\Theta,\lambda_0,\lambda_1,\phi_0,\phi_1)$. Recall that the conditional density of $\beta_{1}$ given $\beta_{0}$ can be written as
\begin{align}
\pi(\beta_{1}\C\beta_0)&=(1-\theta_{1})\psi_0(\beta_{1}\C\lambda_0)+\theta_{1}\psi_1(\beta_{1}\C\mu_1,\lambda_1).
%&=\frac{\Theta_j\phi\left({\beta_{t-1j}};0,{v_1}{(1-\phi_{1j}^2)}\right)\phi\left(\beta;\frac{\phi_{0j}}{1-\phi_{1j}},{v_1}{(1-\phi_{1j}^2)}\right)}{\Theta_j\phi\left({\beta_{t-1j}};0,{v_1}{(1-\phi_{1j}^2)}\right)+(1-\Theta_j)\phi\left({\beta_{t-1j}};0,v_0\right)}\phi\left({\beta_{t-1j}};0,{v_1}{(1-\phi_{1j}^2)}\right)\phi(\beta_{tj};\phi_{0j}+\phi_{1j}\beta_{t-1j},v_1)+\frac{(1-\Theta_j)}{\Theta_j\phi\left({\beta_{t-1j}};0,{v_1}{(1-\phi_{1j}^2)}\right)+(1-\Theta_j)\phi\left({\beta_{t-1j}};0,v_0\right)}\phi\left({\beta_{t-1j}};0,v_0\right)\phi(\beta_{tj};0,v_0)
\end{align}
From the definition of  $\theta_1$ in \eqref{weights}, we can write the joint distribution  as
$$
\pi(\beta_1,\beta_0)=\Theta\,\psi^{ST}_1\left(\beta_0\C\lambda_1,\phi_0,\phi_1\right)\psi_1(\beta_1\C\mu_{1},\lambda_1)+{(1-\Theta)\psi_0\left(\beta_{0}\C\lambda_0\right)\psi_0(\beta_{1}\C\lambda_0)}.
$$
Integrating $\pi(\beta_1,\beta_0)$ with respect to $\beta_0$, we obtain
\begin{align*}
\pi(\beta_1)=\int \pi(\beta_1,\beta_0)d\beta_0=&\Theta \left[\int_{\beta_0}\psi_1\left(\beta_1\C\mu_{1},\lambda_1\right)\psi_1\left(\beta_{0}\,\Big|\, \phi_{0},\frac{\lambda_1}{1-\phi_{1}^2}\right)\d\beta_0\right]+(1-\Theta)\psi_0(\beta_1\C\lambda_0) \\
=&\Theta\, \psi^{ST}_1(\beta_1\C\lambda_1,\phi_0,\phi_1)+(1-\Theta)\psi_0(\beta_1\C\lambda_0).\quad\quad \qedhere
\end{align*}
\end{proof}
 Theorem \ref{thm1} describes the very elegant property of $DSS$ that the univariate marginals of this mixture process are  $\Theta$-weighted mixtures of marginals.
 It also suggests a more general recipe for mixing multiple stationary processes  through the construction of mixing weights \eqref{weights}. 
 % In the next section, we provide examples of such mixing with Laplace  slab distributions.
 
 \begin{remark}
For autoregressive polynomials of higher order $h>1$, the transition weights $\theta_t$ could be defined  in terms of a multivariate stationary distribution evaluated at the last $h$ values of the process, not only the last one. The marginals of such process could be then characterized in terms of a mixture of multivariate Gaussian distributions (Theorem 1 of  \cite{kalli}).
\end{remark}

%We pursue a  different strategy here. Rather than regarding $\theta_t$ as a realization from a random process, we treat it as a deterministic   function of the previous value $\beta_{t-1}$. 
{

It is tempting to regard $\Theta$ as the marginal proportion of nonzero coefficients. Such an interpretation is a bit misleading since
the sparsity levels are ultimately determined by the $\theta_t$ sequence, which is influenced by the component stationary distributions  $\psi_0(\cdot)$ and $\psi_1^{ST}(\cdot)$, in particular by the amount of their overlap around zero. 
With {\sl continuous} spike-and-slab mixtures considered here,  more caution is needed for calibration \citep{rockova15}. {This issue will be revisited  in Section \ref{sec:pen}}. One can nevertheless regard $\Theta$ as a global sparsity parameter, as we now show.

Unlike with point-mass spike and slab priors \citep{mitchell_beauchamp}, which assign prior mass directly on sparse vectors, our prior is continuous  where exact sparsity can be achieved through posterior modes (using the Laplace spike) or through thresholding. As with other continuous priors \citep{dirichlet_laplace,rockova15} one can quantify the ``effective dimensionality" defined as the number of coefficients which are large enough to be non-negligible. 
The availability of the  stationary distribution is helpful for understanding  the {\em marginal} prior effective dimensionality at each time $t$. For example, with the Laplace slab and the Gaussian spike \eqref{spike_gaussian} the (smaller) intersection points $\pm\delta$  between the stationary spike and slab densities   satisfy (choosing $\phi_0=0$ for simplicity)
$$
\delta=\frac{\lambda_1\lambda_0}{1-\phi_1^2}-\sqrt{\left(\frac{\lambda_1\lambda_0}{1-\phi_1^2}\right)^2-\frac{2\lambda_1}{1-\phi_1^2}\log\left(\frac{1-\Theta}{\Theta}\sqrt{\frac{2\pi\lambda_1}{1-\phi_1^2}}\frac{\lambda_0}{2}\right)}.
$$
%With the Gaussian slab \eqref{slab_process} and the Gaussian spike \eqref{spike_gaussian} 
%satisfy (choosing $\phi_0=0$ for simplicity and $\lambda_0<\lambda_1/(1-\phi_1^2)$)
%$$
%\delta=\sqrt{\frac{\lambda_1\lambda_0}{\lambda_1-(1-\phi_1^2)\lambda_0)}2\log\left(\frac{1-\Theta}{\Theta}\sqrt{\frac{\lambda_1}{\lambda_0(1-\phi_1)^2}}\right)}.
%$$
Defining  $\gamma(\beta)=\mathbb{I}(|\beta|>\delta)$ as the indicator for whether or not  the coefficient is important, one obtains
\begin{align}
\P\left[\gamma(\beta)=1\right]&=\Theta\left[2\left(1-\Phi\left(\delta,0,\frac{\lambda_1}{1-\phi_1^2}\right)\right)+\frac{1}{\lambda_0}\phi\left(\delta,0,\frac{\lambda_1}{1-\phi_1^2}\right)\right]
%&<\Theta\left[2+\frac{1}{\lambda_0\sqrt{2\pi\lambda_1/(1-\phi_1^2)}}\right],
\end{align}
where $\Phi(x,\mu,\sigma^2)$ and $\phi(x,\mu,\sigma^2)$ are the cumulative distribution function and the density of the Gaussian distribution with mean $\mu$ and variance $\sigma^2$.
Standard Gaussian tail bounds yield $\P\left[\gamma(\beta)=1\right]<\Theta\left[2+\frac{1}{\lambda_0\sqrt{2\pi\lambda_1/(1-\phi_1^2)}}\right]$, from which one  deduces that the parameter $\Theta$ takes the role of a {\em global sparsity parameter}.
Defining the {\em effective dimensionality} at time $t$ as $|\bg(\b_t)|=\sum_{j=1}^p\gamma(\beta_{tj})$,
it is desirable that  $|\bg(\b_t)|$ accumulates roughly around the true dimensionality $p_t=\sum_{j=1}^p\mathbb{I}[\beta_{tj}^0\neq0]$. 
Using the Chernoff bound for binomial random variables,  one obtains
$$
\P(|\bg(\b_t)|>C\,p_t)\leq\exp(-p_t\, C\log 2)\quad\text{when}\quad  \Theta\leq C_1 p_t/p\quad\text{for}\quad C_1>0
$$
and for any $C>2 C_1\e\left[2+\frac{1}{\lambda_0\sqrt{2\pi\lambda_1/(1-\phi_1^2)}}\right]$. 
This means that as long as the  parameter $\Theta$ does not overshoot the true sparsity proportion, the prior will concentrate on small subsets up to a constant multiple of the true model size. This property  will be ultimately reflected in the posterior.

}
{
\begin{remark}(Random Walk Extensions)\label{remark:nonstat}
The definition of $DSS$ transition weights in \eqref{weights} requires stationary distributions under the two spike and slab regimes.
It is possible to extend our framework to non-stationary  random walk slab process (obtained with $\phi_1=1$)   by modifying  transition weights $\{\theta_{t}\}_{t=1}^T$. 
 Because the series $\{\theta_{t}\}_{t=1}^T$ is a key driver of  sparsity, it is important that it be stable (not too erratic over time) and that it reflects all relevant information, including not only the previous value  $\theta_{t-1}$, but also the previous value $\beta_{t-1}$.  One viable strategy  would be to treat $\{\theta_t\}_{t=1}^T$ as random  and relate $\theta_t$ to the previous value $\theta_{t-1}$ via the conditional beta autoregressive process  \citep{casarin,casarin2} or a marginal beta autoregressive process  \citep{beta_marginal}. However,  the weights may be prone to  transitioning too often between the spike/slab states when treated as random.
For the random walk extensions, one can set $\theta_{t}$  equal to some deterministic sequence  (e.g. as in \cite{Nakajima2010})  or  to a fixed value $\theta_t=\Theta$ for $1\leq t\leq T$. 
\end{remark}}

{
\section{Dynamic SSVS}\label{sec:dynamic_SSVS}
In this section, we develop an MCMC algorithm   for dynamic spike-and-slab priors which can be regarded as the dynamic extension of SSVS of \cite{GM93}. The $DSS$ prior  specification  here departs slightly from our previous setup.
The Laplace spike distribution $\psi_0(\beta|\lambda_0)=\lambda_0/2\e^{-\lambda_0|\beta|}$ yields sparse posterior modes.
Since MCMC   ultimately reports the posterior mean (which is non-sparse even under the Laplace prior), we will  assume the Gaussian spike to capitalize on its direct conditional conjugacy for posterior updating. 
{In particular, we assume the following spike density  for $\lambda_0<<\lambda_1$
\begin{equation}\label{spike_gaussian}
	\psi_0(\beta\C \lambda_0) =\exp\{-\beta^2/(2\lambda_0)\}/\sqrt{2\pi\lambda_0}.
\end{equation}
This yields the following conditional Gaussian distribution
$$
\beta_t\,|\,\gamma_t,\beta_{t-1}\sim\mathcal{N}\left( \gamma_t\mu_t\,,\, \gamma_t\lambda_1+(1-\gamma_t)\lambda_0\right)
$$
and   transition weights $\theta_t$  in \eqref{weights} with the Gaussian stationary spike distribution $\psi_0^{ST}(\beta_{t-1}|\lambda_0)=\psi_0(\beta\C \lambda_0)$.
An extension to the Laplace spike is possible with an additional augmentation step, casting the Laplace distribution as a scale mixture of Gaussians with an exponential mixing distribution \citep{casella}.
The MCMC algorithm has a Gibbs structure, sampling iteratively from the conditional posteriors of the regression coefficients $\beta_{0:T}$,  latent indicators $\bg_{0:T}$ and variances $v_{0:T}$ (\citealt{Schnatter1994}; \citealt[][Sect 15.2]{WestHarrison1997book2}; \citealt[][Sect 4.5]{Prado2010}).

For the stationary $DSS$ prior, we assume that the autoregressive parameter $|\phi_1|<1$ is assigned the following beta prior (as in \citep{kim_etal})
\begin{equation}\label{beta_prior}
\pi(\phi_1)\propto \left(\frac{1+\phi_1}{2}\right)^{a0-1} \left(\frac{1-\phi_1}{2}\right)^{b0-1}\mathbb{I}(|\phi_1|<1)\quad\text{with $a_0=20$ and $b_0=1.5$},
\end{equation}
implying a prior mean of $2a_0/(a_0+b_0)-1=0.86$.  As was pointed out by  \cite{phillips},   a non-informative prior on $\phi_1$ might result in instability.
\cite{zellner_book} in Chapter 7 recommends a subjective beta prior peaked around one (see also  \cite{kastner_var,Nakajima2010}).
Alternatively, \cite{lopes_mcc_tsay} considered a grid of possible values for $\phi_1$ through a discretized  Gaussian prior distribution centered at one with a small variance. 
We will update $\phi_1$ with a Metropolis step, using a uniform proposal density on the interval $[0.8,1]$. While we assume $\phi_0=0$ throughout, one can update $\phi_0$  in a similar vein.

%\paragraph{Initialization:} 
%We initialize the calculation with suitably chosen inclusion indicators $\gamma_{tj}$ for $0\leq t\leq T$. Anticipating sparsity, one could sample a binary matrix (each entry independently) with a small success probability or use the EM output for initialization. 

\begin{table}[!t]
\small
\begin{center}
\scalebox{0.8}{
\begin{tabular}{|lll|}
\hline
\multicolumn{3}{|c|}{\cellcolor[HTML]{C0C0C0} \textbf{Algorithm:} \textit{MCMC algorithm for $DSS$ with a Gaussian spike}} \\ \hline\hline
 &     &Initialize $\gamma_{tj}$ and $v_0$ for $0\leq t\leq T$ and $1\leq j\leq p$ and choose $n_0,d_0$.\\
\multicolumn{3}{|c|}{\cellcolor[HTML]{C0C0C0}Sampling Regression Coefficients}                                                  \\
%& & For $j=1,\dots, p$  and $t=1,\dots, T$                                             \\
       &  \em Forward filtering  &  For $1\leq t\leq T$\\
       &       & Compute  $\bm a_t=\bm H_t+\bm \Gamma_t(\bm m_{t-1}-\bm H_t)$.                                                                \\
       &       & Compute  $\bm R_t=\bm\Gamma_t\bm C_{t-1}\bm \Gamma_t'+\bm W_t$.                                                                \\
       &       & Compute  $f_t=\bm x_t'\bm a_t$.                                                                \\
       &       & Compute  $q_t=\bm x_t'\bm R_t\bm x_t+v_t$ and $e_t=y_t-f_t$.                                                                \\
       &       & Compute  $\bm m_t=\bm a_t+  \bm A_te_t$ and $\bm C_t=\bm R_t-\bm A_t\bm A_t'q_t$ with $\bm A_t=\bm R_t\bm x_t/q_t$.                  \\
& \em Backward sampling & Simulate $\b_T\sim\mathcal{N}(\bm m_T,\bm C_T)$. \\
& & For $t=T-1,\dots, 0$ \\
& &  Compute $\bm a_T(t-T)=\bm \m_t +\bm B_t[\b_{t+1}- \bm a_{t+1} ]$. \\
& &  Compute $\bm R_T(t-T)=\bm C_t-\bm B_t\bm R_{t+1}\bm B_t'$, where $\bm B_t=\bm C_t\bG_{t+1}'\bm R_{t+1}^{-1}$. \\
& &  Simulate $\b_t\sim\mathcal{N}(\bm a_T(t-T),R_T(t-T))$.\\
\multicolumn{3}{|c|}{\cellcolor[HTML]{C0C0C0}Sampling Indicators}    \\
& & For $j=1,\dots, p$                                              \\
             &           & Compute $\theta_{tj}=\theta(\beta_{t-1j}) $ for $1\leq t\leq T$ from \eqref{weights}.                                             \\ 
& & Compute $p^\star_{tj}=p^\star_{tj}(\beta_{tj})$ for $1\leq t\leq T$  from \eqref{pstar}. \\
& & Compute $p^\star_{0j}=\theta(\beta_{0j})$   from \eqref{weights}. \\
& & Sample $\gamma_{tj}\sim\mathrm{Bernoulli}[p^\star_{tj}(\beta_{tj})]$ for $0\leq t\leq T$. \\
\multicolumn{3}{|c|}{\cellcolor[HTML]{C0C0C0}Sampling Precisions $\nu_t=1/v_t$}    \\
& & For $t=1,\dots, T$                                              \\
&  \em Forward filtering       & Compute $n_t=\delta n_{t-1}+1$ and $d_t=\delta d_{t-1}+r_t^2$, where $r_t=y_t-\bm x'_t\b_t$.\\
&  \em Backward sampling  & Sample $\nu_T\sim G(n_T/2,d_T/2)$.\\
& & For $t=1,\dots, T$    	\\
& & Sample $\eta_{T-t}\sim G[(1-\delta)n_{T-t}/2,d_{T-t}/2]$.\\
&& Set $\phi_{T-t}=\eta_{T-t}+\delta\phi_{T-t+1}.$\\
\hline\hline 
\end{tabular}}
\end{center}
\caption{\small An MCMC algorithm with $DSS$ priors and a Gaussian spike. Note that $G(a,b)$ denotes a gamma distribution with a mean $a/b$.}
\label{EM}
\end{table}

\subsubsection{MCMC  Step 1: Sampling Regression Coefficients $\beta_t$} 
Conditionally on the inclusion indicators $\gamma_{tj}$ for $0\leq t\leq T$ and $1\leq j\leq p$ and variances $v_t$,  we  have a  conjugate dynamic linear model  $y_t=\x_t'\b_t+\varepsilon_t, \quad \varepsilon_t\sim \mathcal{N}(0,v_t)$ with
\begin{align*}
	\b_t&=\H_t+\bG_t(\b_{t-1}-\H_t)+\bm{e}_t, \quad \bm{e}_t\sim \mathcal{N}(0, \W_t),
\end{align*}
 where
 \begin{align}
 \W_t&=\textrm{diag}\left\{\gamma_{tj}\lambda_{1}+(1-\gamma_{tj})\lambda_{0}\right\}_{j=1}^p,\label{variance_state}\\
 \bG_t &= \textrm{diag}\left\{\gamma_{tj}\phi_{1}\right\}_{j=1}^p,\\
 \H_t &= \phi_0\bg_t'.
 \end{align}
 We note that the initial vector at time zero $\b_{0}$ is subject to estimation as well. As its prior, we use the stationary distribution $\b_{0}\sim\mathcal{N}(\m_0 , \bC_0),$
 where 
 \begin{equation}\label{muC}
 \m_0=\phi_0\bg_0\quad\text{ and} \quad \bC_0=\mathrm{diag}\{\gamma_{0j}\lambda_1/(1-\phi_1^2)+(1-\gamma_{0j})\lambda_0\}_{j=1}^p
 \end{equation}
  are obtained from \eqref{initial} with a Gaussian spike.

For simplicity of notation we will denote  with $\b_{t:T}=[\b_{t},\dots,\b_T]$ the collection of all $p$ coefficient series from time $t<T$ to time $T$ (similarly $\x_{t:T}$ and $\bg_{t:T}$ for the observations and latent inclusion indicators). Writing   $\y_{1:T}=(y_1,\dots, y_T)'$, 
the conditional posterior $\pi(\b_{0:T}|\y_{\seq1T},\bg_{0:T},\bm v_{0:T})$ can be simulated  from using the standard  FFBS 
algorithm \citep{Schnatter1994}.  In detail, the calculations  proceed as follows. 

\begin{itemize} 
\item[]{\em\bf Forward filtering:} As described in Section 4.3.1 of \cite{Prado2010}, for each $t>0$ we perform the following steps.
\begin{itemize} 
	\item[1.]{\em Time $t-1$ posterior as a prior for $\b_t$:} Based on the information up to $t-1$, we obtain the following Gaussian prior for the state vector $\b_t$:
		\begin{align*}
		\b_{t}\,|\, \bm y_{\seq1{t-1}},\bg_{1:t},\bm v_{\seq1{t-1}}&\sim \mathcal{N}(\a_{t}, \bR_t), 
		\end{align*}
		where $\a_{t}=\H_t+\bG_t(\m_{t-1}-\H_t)$ and $\bR_{t}=\bG_t\bm C_{t-1}\bG_t'+\W_t$. 
		 
	\item[3.]{\em  One-step-ahead predictive distribution:} At time $t-1$ compute
		$$
		y_t \,|\, \y_{\seq1{t-1}},  \bm v_{1:t} \sim \mathcal{N}(f_t,q_t),
		$$ 
		where
		$$
		f_t=\x_t'\a_{t}\quad \textrm{and}\quad q_t=\x_t'\bR_t\x_t+v_t.
		$$
		Observing $y_t$ produces the forecast error $e_t=y_t-f_t$.
	\item[4.]{\em  Posterior for  $\b_t$:} Given  current information up to time $t$, we have  
		\begin{align*}
		\b_{t}\,|\, \y_{\seq1{t}},\bg_{\seq1{t}},\bm v_{1:t}&\sim \mathcal{N}(\m_{t}, \bC_{t}),
		\end{align*}
		 with mean and covariance   $\m_t=\a_{t}+\A_t e_t$ and $ 	\bC_{t}=\bR_t-\bm A_t\bm A_t'q_t,$
		where $\A_t=\bR_t\x_t/q_t$. 
\end{itemize} 
	\item[]{\em\bf Backward sampling:}  Having run the forward filtering analysis up to time $T,$ one then extrapolates into the past with backward sampling. 
	This proceeds as follows.
	 \begin{itemize}
	 \item[a.] At time $T$, simulate $\b_T$ from the normal posterior 
	 $$
	 \b_T\sim\mathcal{N}(\m_T,\bC_T).
	 $$
	 \item[b.]{\em Recursively sample backwards in time}.  For any $t\leq T$,   sample  
	 	$\b_t$ from the conditional normal posterior 
	 			
				$$
				 \b_{t}\,|\, \b_{t+1:T}, \y_{\seq1T},\bg_{\seq1T},\bm v_{1:T}\sim \mathcal{N}\left(\bm a_T(t-T),\bm R_{T}(t-T)\right),
				$$ 
				where  
				\begin{align}
				\bm a_T(t-T)&=\bm \m_t +\bm B_t[\b_{t+1}- \bm a_{t+1} ],\\
				\bm R_T(t-T)&=\bm C_t-\bm B_t\bm R_{t+1}\bm B_t',
				\end{align}
				where $\bm B_t=\bm C_t\bG_{t+1}'\bm R_{t+1}^{-1}$. 				%with mean 
	 	 \end{itemize} 
\end{itemize} 
\subsubsection{MCMC   Step 2:   Sampling the Inclusion Indicators $\bg_{0:T}$}

 Conditionally on the most recently sampled values of the DLM parameters $\b_{0:T},$ the  MCMC calculation proceeds with sampling the 
inclusion indicators $\bg_{0:T}$ from their full conditional posterior. This amounts to   sampling   each entry $\gamma_{tj}$ individually,  making distributed implementations possible, if needed.
For each $1\leq t\leq T$ and $1\leq j\leq p$ we perform the following steps.
\begin{itemize} 
	\item[1.]{\em Compute the mixing weight $\theta_{tj}$:} 
	
	We first recall the stationary spike and slab distributions
		\begin{equation*}
			\psi_0^{ST}(\beta \C \lambda_0)\equiv \mathcal{N}( 0,{\lambda_0})\quad\text{and}\quad
				\psi_1^{ST}(\beta \C \lambda_1,\phi_0,\phi_1)\equiv \mathcal{N}\left( \phi_0,\frac{\lambda_1}{1-\phi_1^2}\right).
		\end{equation*}
		Given $\Theta$, we then compute the mixing weight $\theta_{tj}$ as
		\begin{equation}\label{theta_for_MCMC}
			\theta_{tj}\equiv\theta(\beta_{t-1,j})=\frac{\Theta\psi_1^{ST}\left(\beta_{t-1,j}|\lambda_1,\phi_0,\phi_1\right)}{\Theta\psi_1^{ST}\left(\beta_{t-1,j}|\lambda_1,\phi_0,\phi_1\right)+(1-\Theta)\psi_0^{ST}\left(\beta_{t-1,j}|\lambda_0\right)}.
		\end{equation}

	\item[2.]{\em Compute the conditional inclusion probability $p^\star_{tj}(\beta_{tj})$:}

        First, we update $\mu_{tj}$ from
	$	\mu_{tj}=\phi_{0}+\phi_{1}(\beta_{t-1,j}-\phi_0)$
	 and   recall the conditional spike and slab distributions
	\begin{equation*}
			\psi_0(\beta_{tj}\C  \lambda_0)\equiv \mathcal{N}(0,{\lambda_0})\quad\text{and}\quad
			\psi_1(\beta_{tj}\C \mu_t, \lambda_1)\equiv \mathcal{N}\left(\mu_{tj},\lambda_1\right).
		\end{equation*}
	
	We then compute $p^\star_{tj}(\beta_{tj})$  as
	\begin{equation}\label{pstar_first}
		p^\star_{tj}(\beta_{tj})\equiv\frac{\theta_{tj}  \psi_1(\beta_{tj}\C\mu_{tj},\lambda_1)}{\theta_{tj}  \psi_1(\beta_{tj}\C\mu_{tj},\lambda_1)+(1-\theta_{tj}) \psi_0(\beta_{tj}\C\lambda_0)}.
	\end{equation}
	
	\item[3.]{\em  Sample the indicator $\gamma_{tj}$:}
	
		Given $p^\star_{tj}(\beta_{tj})$, we  sample $\gamma_{tj}$ from 
		\begin{equation}\label{eq:sample_gamma}
		\gamma_{tj}\sim \mathrm{Bernoulli}\, [p^\star_{tj}(\beta_{tj})].
		\end{equation}
\end{itemize} 
Finally, to update the indicators at time $t=0$, we sample
$
\gamma_{0j}\sim \mathrm{Bernoulli}[ \theta(\beta_{0j})].
$

{
\subsubsection{MCMC  Step 3: Sampling Observation Variances $v_t$} 
Recall that the observation precisions $\nu_t=1/v_t$ follow the discounted stochastic volatility process \eqref{eq:dsvt}. Proceeding similarly as in Sections 10.8.2 and 10.8.4 of \cite{WestHarrison1997book2}, one can forward-filter and backward-sample from the conditional distributions. Our calculations here are slightly different because we are {\em conditioning} on $\b_{0:T}$ rather than margining them out (as in, e.g.,  Theorem 4.3 of \cite{WestHarrison1997book2}). Starting with a gamma prior $\phi_0\sim G(n_0/2,d_0/2)$, at time $t>0$ the prior distribution
$$
\nu_t\,\C\, \y_{1:(t-1)},\b_{1:(t-1)}\sim G(\delta n_{t-1}/2,\delta d_{t-1}/2)
$$
updates $y_t$ into the posterior distribution 
\begin{equation}\label{eq:recurrent_variance}
\nu_t\,\C\, \y_{1:t},\b_{1:t}\sim G( n_{t}/2,  d_{t}/2)\quad\text{with}\quad n_t=\delta n_{t-1}+1\quad\text{and}\quad d_t=\delta d_{t-1}+r_t^2
\end{equation}
where $r_t=y_t-\x_t'\b_t$. Note that  in this parametrization, the mean equals
\begin{equation}\label{eq:mean}
1/S_{t}\equiv \E [\nu_t\C  \y_{1:t},\b_{1:t}] =n_{t}/d_{t}.
\end{equation}
These  forward  equations are followed by backward sampling. First, one samples $\nu_T\,\C\,  \y_{1:T},\b_{1:T} \sim G( n_{T}/2,  d_{T}/2)$.
%In order  to sample from
%$\phi_{T-t}\C \phi_{T-t+1}, \y_{1:T},\b_{1:T} $
Using the recurrent relations (page 364 of \cite{WestHarrison1997book2})
\begin{equation}\label{eq:back_sample}
\phi_{T-t}=\eta_{T-t}+\delta \phi_{T-t+1}\quad\text{where}\quad \eta_{T-t}\sim G[(1-\delta)n_{T-t}/2, d_{T-t}/2]
\end{equation}
one then draws $\eta_{T-t}$ to obtain a sample $\phi_{T-t}$ from \eqref{eq:back_sample}.
}

In addition to the MCMC algorithm, we also derive MAP smoothers using a penalized likelihood approach.

\section{Dynamic EMVS}\label{sec:dynamic_EMVS}
Unlike previous developments \citep{nakajima_west,kalli_griffin}, this paper also views  Bayesian dynamic shrinkage  through the lens of optimization. Rather than distilling posterior samples  to learn about $\b_{1:T}=[\b_1,\dots,\b_T]$,  we focus  on finding the MAP trajectory $\wh{\b}_{1:T}=\arg\max\pi(\b_{1:T}\C\y_{1:T})$.  MAP sequence estimation problems  (for non-linear non-Gaussian dynamic models)  were addressed previously with,  e.g., Viterbi-style algorithms \citep{godsill}. 
Our optimization strategy is conceptually very different and builds on the  EMVS procedure of \cite{RG14}.  First, we focus on the Gaussian spike prior variant \eqref{spike_gaussian} which allows for very fast block updates in closed form.

{A (local) posterior mode $\wh{\b}_{0:T}$ can be obtained  indirectly through an EM algorithm, treating  $\bG$ and precision parameters $\nu_t=1/v_t$  as the missing data. 
%Direct optimization \eqref{one_site}, even with the aid of Lemma \ref{},  does not admit closed-form solution.   the computation of the selection thresholds $\Delta^{+}_{tj}$ and $\Delta^{-}_{tj}$ as well as finding the maximum
The initial vector $\b_{t=0}=(\beta_{01},\dots,\beta_{0p})'$ at time $t=0$  will be estimated together with all the remaining coefficients $\b_{1:T}$. We assume that $\b_0$ comes from the  stationary distribution  
described in Theorem \ref{thm1},
\begin{equation}\label{initial}
\pi(\b_0|\bg_0)=\prod_{j=1}^p\left[\gamma_{0j}\psi_1^{ST}(\beta_{0j}\C\lambda_1,\phi_0,\phi_1)+(1-\gamma_{0j})\psi_0(\beta_{0j}\C\lambda_0)\right],
\end{equation}
where $\bg_0=(\gamma_{01},\dots,\gamma_{0p})'$  are independent binary indicators with $\P[\gamma_{0j}=1\C\Theta]=\Theta$ for $1\leq j\leq p$. 
Knowing the stationary distribution is thereby useful   for specifying the initial conditions. The goal is obtaining the   mode $\wh{\b}_{0:T} $ of the functional 
$\pi(\b_{0:T}|\y_{1:T})$. To this end, we proceed iteratively by augmenting this objective function with the missing data $\bg_{0:T}$, as prescribed by \cite{RG14},  and then maximizing w.r.t. $\b_{0:T}$. 
An important  observation, that facilitates the derivation of the algorithm, is that the prior distribution $\pi(\b_{0:T},\bg_{0:T}, \v_{1:T})$ can be factorized into the following products
$$
\pi(\b_{0:T},\bg_{0:T}, \v_{1:T})=\pi(\b_0|\bg_0)\pi(\bg_0)\prod_{t=1}^T\left[\pi(v_t\C v_{t-1})\prod_{j=1}^p\pi(\beta_{tj}|\gamma_{tj},\beta_{t-1j})\pi(\gamma_{tj}|\beta_{t-1j})\right],
$$
where $\pi(\beta_{tj}|\gamma_{tj},\beta_{t-1j})$ and $\pi(\gamma_{tj}|\beta_{t-1j})$ are defined in \eqref{betas} and \eqref{gammas}, respectively.
For simplicity, we will outline the procedure assuming $\phi_0=0$ and thereby $\mu_{tj}=\phi_1\beta_{t-1j}$. Then, we can write
\begin{align}\label{Q1}
&\log\pi(\b_{0:T},\bg_{0:T}, \v_{1:T}|\bm y_{1:T})=C(\bm v_{1:T},\phi_1) + \sum_{t=1}^T\sum_{j=1}^p\left[\gamma_{tj}\log\theta_{tj}+(1-\gamma_{tj})\log(1-\theta_{tj})\right]\\
&\quad-\sum_{t=1}^T\left\{\frac{(y_t-\x_t'\b_t)^2}{2v_t}+\sum_{j=1}^p\left[ \gamma_{tj}\frac{(\beta_{tj}-\phi_1\beta_{t-1j})^2}{2\lambda_1} +  (1-\gamma_{tj})\frac{\beta_{tj}^2}{2\lambda_0}\right]+\log\pi(v_t\C v_{t-1})\right\}\notag\\
&\quad-\sum_{j=1}^p\left[\gamma_{0j}\frac{\beta_{0j}^2(1-\phi_1^2)}{2\lambda_1}+(1-\gamma_{0j})\frac{\beta_{0j}^2}{2\lambda_0}-\gamma_{0j}\log\Theta-(1-\gamma_{0j})\log(1-\Theta)\right]\notag.
\end{align}
We will  endow  the parameters  $\b_{0:T}$ with a superscript $m$ to designate their most recent values at the $m^{th}$ iteration. In the E-step, we compute the conditional expectation 
of \eqref{Q1} with respect to the conditional distribution of $[\bg_{0:T},\bm \nu_{1:T}]$, given $\b_{0:T}^{(m)}$ and $\bm y_{1:T}$. This boils down to computing conditional inclusion probabilities
 $
 p^\star_{tj}=\P(\gamma_{tj}=1|\beta_{tj}^{(m)},\beta_{t-1j}^{(m)},\theta_{tj})
 $ 
 from \eqref{pstar_first}, when $t>0$, and $p^\star_{0j}\equiv\theta_{1j}\equiv\theta(\beta_{0j})$ from \eqref{weights},  and replacing all the $\gamma_{tj}$'s in \eqref{Q1} with $p^\star_{tj}$'s. 
 Additionally, one replaces $1/v_t$ with the conditional expectation $\E[\nu_t\C \b_{0:T}, \bm y_{1:T}]$ available in closed from the recurrent relations (\cite{WestHarrison1997book2} on page 364)
 $$
 \E[\nu_t\C\b_{0:T}^{(m)}, \bm y_{1:T}]=(1-\delta)n_{t}/d_t+\delta \E[\nu_{t+1}\C \b_{0:T}^{(m)}, \bm y_{1:T}]\quad\text{for}\quad 1\leq t<T,
 $$
 where $n_t$ and $d_t$ are obtained from \eqref{eq:recurrent_variance} and where $ \E[\nu_T\C \b_{0:T}^{(m)}, \bm y_{1:T}]=n_T/d_T$.
%The objective function to be maximized in the M-step is thus
%\begin{align*}
%\E_{\bg_0,\bG|\cdot}\log&\pi(\b_0,\B,\bg_0,\bG|\Y)=-\frac{1}{2}\sum_{t=1}^T{(y_t-\x_t'\b_t)^2}-\sum_{t=1}^T\sum_{j=1}^p\left[p^\star_{tj}\frac{(\beta_{tj}-\mu_{tj})^2}{2\lambda_1}+(1-p^\star_{tj})\lambda_0|\beta_{tj}|\right]\\
%&+\sum_{t=1}^T\sum_{j=1}^p[p^\star_{tj}\log\theta_{tj}+(1-p^\star_{tj})\log(1-\theta_{tj}),
%\end{align*}
%where $\E_{\bg_0,\bG|\cdot}(\cdot)$ denotes the conditional expectation $\E(\cdot|\b_0^{(m)},\B^{(m)},\Y)$.
In the M-step, we set out to maximize $\E_{\bg_{0:T},\bm \nu_{1:T}|\cdot}\log\pi(\b_{0:T},\bg_{0:T},\bm v_{1:T}|\y_{1:T})$ w.r.t. $\b_{0:T}$.  This is achieved in a block-wise fashion, where we update $\b_t$ given the most recent updates of $\b_{t-1}$ and $\b_{t+1}$. Given the conjugacy of the Gaussian distribution, these updates have closed forms (similarly as in the EMVS procedure of \cite{RG14}). We summarize the steps in the Table \ref{EM}. 
It is worth pointing out that the matrix inversion $\Sigma_t^{-1}$ in the step M1 in Table \ref{EM} can be avoided  using the fact that $\x_t\x_t'$ is a rank-one matrix. 
Denote with $D_t=\mathrm{diag}\{\frac{p_{tj}^\star}{\lambda_1}+\frac{1-p_{tj}^\star}{\lambda_0}+\mathbb{I}({t<T})\frac{\phi_1^2 p^\star_{t+1j}}{\lambda_1} \}_{j=1}^p$. Then the Woodburry-Shermann matrix inversion lemma yields
$$
\Sigma_t^{-1}=D_t^{-1}-\nu_t^\star D_t^{-1}\frac{\x_t\x_t'}{1+\nu_t^\star \x_t'D_t^{-1}\x_t}D_t^{-1}.
$$
Due to this trick, the computation of the M-step is extremely fast.  Since each update $\b_t$ is conditional on all $\beta_j,j\neq t$, we are performing conditional maximization in the spirit of Expectation-Conditional-Maximization \cite{meng_ECM}.
In order to speed up convergence, we can afford to loop over these simple updates inside each M-step. We found loops of size $100$ to perform well.

Additionally, we can estimate the autoregressive parameter $\phi_1$ under (a discretized version) of the prior \eqref{beta_prior} by updating $\phi_1$ at each iteration with the value that maximizes 
the expected log-complete posterior
$\E_{\bg_{0:T},\bm \nu_{1:T}|\cdot}\log\pi(\b_{0:T},\bg_{0:T},\bm v_{1:T}|\y_{1:T})$. One can compute this criterion for a grid of values $\phi_1$ and pick the one value that maximizes the expected log-complete posterior. Estimation of $\phi_0$ can be incorporated in a similar vein.

In the next section, we develop a penalized likelihood approach to MAP smoothing using a Laplace spike prior.
}

\begin{table}[!t]
\small
\begin{center}
\scalebox{0.8}{
\begin{tabular}{|lll|}
\hline
\multicolumn{3}{|c|}{\cellcolor[HTML]{C0C0C0} \textbf{Algorithm:} \textit{Dynamic EMVS algorithm}} \\ \hline\hline
 &     &Initialize $\beta_{tj}$ for $t=0,\dots, T$ and $j=1,\dots, p$.\\
\multicolumn{3}{|c|}{\cellcolor[HTML]{C0C0C0}E-Step}    \\
& & For $j=1,\dots, p$                                              \\
E1:             & Compute mixing weights                  & Compute $\theta_{tj}=\theta(\beta_{t-1j}) $ for $1\leq t\leq T$ from \eqref{weights}.                                             \\ 
& & Compute $p^\star_{tj}=p^\star_{tj}(\beta_{tj)}$ for $1\leq t\leq T$  from \eqref{pstar}. \\
& & Compute $p^\star_{0j}=\theta(\beta_{0j})$   from \eqref{weights}. \\
E1:             & Compute precisions                  & For $t=1,\dots, T$                                             \\ 
       &    & Compute $n_t=\delta n_{t-1}+1$ and $d_t=\delta d_{t-1}+r_t^2$, where $r_t=y_t-\bm x'_t\b_t$.\\
       &    & Set  $\nu_T^\star=n_T/d_T$.\\
& & For $t=T-1,\dots, 1$ set $\nu_{t}^\star= (1-\delta)n_t/d_t+\delta\nu_{t+1}^\star.$\\
\multicolumn{3}{|c|}{\cellcolor[HTML]{C0C0C0}M-Step: Gaussian spike version}                                                  \\
M1:  & Compute regression coefficients   & For $t=1,\dots,T$\\
	&     & Compute $\Sigma_t=\nu_t^\star\x_t\x_t' +\mathrm{diag}\{\frac{p_{tj}^\star}{\lambda_1}+\frac{1-p_{tj}^\star}{\lambda_0}+\mathbb{I}({t<T})\frac{\phi_1^2 p^\star_{t+1j}}{\lambda_1} \}_{j=1}^p$\\
	&     & Compute $\bm \mu_t=\nu_t^\star y_t \x_t+\frac{\phi_1}{\lambda_1}\b_{t-1}\odot \bm p^\star_{t}+\mathbb{I}({t<T})\frac{\phi_1}{\lambda_1}\b_{t+1} \odot \bm p^\star_{t+1}$\\
	&     & Update $\b_t =\Sigma_t^{-1}\bm \mu_t$\\
	&     & Compute $\Sigma_0=\mathrm{diag}\{\frac{(1-\phi_1^2)p_{0j}^\star}{\lambda_1}+\frac{1-p_{0j}^\star}{\lambda_0}+\frac{\phi_1^2 p^\star_{1j}}{\lambda_1} \}_{j=1}^p$\\
	&     & Update $\b_0 =\frac{\phi_1}{\lambda_1}\Sigma_0^{-1}\b_{1} \odot \bm p^\star_{1}  $\\
\multicolumn{3}{|c|}{\cellcolor[HTML]{C0C0C0}M-Step: Laplace spike version}                                                  \\
       &                                    & For $j=1,\dots, p$  and $t=1,\dots, T$                                             \\
M2: &  Update  regression coefficients  & Compute $\beta_{0j}$  using \eqref{beta0_update}. \\
       &                                    & Compute $\beta_{tj}$  using \eqref{beta_update}.                                                                \\
\hline\hline
\end{tabular}}
\end{center}
\caption{\small  Dynamic EMVS algorithm for both the Gaussian spike \eqref{spike_gaussian} and the Laplace spike. The notation $\bm a\odot \bm b$ denotes elementwise vector multiplication. %\textcolor{green}{EM algorithm obtained with initial condition $\bm\omega_0 \sim \mathcal{N}_K(\bm 0, 1/(1-\wt{\phi}^2)\bm I_K)$}
}
\label{EM}
\end{table}

}

}

\section{Dynamic Spike-and-Slab Penalty}\label{sec:pen}
Spike-and-slab priors give rise to self-adaptive penalty functions for MAP estimation, as detailed in \cite{rockova15} and \cite{SSL}. Here, we introduce  elaborations for dynamic shrinkage implied by the $DSS$ priors. 
%To make the exposition simpler, we will assume $p=1$ and suppress the subscript $j$ from the notation.
%We will focus on one particularly useful variant of the prior with a Laplace spike (i.e. $\psi_0(\beta|\lambda_0)=\frac{\lambda_0}{2}\e^{-|\beta|\lambda_0}$) and a Gaussian slab, centered around $\mu_{t}$ with variance $\lambda_1$.

\begin{definition}
For a given set of parameters $(\Theta,\lambda_0,\lambda_1,\phi_0,\phi_1)$, we define a {\sl prospective} penalty function implied by  \eqref{betas} and \eqref{weights}
as follows:
\begin{equation}\label{penalty1}
pen(\beta\C\beta_{t-1})=\log\left[\left(1-\theta_t\right)\psi_0(\beta\C\lambda_0)+\theta_t\,\psi_1(\beta\C\mu_t,\lambda_1)\right].
\end{equation}
Similarly, we define a {\sl retrospective}  penalty $pen(\beta_{t+1}\C\beta)$  as a function of the second argument $\beta$ in \eqref{penalty1}.
The Dynamic Spike-and-Slab  (DSS)  penalty is then defined as 
\begin{equation}\label{penalty_total}
Pen(\beta\C\beta_{t-1},\beta_{t+1})=pen(\beta\C\beta_{t-1})+pen(\beta_{t+1}\C\beta)+C,
\end{equation}
where $C\equiv-Pen(0\C\beta_{t-1},\beta_{t+1})$ is a norming constant.
% For $t=T$, the DSS penalty is defined simply as $Pen(\beta_T|\beta_{T-1})\equiv pen(\beta_{T}|\beta_{T-1})$.
\end{definition}
\begin{remark}
%As is customary, the penalty \eqref{penalty1} has been  normalized so that $pen(0|\beta_{t-1})=0$.
Note that the dependence on the previous value $\beta_{t-1}$ in $pen(\beta\C\beta_{t-1})$ is hidden in $\theta_t$ and $\mu_t$. Throughout the paper, we will  write $\partial \theta_t/\partial\beta_{t-1}$ and $\partial\mu_t/\partial\beta_{t-1}$  without reminding ourselves of this implicit relationship.
\end{remark}

\begin{figure}[!t]
     \subfigure[$pen(\beta|\beta_{t-1}=1.5)$]{
     \begin{minipage}[t]{0.32\textwidth}
       \hskip-1pc  \scalebox{0.27}{\includegraphics{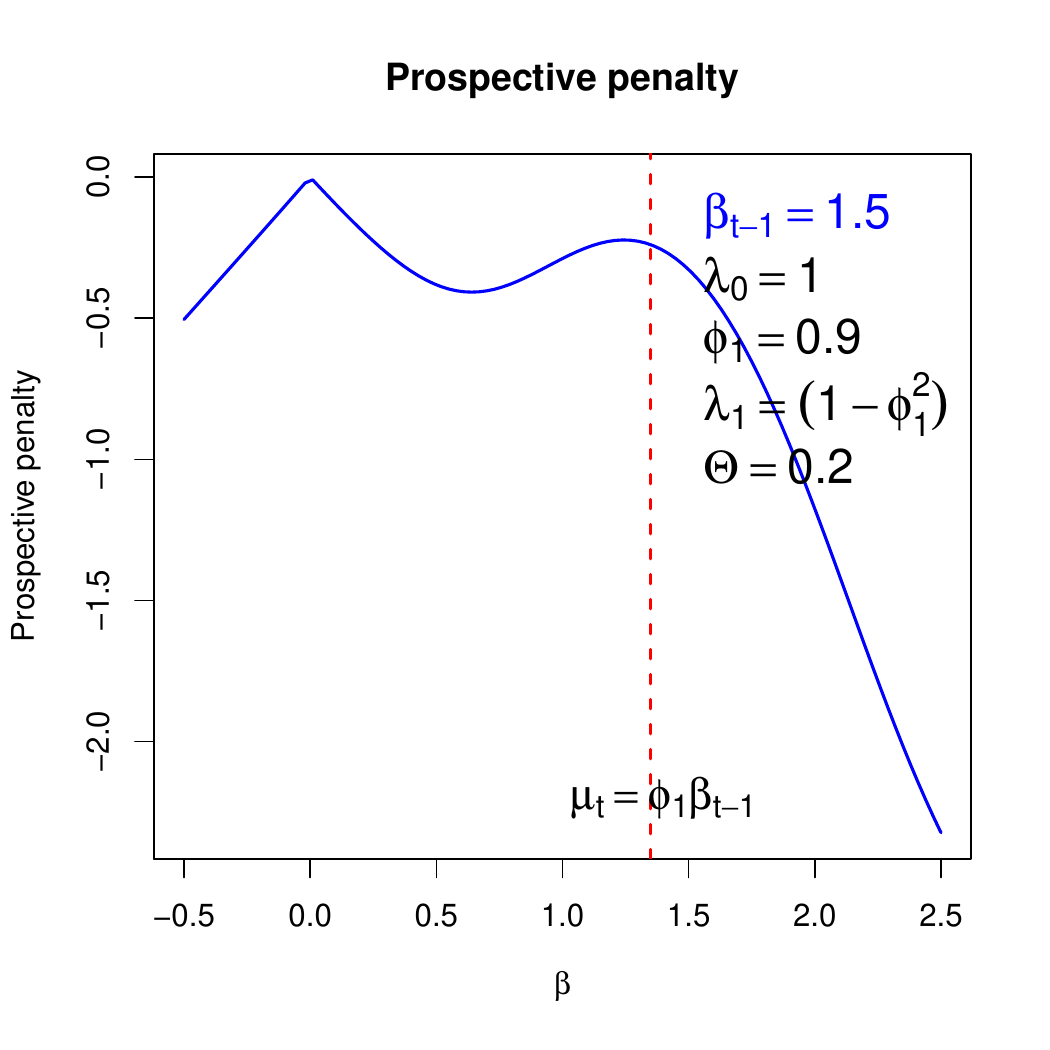}}  \label{fig1a}
    \end{minipage}}
     \subfigure[$pen(\beta|\beta_{t-1}=1.5)$]{
    \begin{minipage}[t]{0.32\textwidth}
    \hskip-1pc  \scalebox{0.27}{\includegraphics{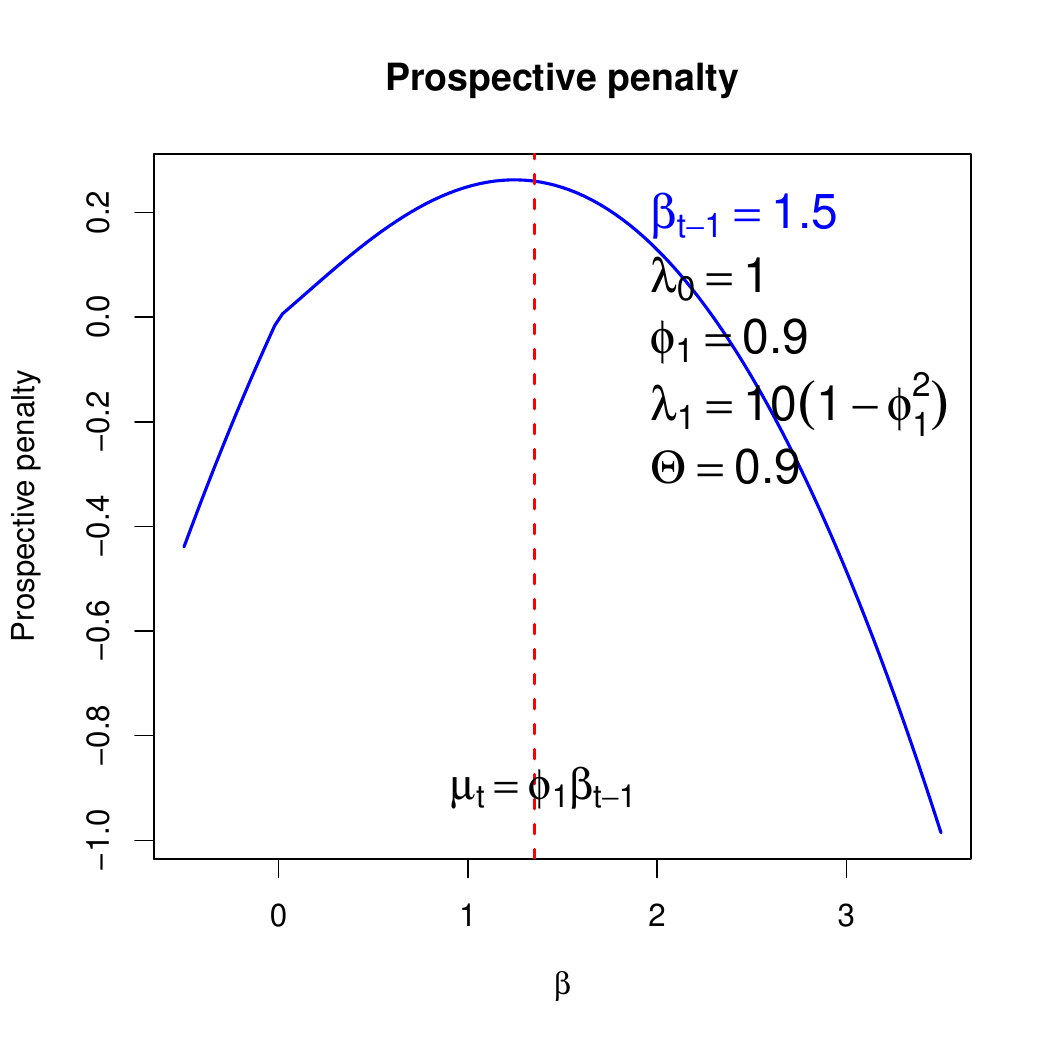}}     \label{fig1b}
    \end{minipage}}
     \subfigure[$pen(\beta|\beta_{t-1}=0.5)$]{
    \begin{minipage}[t]{0.32\textwidth}
     \hskip-1pc \scalebox{0.27}{\includegraphics{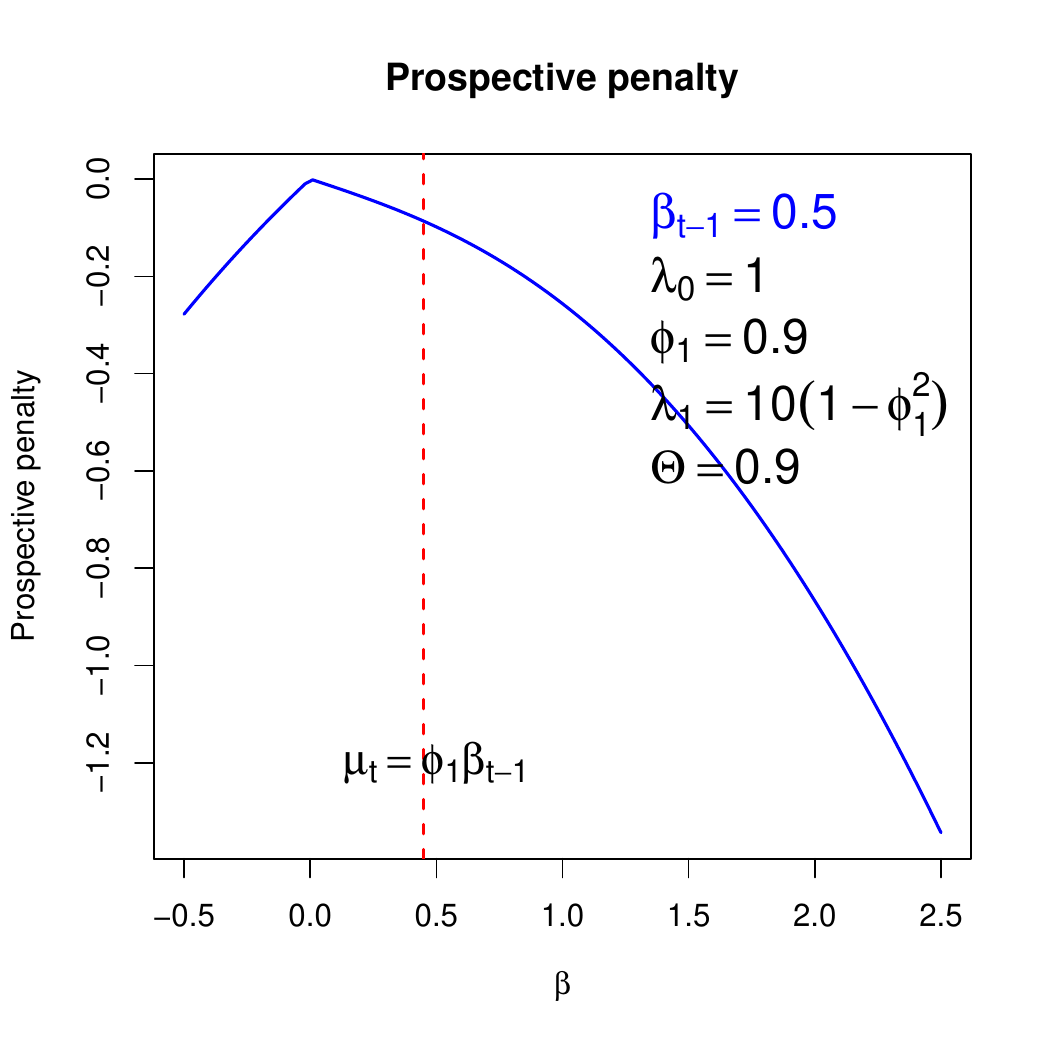}}   \label{fig1c}
    \end{minipage}}
    \caption{Plots of the prospective penalty function under the Laplace spike.}    \label{fig1}
\end{figure}

As an example, we consider the Laplace spike prior $\psi_0(\beta\C\lambda_0)=\lambda_0/2\e^{-\lambda_0|\beta|}$.
Figure \ref{fig1} portrays the prospective penalty for two choices of $\beta_{t-1}$ and two sets of tuning parameters $\phi_1,\lambda_1,\lambda_0$ and $\Theta$ (assuming $\phi_0=0$). 
Because the conditional transfer equation \eqref{betas}  is a mixture, $pen(\beta\C\beta_{t-1})$ is apt to be multimodal.  Figure \ref{fig1a} shows an obvious peak at zero (due to the Laplace spike), but also  a peak around $\mu_t=0.9\times\beta_{t-1}$,  prioritizing values in the close vicinity of the previous value (due to the non-central slab).  From an implementation viewpoint, however, it is more desirable that the penalty be uni-modal, reflecting the size of the previous coefficient without ambiguity by suppressing one of the peaks. Such behavior is illustrated in Figure \ref{fig1b} and Figure \ref{fig1c}, where the  penalty flexibly adapts to $|\beta_{t-1}|$ by promoting  either zero {\sl or} a value close to  $\beta_{t-1}$. This effect is achieved with a relatively large stationary slab variance, such as $\lambda_1/(1-\phi_1^2)=10$, a mild Laplace peak $\lambda_0=1$ and the marginal importance weight  $\Theta=0.9$. Smaller values  $\Theta$ would provide an overwhelming support for the zero mode.  The parameter $\Theta$,  thus should not be regarded as a proportion of active coefficients (as is customary with point-mass mixtures), but rather an interpretation-free tuning parameter. 

Figure \ref{fig1} plots $pen(\beta\C\beta_{t-1})$ prospectively as a function of $\beta$, given the previous value $\beta_{t-1}$. It is also illuminating to plot $pen(\beta_{t+1}\C\beta)$ retrospectively as a function of $\beta$, given the future value $\beta_{t+1}$. Two such retrospective penalty plots are provided in Figure \ref{fig11a} and Figure \ref{fig11b}. When the future value is relatively large ($\beta_{t+1}=1.5$ in Figure \ref{fig11b}),  the penalty $pen(\beta_{t+1}\C\beta)$ has a peak near $\beta_{t+1}$, signaling that the value $\beta_{t}$ must be large too. When the future value is small ($\beta_{t+1}=0$ in Figure \ref{fig11a}), the penalty has a peak at zero signaling that the current value $\beta_{t}$ must have been small. Again, this balance is achieved with a relatively large stationary slab variance and a large  $\Theta$.  Note that under the Gaussian spike   \eqref{spike_gaussian}, the penalty functions will be differentiable at zero.

\begin{figure}[!t]
     \subfigure[$pen(\beta|\beta_{t+1}=0)$]{
     \begin{minipage}[t]{0.32\textwidth}
       \hskip-1pc  \scalebox{0.27}{\includegraphics{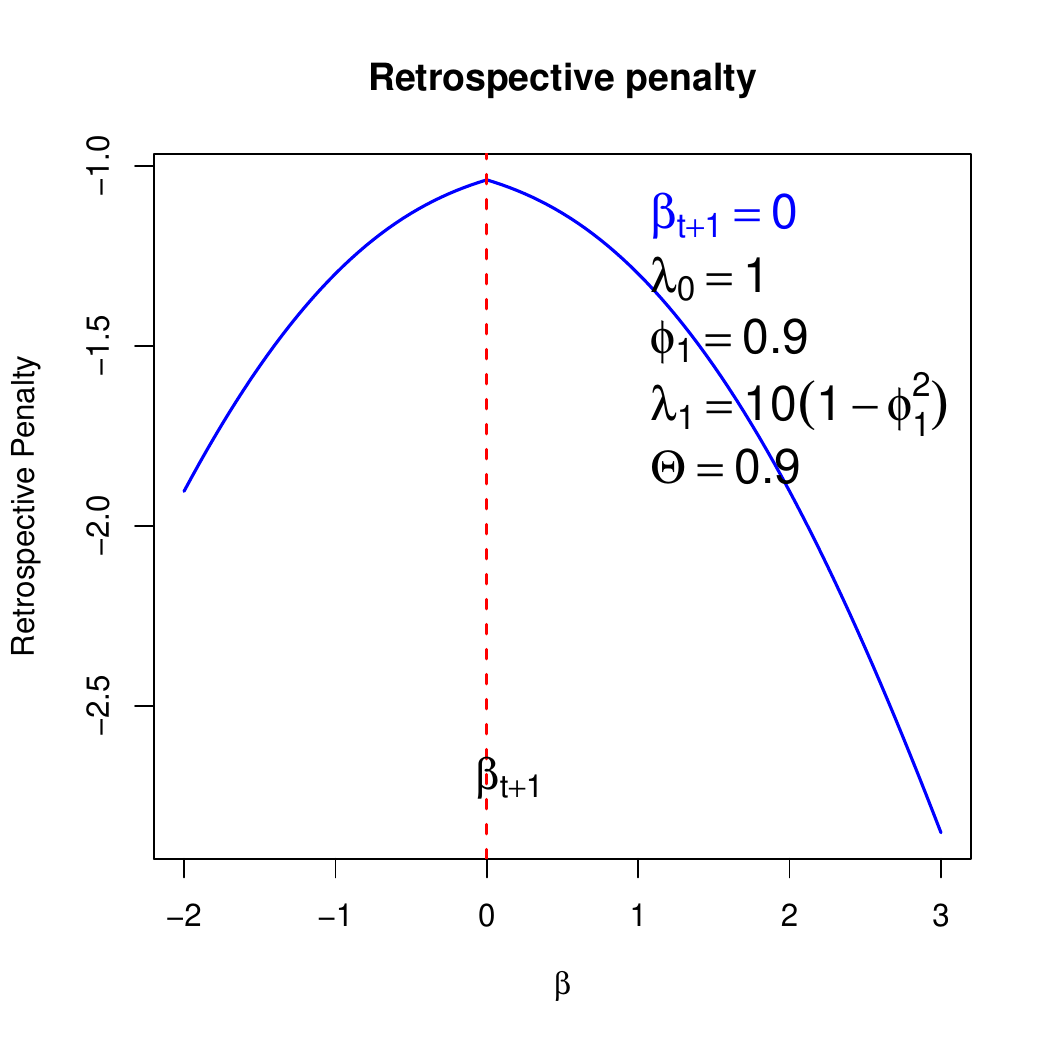}}  \label{fig11a}
    \end{minipage}}
     \subfigure[$pen(\beta|\beta_{t+1}=1.5)$]{
    \begin{minipage}[t]{0.32\textwidth}
    \hskip-1pc  \scalebox{0.27}{\includegraphics{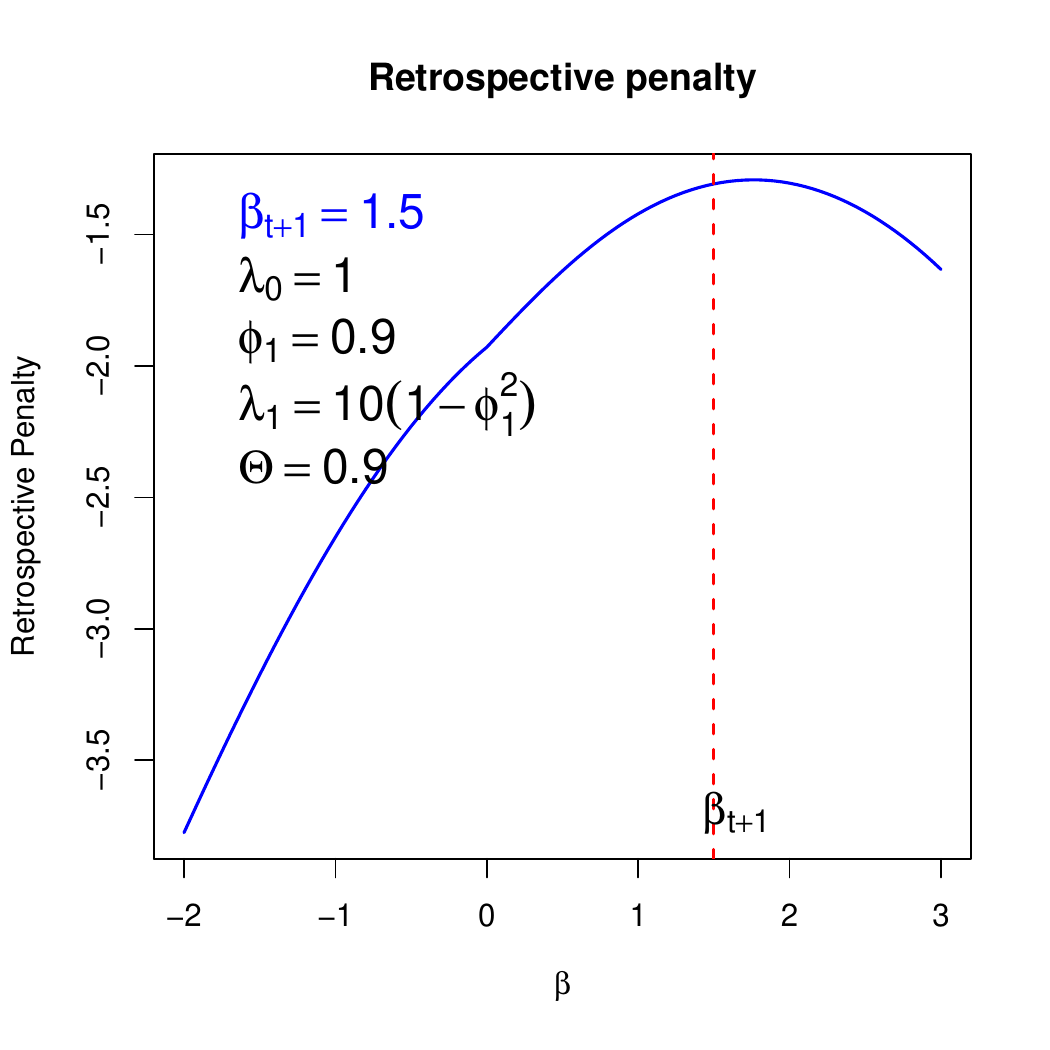}}     \label{fig11b}
    \end{minipage}}
     \subfigure[$\phi_0=0,\phi_1=0.9,\lambda_1=10(1-\phi_1^2)$]{
    \begin{minipage}[t]{0.32\textwidth}
     \hskip-1pc \scalebox{0.27}{\includegraphics{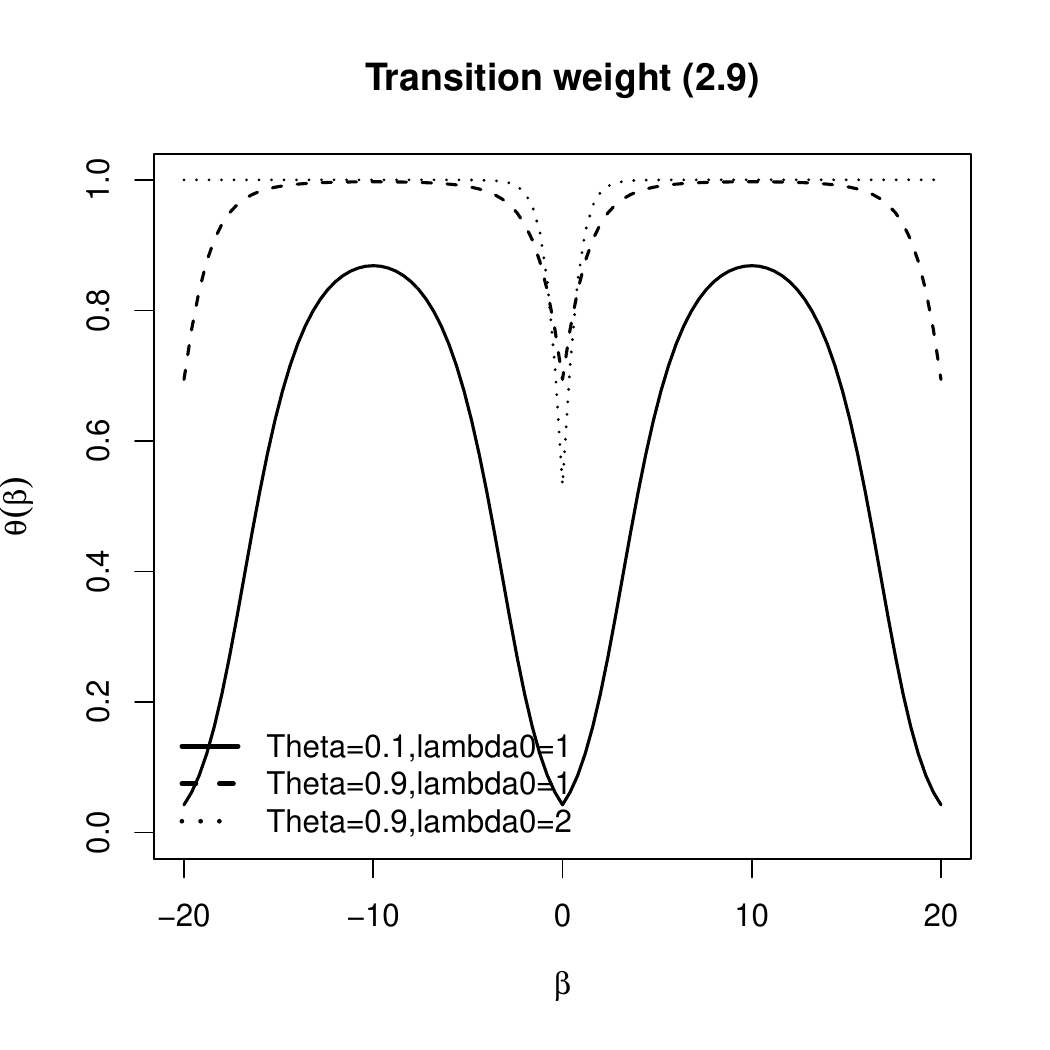}}   \label{fig11c}
    \end{minipage}}
    \caption{Plots of the retrospective penalty function and the mixing weight \eqref{weights} under the Laplace spike.}    \label{fig11}
\end{figure}

The behavior of  the prospective and retrospective penalties is ultimately tied to the mixing weight $\theta_t\equiv\theta(\beta)$  in \eqref{weights}. It is desirable that $\theta(\beta)$ is increasing with  $|\beta|$. However, Laplace tails  will begin to dominate for large enough $|\beta|$, where the probability $\theta(\beta)$ will begin to drop (for $|\beta|$ greater than $\delta\equiv(\lambda_0+\sqrt{2C/A})A$, where $A=\lambda_1/(1-\phi_1^2)$ and $C=\log[(1-\Theta)/\Theta\lambda_0/2\sqrt{2\pi A}]$). However, we can make the turning point $\delta$ large enough with larger values $\Theta$ and smaller values $\lambda_0$, as indicated in Figure \ref{fig11c}.

To describe  the shrinkage dynamics implied by the  penalty \eqref{penalty_total}, it is useful to study the partial derivative $\partial Pen(\beta\C\beta_{t-1},\beta_{t+1})/\partial|\beta|$. This  term encapsulates how much shrinkage we expect at time  $t$, conditionally on $(\beta_{t-1},\beta_{t+1})$.
We will separate the term into two pieces: a prospective shrinkage effect $\lambda^\star(\beta\C\beta_{t-1})$, driven by the past value $\beta_{t-1}$, and a retrospective shrinkage effect $\widetilde{\lambda}^\star(\beta\C\beta_{t+1})$, driven by the future value $\beta_{t+1}$. More formally, we write
$$
\frac{\partial\, Pen(\beta\C\beta_{t-1},\beta_{t+1})}{\partial|\beta|}\equiv -\Lambda^\star(\beta\C\beta_{t-1},\beta_{t+1}),
$$
where
\begin{equation}\label{big_lambda}
\Lambda^\star(\beta\C\beta_{t-1},\beta_{t+1})=\lambda^\star(\beta\C\beta_{t-1})+\widetilde{\lambda}^\star(\beta\C\beta_{t+1}),
\end{equation}
and
$$
\lambda^\star(\beta\C\beta_{t-1})=-\frac{\partial\, pen(\beta\C\beta_{t-1})}{\partial|\beta|}\quad\text{and}\quad \widetilde{\lambda}^\star(\beta\C\beta_{t+1})=-\frac{\partial\, pen(\beta_{t+1}|\beta)}{\partial|\beta|}.
$$

\subsection{Shrinkage ``from the Past"}\label{sec:prosp}
The prospective shrinkage term $\lambda^\star(\beta\C\beta_{t-1})$ pertains to Bayesian penalty mixing introduced by \cite{rockova15} and \cite{SSL} in the sense that it can be characterized as an adaptive linear combination of individual spike and slab shrinkage terms. In particular, we can write
\begin{align}
\lambda^\star(\beta\C\beta_{t-1})&=-p^\star_t(\beta)\frac{\partial\log \psi_1(\beta\C\mu_{t},\lambda_1)}{\partial|\beta|}-[1-p^\star_{t}(\beta)]\frac{\partial\log \psi_0(\beta\C\lambda_0)}{\partial|\beta|},
\end{align}
where 
\begin{equation}\label{pstar}
p^\star_{t}(\beta)\equiv\frac{\theta_{t}  \psi_1(\beta\C\mu_{t},\lambda_1)}{\theta_{t}  \psi_1(\beta\C\mu_{t},\lambda_1)+(1-\theta_{t}) \psi_0(\beta\C\lambda_0)}.
\end{equation}
For example, using the Laplace spike, one obtains
$$
\lambda^\star(\beta\C\beta_{t-1})=p^\star_t(\beta)\left(\frac{\beta-\mu_t}{\lambda_1}\right)\mathrm{sign}(\beta)+[1-p_t^\star(\beta)]\lambda_0\label{lambdastar}.
$$
Two observations are in order: first, by writing $p^\star_{t}(\beta)=\P(\gamma_{t}=1|\beta_t=\beta,\beta_{t-1},\theta_t)$, \eqref{pstar} can be viewed as a posterior probability for classifying $\beta$ as arising from the {\sl conditional} slab (versus  the spike) at  time $t$,  given the previous value $\beta_{t-1}$.  Second, these weights are {\sl very different from} $\theta_{t}$ in \eqref{weights}, which are classifying $\beta$  as arising from the {\sl marginal} slab (versus the spike). 
From \eqref{pstar}, we can see how $p_{t}^\star(\beta)$ hierarchically transmits information about the past value $\beta_{t-1}$ (via  $\theta_t$) to determine the right shrinkage for  $\beta_t$. This is achieved with a doubly-adaptive  chain reaction. %, which adapts to both $(\beta_{t-1},\beta_t)$. 
Namely, if the previous value $\beta_{t-1}$ was large, $\theta_t$ will be close to one signaling that the next coefficient $\beta_t$ is prone to be in the slab.  Next,  if $\beta_t$ is in fact  large, $p^\star_t(\beta_t)$  will be close to one, where the first summand in \eqref{lambdastar}  becomes the leading term and  shrinks $\beta_t$ towards $\mu_t$. If $\beta_t$ is small, however, $p^\star_t(\beta_t)$ will be  small as well, where the second term in \eqref{lambdastar} takes over to shrink $\beta_t$ towards zero. This gravitational pull is accelerated when the previous value $\beta_{t-1}$ was negligible (zero), in which case $\theta_t$ will be even smaller, making it even more difficult for the next coefficient $\beta_t$ to escape the spike. This mechanism explains how the prospective penalty adapts to both $(\beta_{t-1},\beta_t)$, promoting smooth forward proliferation of spike/slab  allocations and coefficients.

\subsection{Shrinkage ``from the Future"}\label{sec:retro} While the prospective shrinkage term promotes smooth forward proliferation, the retrospective shrinkage term $\widetilde{\lambda}^\star(\beta\C\beta_{t+1})$ operates backwards.
For the Laplace spike,  we can write
 \begin{align}\label{shrinkage_future}
 \widetilde{\lambda}^\star(\beta\C\beta_{t+1})=&-\frac{\partial \theta_{t+1}}{\partial |\beta|}\left[\frac{p^\star_{t+1}(\beta_{t+1})}{\theta_{t+1}}-\frac{1-p^\star_{t+1}(\beta_{t+1})}{1-\theta_{t+1}}\right]\notag\\
 &-p^\star_{t+1}(\beta_{t+1})\phi_1\mathrm{sign}(\beta)\left[\frac{\beta_{t+1}-\mu_{t+1}}{\lambda_1}\right], 
\end{align}
where 
\begin{equation}\label{theta_der}
\frac{\partial\theta_{t+1}}{\partial|\beta|}=\theta_{t+1}(1-\theta_{t+1})\left[\lambda_0-\mathrm{sign}(\beta)\left(\frac{\beta-\phi_0}{\lambda_1/(1-\phi_1^2)}\right)\right].
\end{equation}
For simplicity, we will write $p^\star_{t+1}=p^\star_{t+1}(\beta_{t+1})$. Then we have
\begin{align}
& \widetilde{\lambda}^\star(\beta\C\beta_{t+1})=\notag\\
&\quad\left[\lambda_0-\mathrm{sign}(\beta)\left(\frac{\beta-\phi_0}{\lambda_1/(1-\phi_1^2)}\right)\right]\Big[
 (1-p^\star_{t+1})\theta_{t+1}-p^\star_{t+1}(1-\theta_{t+1})
  \Big]\label{first_summand}\\
  &\quad-p^\star_{t+1}\phi_1\mathrm{sign}(\beta)\left(\frac{\beta_{t+1}-\mu_{t+1}}{\lambda_1}\right).\label{second_summand}
  \end{align}
  The retrospective term  synthesizes information from both $(\beta_{t+1},\beta_t)$ to contribute to shrinkage at time $t$.
When $(\beta_{t+1},\beta_t)$ are both large, we obtain $p^\star_t(\beta_{t+1})$ and $\theta_{t+1}$ that are both close to one. The shrinkage  is then driven by the second summand in  \eqref{second_summand}, forcing $\beta_t$ to be shrunk towards the future value $\beta_{t+1}$ (through $\mu_{t+1}=\phi_0+\phi_1(\beta_t-\phi_0)$). When either $\beta_{t+1}$ or $\beta_t$ are small, shrinkage is  targeted towards the stationary mean through the dominant term \eqref{first_summand}.

\subsection{Dynamic Spike-and-Slab Fused LASSO}\label{sec:mode}
As we now show, the Laplace spike has the advantage of shrinking coefficient directly to zero, where no additional thresholding is needed for variable selection \citep{rockova15}. 
This has beneficial consequences for computation, where calculations can be narrowed down to active sets of coefficients. In this section, we develop a dynamic coordinate-wise strategy, building on the Spike-and-Slab LASSO method of \cite{SSL}   for static high-dimensional variable selection.

The key to our approach will be drawing upon the penalized likelihood perspective developed in Section \ref{sec:pen}. 
%\subsection{The One-dimensional Case}\label{sec:dyn_mean}
To illustrate the functionality of the dynamic penalty  from Section \ref{sec:pen}, we start by assuming $p=1$ and $x_t=1$ in \eqref{model}. This simple case corresponds to a sparse normal-means model, where the means are dynamically intertwined.  We begin by characterizing some basic properties of the {\em conditional} posterior mode
$$
\wh{\b}=\arg\max_{\b}\pi(\b\C\y,\bm v),
$$
given the variances $\bm v=(v_1,\dots,v_T)'$, where $\y=(y_1,\dots,y_T)'$ arises from \eqref{model} and $\b=(\beta_1,\dots,\beta_T)'$ is assigned the $DSS$ prior. One of the attractive features of the Laplace spike in  \eqref{betas} is that $\wh{\b}$ has a thresholding property.  This property is revealed from necessary characterizations for each $\wh{\beta}_t$ (for $t=1,\dots,T$), once we condition on the rest of the directions through $(\wh{\beta}_{t-1},\wh{\beta}_{t+1})$. 
The {\sl conditional} thresholding rule  can be characterized using standard arguments, as with similar existing regularizers \citep{zhang_mcp,fanli,antoniadis_fan,zhang_zhang,SSL}.  While the typical sparsity-inducing penalty functions are symmetric,  the penalty \eqref{penalty_total} is not, due to its dependence on the previous and future values $(\beta_{t-1},\beta_{t+1})$. Thereby, instead of a single selection threshold, we  have two: 
{\begin{align}
\Delta^-(x,\beta_{t-1},\beta_{t+1})&=\sup_{\beta<0}\left\{\frac{\beta x^2}{2}-\frac{v_t\,Pen(\beta\C\beta_{t-1},\beta_{t+1})}{\beta}\right\}\label{left}\\
\Delta^+(x,\beta_{t-1},\beta_{t+1})&=\inf_{\beta>0}\left\{\frac{\beta x^2}{2}-\frac{v_t\,Pen(\beta\C\beta_{t-1},\beta_{t+1})}{\beta}\right\}\label{right}.
\end{align}}

 The following necessary characterization  links  the behavior of $\wh{\b}$ to the shrinkage terms characterized in Section \ref{sec:prosp} and Section \ref{sec:retro}.
\begin{lemma}\label{lemma2}
Denote by $\wh{\b}=(\wh{\beta}_1,\dots,\wh{\beta}_T)'$  the global mode of $\pi(\b_{1:T}\C\y_{1:T},\bm v_{1:T})$ and by $\Delta_{t}^-$ and $\Delta_{t}^-$ the selection thresholds \eqref{left} and \eqref{right} with $x=1$, $\beta_{t-1}=\wh{\beta}_{t-1}$ and $\beta_{t+1}=\wh{\beta}_{t+1}$.
Then, conditionally on $(\wh{\beta}_{t-1},\wh{\beta}_{t+1})$, we have for $1< t< T$
\begin{equation}
\wh{\beta}_{t}=\begin{cases} 
0&\quad\text{if}\quad \Delta_{t}^- <y_t<\Delta_{t}^+\\
\left[|y_t|-v_t\Lambda^\star(\wh{\beta}_t\C\wh{\beta}_{t-1},\wh{\beta}_{t+1})\right]_+\mathrm{sign}(y_t)&\quad\text{otherwise},
\label{nonzero2}\end{cases}
\end{equation}
where $\Lambda^\star(\wh{\beta}_t\C\wh{\beta}_{t-1},\wh{\beta}_{t+1})$ was defined in \eqref{big_lambda}.
\end{lemma}
\begin{proof}
We  begin by noting that $\wh{\beta}_{t}$ is a maximizer in $t^{th}$ direction while keeping  $(\wh{\beta}_{t-1},\wh{\beta}_{t+1})$ fixed, i.e.
\begin{equation}\label{one_site1}
\wh{\beta}_{t}=\arg\max_{\beta}\left\{-\frac{1}{2v_t}(y_{t}-\beta)^2+Pen(\beta\C\wh{\beta}_{t-1},\wh{\beta}_{t+1})\right\}.
\end{equation}
 It turns out that $\wh{\beta}_{t}=0$ iff 
$
\beta\left(y_t-\frac{\beta}{2}+v_t\frac{Pen(\beta\C\wh{\beta}_{t-1},\wh{\beta}_{t+1})}{\beta}\right)<0,\, \forall\beta\in\mathbb{R}\backslash\{0\}
$ \citep{zhang_zhang}. The rest of the proof follows from the definition of $\Delta_{t}^+$ and $\Delta_{t}^-$ in \eqref{left} and \eqref{right}. Conditionally on $(\wh{\beta}_{t-1},\wh{\beta}_{t+1})$,  the global mode $\wh{\beta}_t$, once nonzero, has to satisfy \eqref{nonzero2} from the first-order necessary condition.
\end{proof}
Lemma \ref{lemma2} formally certifies that the posterior  mode {\em under the Laplace spike} exhibits  both (a) sparsity and (b) smoothness (through the prospective/retrospective shrinkage terms). 
\begin{remark}
While  Lemma \ref{lemma2} assumes $1<t<T$, the characterization applies also for $t=1$, once we specify the initial condition $\beta_{t=0}$.  The value $\beta_{t=0}$ is not assumed   known and will be  estimated together with all the remaining parameters . For $t=T$, an analogous characterization exists, where the shrinkage term and the selection threshold  only contain the prospective portion of the penalty. 
\end{remark}

When $p>1$, there is a delicate interplay between the multiple series, where overfitting in one direction may impair recovery in other directions.
As will be seen in Section~\ref{sec:simul}, anchoring on sparsity is a viable remedy to these issues. We obtain analogous characterizations of the global mode.
We will denote with $\Delta_{tj}^-$ and $\Delta_{tj}^-$ the selection thresholds \eqref{left} and \eqref{right} with $x=x_{tj},\beta_{t-1}=\wh{\beta}_{t-1j}$, and
$\beta_{t+1}=\wh{\beta}_{t+1j}$.

{
\begin{lemma}\label{thm2}
Denote by $\wh{\B}=\{\wh{\beta}_{tj}\}_{t,j=1}^{T,p}$  the global mode of $\pi(\b_{1:T}\C\y_{1:T},\bm v_{1:T})$ and $\wh{\B}_{\sls tj}$ all but the $(t,j)^{th}$ entry in $\wh{\B}$. Let
$z_{tj}=y_t-\sum_{i\neq j}x_{ti}\wh{\beta}_{ti}$ and $Z_{tj}=x_{tj}{z_{tj}}$. Then $\wh{\beta}_{tj}$ satisfies the following necessary condition
$$
\wh{\beta}_{tj}=
\begin{cases}
\frac{1}{x_{tj}^2}\left[|Z_{tj}|-v_t\Lambda^\star(\wh{\beta}_{tj}\C\wh{\beta}_{t-1j},\wh{\beta}_{t-1j})\right]_+\mathrm{sign}(Z_{tj})&\quad\quad \text{otherwise}\\
0&\text{if}\quad \Delta_{tj}^-<Z_{tj}<\Delta^+_{tj}.
\end{cases}
$$
\end{lemma}
\proof
Follows from Lemma \ref{lemma2}, noting 
 that $\wh{\beta}_{tj}$ is a maximizer in $(t,j)^{th}$ direction while keeping  $\wh{\B}_{\sls tj}$ fixed, i.e.
\begin{equation}\label{one_site2}
\wh{\beta}_{tj}=\arg\max_{\beta}\left\{-\frac{1}{2v_t}(z_{tj}-x_{tj}\beta)^2+Pen(\beta\C\wh{\beta}_{t-1j},\wh{\beta}_{t+1j})\right\}.\qquad\quad\quad\quad\square
\end{equation}
}
Lemma \ref{thm2} evokes coordinate-wise optimization for obtaining the posterior mode. However,  the computation of selection thresholds $(\Delta^-_{tj},\Delta^+_{tj})$ (as well as the one-site maximizers \eqref{one_site1}) requires numerical optimization. The lack of availability of closed-form thresholding  hampers practicality when  $T$ and $p$ are even moderately large. In the next section, we propose an alternative strategy which capitalizes on closed-form thresholding rules.

%The coordinate-wise strategy naturally extends to the case when $p>1$, where the current update $\wh{\beta}_{tj}^{(m)}$ depends not only on the previous and future values $(\wh{\beta}_{t-1j}^{(m)},\wh{\beta}_{t+1j}^{(m)})$, but also other values of the regression coefficients $\wh{\beta}_{ti}^{(m)}, i\neq j$ at time $t$. The updates can be organized into within time cycles with an outer loop over time.

%Another interesting strategy would be splitting the updates into forward steps, deploying only the prospective penalty, and then backward step with the retrospective penalty. There is a lot to explore here.

A (local) posterior mode $\wh\b_{0:T}$ can be obtained either directly, by cycling over one-site updates \eqref{one_site2}, or indirectly through an EMVS algorithm outlined in the previous section. The direct
algorithm consists of integrating out $\bg_{0:T}$ and solving a sequence of non-standard optimization problems \eqref{one_site2}, which necessitate numerical optimization. The EMVS algorithm, on the other hand, obviates the need for numerical optimization by offering closed form one-site updates. The E-step is very similar to the Gaussian case. The expected previsions $\v_t^\star$ can be calculated as before. In the calculation of $p^\star_{tj}$ and $\theta_{tj}$, we now have to replace the Laplace spike density. For updating $\b_{0:T}$,
we proceed coordinate-wise, iterating over the following single-site updates while keeping all the remaining parameters fixed.  For $1<t<T$, we have
$$
\beta_{tj}^{(m+1)}=\arg\max_\beta Q_{tj}(\beta),
$$
where
\begin{align}
Q_{tj}(\beta)=&-\frac{\nu^\star_t}{2}(z_{tj}-x_{tj}\beta)^2-\frac{p^\star_{tj}}{2\lambda_1}(\beta-\phi_1\beta_{t-1j}^{(m)})^2-
\frac{p^\star_{t+1j}}{2\lambda_1}(\beta_{t+1j}^{(m)}-\phi_1\beta)^2\notag\\
&-(1-p^\star_{tj})\lambda_0|\beta|+p^\star_{t+1j}\log\theta_{t+1j}+(1-p^\star_{t+1j})\log(1-\theta_{t+1j})\label{q_tj},
\end{align}
and where $z_{tj}=y_t-\sum_{i\neq j}x_{ti}\beta_{ti}^{(m)}$. %The solution $\beta_{tj}^{(m+1)}$ could be obtained with an elastic net-like update if the last two summands in \eqref{q_tj} were absent.
From the first-order condition, the solution $\beta_{tj}^{(m+1)}$, if nonzero, needs to satisfy  $\partial Q_{tj}(\beta)/\partial\beta\big|_{\beta=\beta_{tj}^{(m+1)}}=0$. 
To write the derivative slightly more concisely, we introduce the following notation: 
$$
Z_{tj}=\nu^\star_t x_{tj}z_{tj}+\frac{p^\star_{tj}\phi_1}{\lambda_1}\beta_{t-1j}^{(m+1)}+\frac{p^\star_{t+1j}\phi_1}{\lambda_1}\beta_{t+1j}^{(m+1)}
\quad\text{and}
\quad W_{tj}=\left(\nu^\star_t{x_{tj}^2}+\frac{p^\star_{tj}}{\lambda_1}+\frac{p^\star_{t+1j}\phi_1^2}{\lambda_1}\right).
$$
Then  we can write for $\beta\neq0$
\begin{align}\label{derivative}
\frac{\partial Q_{tj}(\beta)}{\partial \beta}=&Z_{tj}-
W_{tj}\beta
-(1-p^\star_{tj})\lambda_0\,\mathrm{sign}(\beta)+\frac{\partial \theta_{t+1j}}{\partial\beta}\left[\frac{p^\star_{t+1j}}{\theta_{t+1j}}-\frac{1-p^\star_{t+1j}}{1-\theta_{t+1j}}\right],
\end{align}
where 
$$
\frac{\partial\theta_{t+1j}}{\partial \beta}=\theta_{t+1j}(1-\theta_{t+1j})\left[\lambda_0\,\mathrm{sign}(\beta)-\frac{\beta(1-\phi_1^2)}{\lambda_1}\right]
$$
is obtained from \eqref{theta_der}.
 Recall that $\theta_{t+1j}$, defined  in \eqref{weights}, depends on $\beta_{tj}$ (denoted by $\beta$ above). This complicates the tractability of the M-step.  If $\theta_{t+1j}$ was fixed, we could obtain a simple closed-form solution $\beta_{tj}^{(m+1)}$  through an elastic-net-like update \citep{zou_hastie}.  We can take advantage of this fact with a one-step-late (OSL) adaptation of the EM algorithm \citep{green_OSL}. The OSL EM algorithm bypasses intricate  M-steps by evaluating the intractable portions of the penalty derivative   at the most recent value, rather than the new value. We apply this trick to the last summand in \eqref{derivative}.
Instead of treating $\theta_{t+1j}$ as a function of $\beta$ in 
\eqref{derivative}, we fix it at the most recent value $\beta_{tj}^{(m)}$. The solution for $\beta$, implied by \eqref{derivative}, is then (when $\Lambda_{tj}>0)$
\begin{equation}\label{beta_update}
\beta^{(m+1)}_{tj}=\frac{1}{W_{tj}+(1-\phi_1^2)/\lambda_1M_{tj}}\left[|Z_{tj}|-\Lambda_{tj}\right]_+\mathrm{sign}(Z_{tj}),\quad\text{for}\quad 1<t<T,
\end{equation}
where $M_{tj}=p^\star_{t+1j}(1-\theta_{t+1j})-\theta_{t+1j}(1-p^\star_{t+1j})$ and $\Lambda_{tj}=\lambda_0[(1-p^\star_{tj})-M_{tj}]$. The update \eqref{beta_update} is a  thresholding rule, with a shrinkage term that reflects the size of $(\beta_{t-1j}^{(m)},\beta_{tj}^{(m)},\beta_{t+1j}^{(m)})$. The exact thresholding property is obtained  from sub-differential calculus, because $Q_{tj}(\cdot)$ is not differentiable at zero (due to the Laplace spike). A very similar update is obtained also for $t=T$, where all the terms involving $p^\star_{t+1j}$ and $\theta_{t+1j}$  in $\Lambda_{tj}, W_{tj}$ and $Z_{tj}$ disappear. For $t=0$, we have
\begin{equation}\label{beta0_update}
\beta_{0j}^{(m+1)}=\frac{1}{p^\star_{1j}\phi_1^2+p^\star_{0j}(1-\phi_1^2)}\left[p^\star_{0j}|\beta_{1j}|\phi_1-(1-p^\star_{0j})\lambda_0\lambda_1\right]_+\mathrm{sign}(\beta_{1j}).
\end{equation}
The updates \eqref{beta_update} and \eqref{beta0_update} can be either cycled-over at each M-step, or performed just once for each M-step. 

{It is straightforward to implement  a random-walk variant of this procedure with $\theta_{tj}=\Theta$  by setting $M_{tj}=0$ in \eqref{beta_update}.}

%In the rest of the paper, we confine attention to the Laplace spike and the Gaussian slab.

%\section{MCMC implementation}

%A similar algorithm can be also derived for forward filtering, using only the prospective part of the penalty.
\begin{remark}
For autoregression with a higher order $h>1$, the retrospective penalty would be similar where $\mu_t$ would depend on $h$ lagged values. The prospective penalty would consist of  not just one, but $h$ terms. For the EM implementation, one would proceed analogously by evaluating the  derivatives of $\theta_{t+1},\dots,\theta_{t+h}$ w.r.t. $\beta$  at the most recent update of the process from the previous iteration and keep them fixed for each one-site update.  
\end{remark}

To illustrate  the ability of  the $DSS$ priors to suppress noise and recover true signal, we consider a high-dimensional synthetic dataset and a topical macroeconomic dataset. 

\section{Synthetic High-Dimensional Data}\label{sec:simul}
We first illustrate our dynamic variable selection procedure on a simulated example with $T=100$ observations generated from the model  \eqref{model} with $p=50$ predictors and with $v_t=0.25$. 
The predictor values $x_{tj}$ are  obtained independently from a standard normal distribution. Out of the 
 $50$ predictors, $46$ never contribute to the model  (predictors $x_{t5}$ through $x_{t50}$), where $\beta_{t5}^0=\beta_{t6}^0=...=\beta_{t50}^0=0$ at all times. The predictor $x_{t1}$ is a persisting predictor, where $\{\beta_{t1}\}_{t=1}^T$ is generated according to an $AR(1)$ process \eqref{slab_process}  with $\phi_0=0$ and $\phi_1=0.98$ and where $|\beta_{t1}^0|>0.5$. 
The remaining three predictors are allowed to enter and leave the model as time progresses. The regression coefficients  $\{\beta^0_{t2}\}_{t=1}^T,\{\beta^0_{t3}\}_{t=1}^T$ and $\{\beta^0_{t4}\}_{t=1}^T$ are again generated from an $AR(1)$ process ($\phi_0=0$ and $\phi_1=0.98$). However, the values are rescaled and thresholded to zero whenever the absolute value of the process drops below $0.5$, creating zero-valued periods. The true sparse series of coefficients are depicted in Figure \ref{fig_simul} (black lines).

We begin with the standard DLM approach, which is equivalent to  $DSS$ when the selection indicators are switched on at all times, i.e.,  $\gamma_{tj}=1$ for $t=0,\dots, T$ and $j=1,\dots,p$. 
This is equivalent to setting $\Theta=1$ in our prior. %The $ASSP$ construction then boils down to the familiar $AR(1)$ Gaussian state space model. %We compute the mode of the joint posterior $\pi(\b_0,\B|\Y)$ through standard ``coordinate-wise" updates of entire vectors $\b_t$. 
{The autoregressive parameter $\phi_1$ is assigned the prior \eqref{beta_prior} and estimated. We also estimate the variances $v_t$ using the discount stochastic volatility  model \eqref{eq:dsvt} with $\delta=0.9$ and $n_0=d_0=10$.  
Plots of the estimated posterior mode trajectories of the
first $6$ series (including the $4$ active ones) are in Figure \ref{fig_simul} (red broken lines).  With the absence of the spike, the estimated series of coefficients cannot achieve sparsity. 
%Unfortunately, the failure to detect zero directions  affects the estimates of the nonzero directions that try to overcompensate for the overfitting.  
By failing to discern the coefficients as active or inactive, the state process confuses the source of the signal, distributing it across the redundant covariates. 
This results in loss of efficiency and  poor recovery.
 %The sum of squared differences between the estimate and the truth equals $SE(\wh{\B},\B_0)\equiv\sum_{t,j}[\wh{\beta}_{tj}-\beta^0_{tj}]^2= 44.56$\footnote{this number needs to be recomputed}. 

 \begin{figure}[t!]
\begin{center}
\includegraphics[width=15cm,height=8cm]{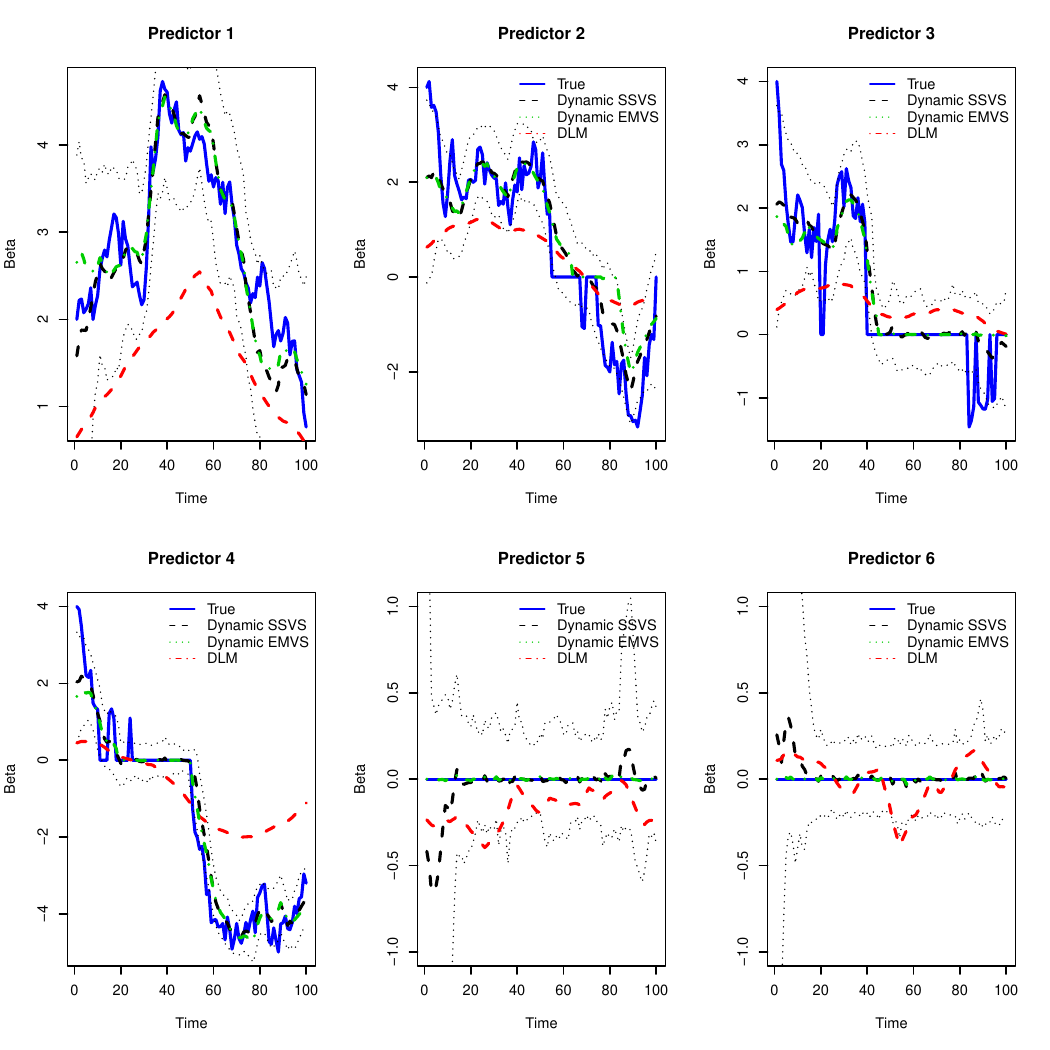}
\end{center}
\caption{The first six regression coefficients of the true (blue solid lines) and estimated  regression coefficients in  the simulated example with $p=50$. The estimates are posterior means from Dynamic SSVS (black broken line) and modes from Dynamic EMVS (green broken line). Comparisons are made with DLM (red dotted line). The black dotted lines denote pointwise credible intervals. }\label{fig_simul}
\end{figure}

 \begin{figure}[htbp!]
\begin{center}
\includegraphics[width=15cm,height=8cm]{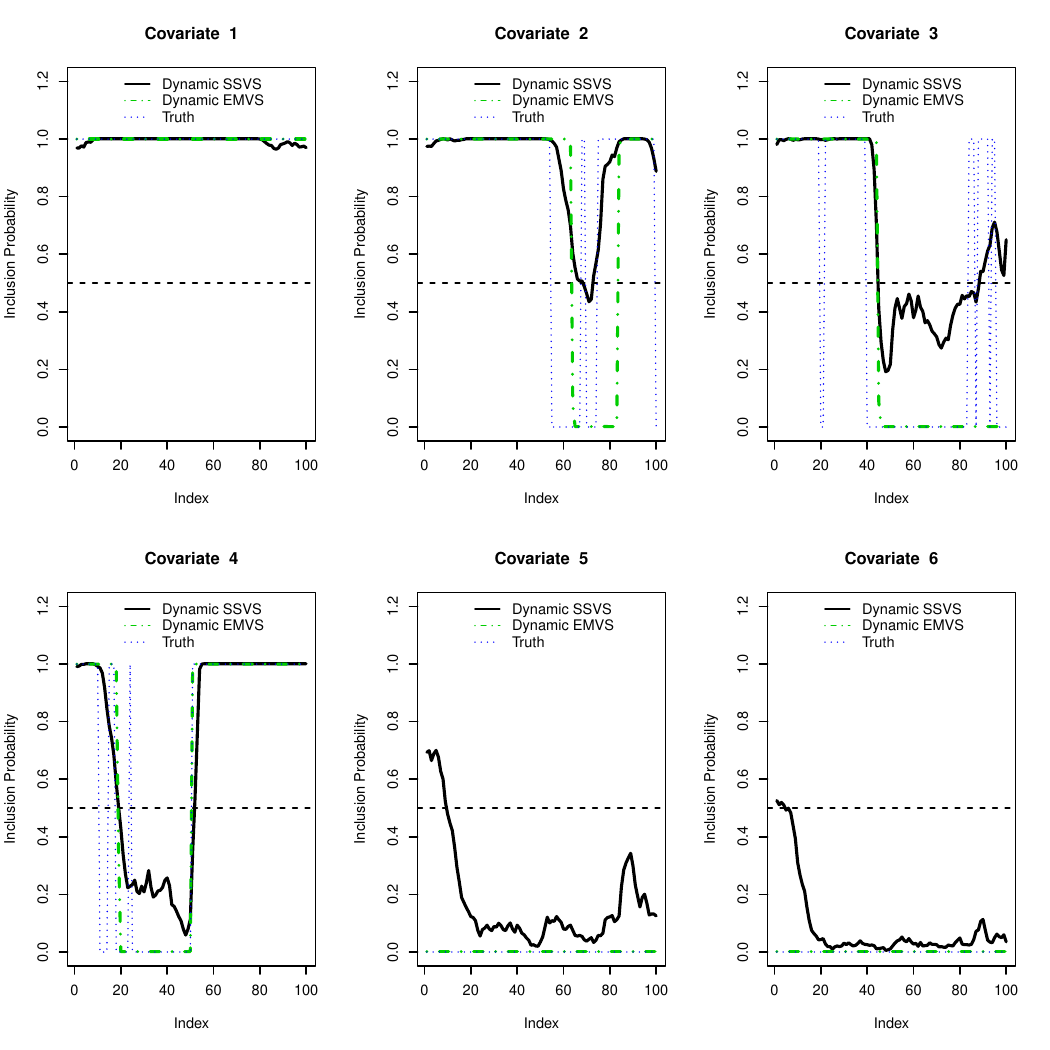}
\end{center}
\caption{Posterior inclusion probabilities  $\P(\gamma_{tj}=1\C\y_{1:T})$ (Dynamic MCMC) and conditional inclusion probabilities  $\P(\gamma_{tj}=1\C\wh\beta_{tj}, \y_{1:T})$ (Dynamic EMVS) and true pattern of sparsity for the first six series.}\label{fig_theta}
\end{figure}

With the hope to improve on this recovery, we deploy the $DSS$ process with a sparsity inducing spike. First, we apply Dynamic SSVS  with $1\,000$ iterations and $200$ burnin time. We set the spike and slab parameters $\Theta=0.1,\lambda_1=0.1$ and $\lambda_0=0.01$ so that the ratio between spike and slab variances is sufficiently large \citep{GM93}.
 The autoregressive parameter $\phi_1$ is estimated under the prior \eqref{beta_prior} and the stochastic volatilities are also estimated with $\delta=0.9$ and $n_0=d_0=10$.
 We plot the posterior mean of the regression coefficients in Figure \ref{fig_simul} (black broken line) together with the credible sets (black dotted lines).
 The recovered series have a strikingly different pattern compared to the non-sparse DLM solution (red dotted lines). First, the estimated series is seen to track closely the periods of predictor importance/irrelevance, achieving dynamic variable selection. Second, by harnessing sparsity, the $DSS$ priors alleviate bias in the nonzero directions, outputting a  cleaner representation of the true underlying signal. The posterior mean of the autoregressive parameter $\phi_1$ is $0.94$. In addition, we  plot the posterior inclusion probabilities $\P(\gamma_{tj}=1\C \y_{1:T})$ for the first $6$ predictors (Figure \ref{fig_theta}, black lines). These quantities can be used to guide variable selection by focusing on those coefficients whose inclusion probability is at least $0.5$ \citep{barbieri}. Indeed, we can see that these estimated probabilities drop below $0.5$ when the true signal is absent, effectively recovering the ``pockets of predictability". The posterior mean of the coefficient $\phi_1$ was estimated at $0.981$ (very close to the true value $0.98$) with posterior a credible interval $(0.973,0.989)$.  The computation  took $151.8$ seconds in R.  We will now turn to Dynamic EMVS to see whether similarly successful recovery can be achieved with less time. 

We apply Dynamic EMVS considering the same spike and slab hyper-parameters, i.e.  $\lambda_1=0.1$ and $\lambda_0=0.01$. The global sparsity weight $\Theta$ can be regarded as a tempering parameter, where $\Theta=1$ corresponds to the DLM case. By choosing smaller values $\Theta$, the posterior becomes more multi-modal making it easier for the EM to get trapped. Since the EMVS computation is very fast, we can alleviate local entrapments by applying a deterministic annealing strategy,  similar to the one suggested in \cite{RG14}. We will consider not only one value $\Theta=0.1$, but a whole sequence of decaying values $\Theta\in\{1,0.9,0.5,0.1\}$ with warm starts. 
Namely, the output obtained with a larger value $\Theta$ will be used as an initialization for the computation at the next smaller value $\Theta$ in the chosen sequence.
 In this way, we obtain an entire solution path (not only one single solution), we accelerate convergence and  increase the chances for the EM to find a promising mode. 
  We successfully apply this  strategy for $\Theta\in\{1,0.9,0.5,0.1\}$ and, similarly as before,  we estimate $\phi_1$ and all the variances $v_t$ under the same priors. The  estimated regression coefficients obtained with $\Theta=0.1$ are depicted in Figure \ref{fig_simul} (green broken lines). We can again see dramatic improvements over DLM (obtained with $\Theta=1$) and, interestingly, a very similar recovery to the posterior mean with Dynamic SSVS.
    The R computations  took $15$ seconds for $\Theta=0.9$,  $6$ seconds for $\Theta=0.5$ and $8$ seconds for $\Theta=0.1$, yielding nontrivial computational dividends compared to MCMC ($151.8$ s). In addition,  Dynamic EMVS outputs conditional inclusion probabilities $\P[\gamma_{tj}=1\C\wh\beta_{tj},\y_{1:T}]$ which can be regarded as the conditional counterpart to the marginal posterior inclusion probabilities $\P[\gamma_{tj}=1\C \y_{1:T}]$ estimated from MCMC. As can be seen from Figure  \ref{fig_theta}, these conditional probabilities track closely the marginal  ones  and, again,  drop below $0.5$ when the true signal is not present. 
  These companion plots are helpful visualizations of the time-varying sparsity profile. In conclusion, the plots for Dynamic SSVS and Dynamic EMVS largely agree.

 \begin{figure}[t!]
\begin{center}
\includegraphics[width=15cm,height=8cm]{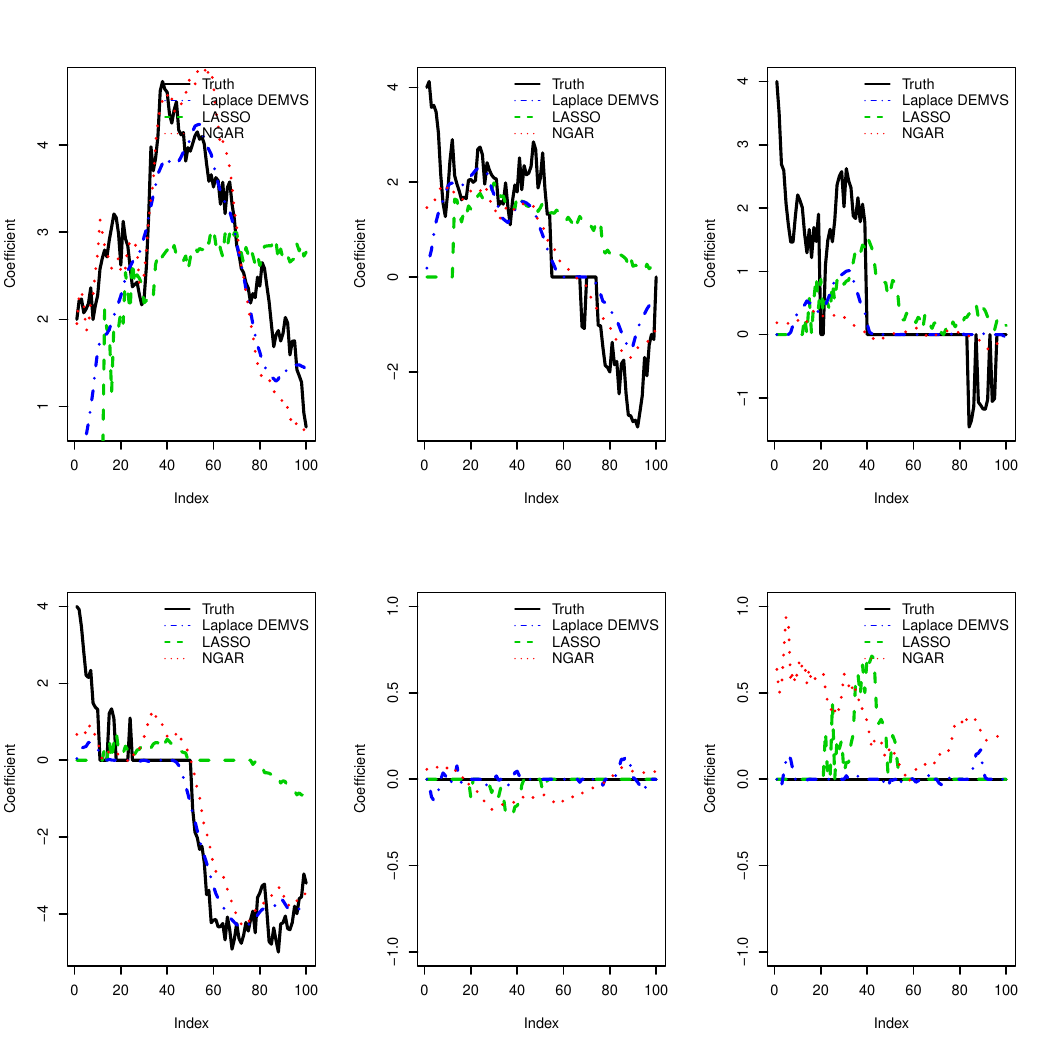}
\end{center}
\caption{The first six regression coefficients of the true and estimated  regression coefficients in  the high-dimensional simulated example with $p=50$. 
We compare DSS (Laplace version) with NGAR and LASSO.}\label{fig_demvs}
\end{figure}

Next, we deploy the Laplace spike variant of DEMVS with a random walk  slab prior and with $\lambda_1=0.1,\lambda_0=1$ and $\Theta=0.5$. %The remaining tuning parameters are chosen as in Section \ref{sec:dyn_mean}, i.e.  $\lambda_0=1$ and $\Theta=0.9$.
 This hyper-parameter choice corresponds to a very mild separation between the  spike and slab distributions.
 We apply the one-step-late EM algorithm outlined in Section \ref{sec:mode}, initializing the calculation with the output from DLM. We assume that the initial vector $\b_{t=0}$ is drawn from a $\Theta$-weighted mixture distribution between the Laplace density and a zero-mean Gaussian density with variance one and  and we estimate it together with all the other parameters, as prescribed in Section \ref{sec:mode}.

We also compare the performance to the NGAR process of \cite{kalli_griffin} and the LASSO method. {The latter} does not take into account the temporal nature of the problem.
For NGAR, we use the default settings, $b^*=s^*=0.1$, with 1,000 burn-in and 2,000 MCMC iterations.
For LASSO, we sequentially run a static regression in an extending window fashion, {where the LASSO regression is refit using $1{:}t$ for each $t=1{:}T$ to produce a series of quasi-dynamic coefficients; a common practice for using static shrinkage methods for time series data \citep{bai2008forecasting,de2008forecasting,stock2012generalized,li2014forecasting}}, choosing $\lambda$ via 10-fold cross-validation. The estimated trajectories are depicted in Figure \ref{fig_demvs}.

For the first series, the only persistent series,  both $DSS$ (Laplace) and NGAR  succeeds well in tracing the true signal.
This is especially true in contrast to DLM and LASSO, which significantly under-valuate the signal. 
The estimated coefficient evolutions for DLM and LASSO become inconclusive for assessing variable importance,  where the coefficient estimates for the relevant variables have been polluted by the elevated estimates for the irrelevant variables. For the second to fourth series with intermittent zeros, we see that $DSS$ and NGAR are able to separate the true zero/nonzero signal (noted by the flat coefficient estimates during inactive periods). 
The LASSO method produces sparse estimates, however the variable selection is not linked over time and thereby erratic.
For the two zero series (series five and six), both $DSS$ and LASSO truly shrink noise to zero.
The $DSS$ priors mitigate overfitting by eliminating noisy coefficients and thereby leaving enough room for the true predictors to capture the trend.  
%As a revealing byproduct, we also obtain the evolving mixing weights determining the relevance of each coefficient at each time. The evolutions of rescaled weights (so that  $\theta(0)=0$) are plotted in Figure \ref{fig_theta}. These companion plots are helpful visualizations of the time-varying sparsity profile.

%For instance, $\{\beta_{3t}\}_{t=1}^T$ and $\{\beta_{3t}\}_{t=1}^T$
%This curse of overfitting is exacerbated when $p$ gets even larger.

We repeat this experiment $10$ times, generating different responses and regressors using the same set of coefficients. We compare the sum of squared error (SSE)  between the recovered estimates $\wh{\B}$ and the true series $\B_0$ as well as the  Hamming distance between the true and estimated sparsity patterns.
For the MCMC version, the sparsity pattern will be estimated from the matrix of posterior inclusions $\Pi=(\pi_{tj})$ where $\pi_{tj}\equiv\P[\gamma_{tj}=1\C\y_{1:T}]$ according to the median probability model rule. 
We then  define the Hamming distance as
$$
\text{Ham}(\Pi, \B_0)=\sum_{j=1}^p\sum_{t=1}^T|\mathbb{I}(\pi_{tj}>0.5)-\mathbb{I}(\beta_{tj}^0\neq 0)|.
$$
For Dynamic EMVS (with Gaussian spike), we can   use the conditional inclusion probabilities instead of $\pi_{tj}$. 
Alternatively, one can obtain sparsity patterns by thresholding out coefficients whose magnitude is smaller than the intersection point between the stationary spike and slab densities (as we explained in Section \ref{sec:weights}).

 Table \ref{tab_synth} reports average performance metrics over the $10$ experiments.
The performance of $DSS$ is compared to the full DLM model \citep{WestHarrison1997book2}, NGAR \citep{kalli_griffin} and LASSO.
%For the DLM, we assume the state dynamics and observation innovations are known  (i.e. $\phi_0=0, \phi_1=0.98,\nu_{1:T}=\sqrt{0.25}$\footnote{what is $\nu$?}) with standard normal priors\footnote{I thought the variance was $10(1-\phi_1^2)$} on the coefficients.
For DLM and NGAR, we use the same specifications as above.
%\begin{figure}[!t]
%\begin{center}
%\scalebox{0.3}{\includegraphics{pics/beta1}}\scalebox{0.3}{\includegraphics{pics/beta2}}\scalebox{0.3}{\includegraphics{pics/beta3}}
%\scalebox{0.3}{\includegraphics{pics/beta4}}\scalebox{0.3}{\includegraphics{pics/beta5}}\scalebox{0.3}{\includegraphics{pics/beta6}}
%\end{center}
%\caption{True and recovered time series of regression coefficients from the low-dimensional simulated example with $p=6$}\label{fig_simul}
%\end{figure}
For $DSS$, we now explore a multitude of combinations of hyper-parameters using both Dynamic SSVS and Dynamic EMVS.
For the Gaussian spike, we choose $\lambda_1\in\{0.1,0.01\}$, $\lambda_0\in\{0.01,0.001\}$, and $\Theta\in\{0.9,0.5,0.1\}$.  We also consider a random-walk prior variant with $\theta_{tj}=\Theta$. For Dynamic EMVS, we initialize the calculations at DLM solutions apart from the settings marked with a star, where we use warm starts (as explained above).
%We also track the average estimation error for $\phi_1$ and average MSE for the residual variances $v_t$ (for the Gaussian spike priors).
 We focus  on the Gaussian EMVS variant, where additional simulations for the Laplace spike are reported in Table \ref{tab_aux1} in the Appendix.

\begin{table}[!t]
\centering
\caption{\small Performance evaluation of $DSS$, LASSO, NGAR and DLM on the simulated example with $p=50$. The results are split for the signal parameters ($x_{1:4}$) and noise parameters ($x_{5:50}$). Dynamic SSVS uses $1\,000$ iterations with $100$ burnin. RW stands for a random-walk variant with $\theta_{tj}=\Theta$. EMVS calculations are initialized at DLM solutions besides settings denoted with  $^\star$where we use warm starts. }
\vspace{0.5cm}
\label{tab_synth}
\scalebox{0.8}{\begin{tabular}{l l l  r  rr   rr  rr }
\hline\hline
                    &      &  &    & \multicolumn{2}{c}{$x_{1:50}$}   & \multicolumn{2}{c}{$x_{1:4}$}   & \multicolumn{2}{c}{$x_{5:50}$} \\
$p=50$       &       &  & Time (s) &   SSE         & Ham.        & SSE         & Ham.       & SSE         & Ham.      \\ \hline
NGAR         &&    &     564.2    &539.2&4708&320.9&108&218.3&4600  \\
LASSO            &&  & 9.3&1621.8&281.2&1595.1&186.4&26.7&94.8      \\
\multicolumn{3}{l}{Dynamic SSVS (Gaussian) }    &                        &   &                        &   &            &             \\
\hline
$\lambda_1=.1$ && $\Theta=1$ (DLM) 			& 164.9&1625.4&4708&1426.8&108&198.5&4600\\
$\lambda_1=.1$ & $\lambda_0=.01$ & $\Theta=.5$  & 132.6&1541.6&2447.7&1457.2&108&84.4&2339.7\\
$\lambda_1=.1$ & $\lambda_0=.01$ &  $\Theta=.1$ & 118.6&108.7&51.4&104.3&40.5&4.5&10.9\\
$\lambda_1=.1$ & $\lambda_0=.01$ & $\Theta=.1$ (RW) & 113.6&683.2&86.5&660.6&63.6&22.6&22.9     \\
$\lambda_1=.01$ &$\lambda_0=.001$&$\Theta=.1$  &117.9&177.6&431&141.2&105.2&36.4&325.8 \\
\hline
\multicolumn{3}{l}{Dynamic EMVS (Gaussian) }    &            &             &   &           &             &                           \\
\hline
$\lambda_1=.1$  & &$\Theta=1$ (DLM) & 5.1&1241&4708&975.2&108&265.7&4600   \\
$\lambda_1=.1$ &$\lambda_0=.01$ &$\Theta=.9^\star$  & 16.5&286.3&106.3&241.9&54.4&44.3&51.9    \\
$\lambda_1=.1$ &$\lambda_0=.01$ &$\Theta=.5^\star$  & 12.4&294.7&99.6&254.7&56&40&43.6   \\
$\lambda_1=.1$ &$\lambda_0=.01$ &$\Theta=.1^\star$  & 14.3&309.6&93&267.5&59.2&42.1&33.8      \\
$\lambda_1=.1$ &$\lambda_0=.01$ &$\Theta=.1$ (RW)  &8&495.8&91.8&494.1&88.2&1.7&3.6   \\
\hline\hline
\end{tabular}}
\end{table}

Looking at Table~\ref{tab_synth}, $DSS$ performs better in terms of both SSE and Hamming distance compared to DLM, NGAR, and LASSO for the majority of the hyperparameters considered. To gain more insights, the table is divided into three blocks: overall performance on $\beta_{1:50}$, active coefficients $\beta_{1:4}$ and noise coefficients $\beta_{5:50}$.
%The Hamming distance is reported in percentages. 
Because DLM and NGAR only shrink (and do not select), the Hamming distance for the block of noisy coefficients is $100\%$. 
For Dynamic SSVS, the sum of squared errors and the Hamming distance is seen to increase with $\Theta$.  It is interesting to note the difference in performance between our stationary $DSS$ version, where $\theta_{tj}$ are dynamically evolving, and the random-walk (RW) version, where $\theta_{tj}=\Theta$.  In this stationary situation, these is a clear advantage in linking the weights over time using the deterministic construction \eqref{weights}.
We found the settings $\lambda_1=0.1,\lambda_0=0.01$ and $\Theta=0.1$ to work well on this example, where the threshold of practical significance (i.e. the intersection point between the stationary spike and slab densities as discussed in Section \ref{sec:weights}) equals $0.086$.  Decreasing this threshold to $0.05$ with a sharper spike-and-slab prior ($\lambda_1=0.01,\lambda_0=0.001$ and $\Theta=0.1$), many more false discoveries  occur (i.e. increased Hamming distance for the noise coefficients) due to the fact that even very small noisy effects can be assigned to the slab distribution. 
%When the two spike and slab stationary densities overlap a lot,  the posterior inclusion indicators will not be clearly separated by the 0.5 threshold. It might be worthwhile to threshold out directly the posterior means for variable selection.

% The stationary slab variance $\lambda_1/(1-\phi_1^2)$ also affects variable selection, where larger values increase the selection threshold and produce less false discoveries. The reverse is true for the signal coefficients. In terms of SSE, $DSS$ outperforms all other methods in estimating  $\beta_{5:50}$, demonstrating great success in suppressing unwanted parameters. 
%Regarding the choice of $\phi_1$, larger values seem beneficial for  the signal coefficients, where borrowing more strength enhances stability in predictive periods. 

 Dynamic EMVS reconstructs  signal much faster compared to the 1\,000 iterations of Dynamic SSVS. While the MAP trajectory is not as good  in terms of SSE (which is expected from a (local) posterior mode), its  performance is still better than  LASSO, DLM and NGAR. Again, we found the setting  $\lambda_1=0.1,\lambda_0=0.01$ and $\Theta=0.1$ to work  well and we can clearly see dividends of  dynamic weighting relative to the random-walk prior. 
%Note that, we forcefully shrunk the coefficients that contain zero in its 95$\%$ interval for NGAR, which leads to significant decrease in Hamming distance, although this is a post-process that is not within the NGAR framework\footnote{I would suggest that we leave out Hamming distance for NGAR and say that it is not made for selection, just for shrinkage. I think it's better if we leave NGAR as 93.2, but leave a note on how it's not shrinkage.}.
%On the other hand, it is clear that the majority of sets of hyperparameters\footnote{Maybe we should include also one very bad setting where it does not work so well so that we can inform the users when to be careful} outperform the other methods compared, for both MSE and Hamming distance.
%Note that $DSS$ outperforms DLM and LASSO for SSE and NGAR and LASSO for Hamming distance for all combinations considered.
%This demonstrates the robustness of $DSS$ under varying hyperparameters.
Comparing the results with DLM and LASSO, $DSS$ showcases the benefits of combining dynamics and shrinkage, since DLM (only dynamics) and LASSO (only shrinkage) underperform significantly.
%Additionally, we note that $DSS$ is substantially faster, in terms of computation time, compared to NGAR and LASSO, which is appealing for practical applications.
{Regarding timing comparisons with NGAR, we need to point out that  NGAR was run with 2,000  iterations and 1,000 burnin, which we found to be sufficient for obtaining satisfactory results. }

\begin{table}[!t]
\centering
\caption{\small Performance evaluation of $DSS$, LASSO  and DLM on the simulated example with $p=200$. The results are split for the signal parameters ($x_{1:4}$) and noise parameters ($x_{5:50}$). RW stands for a random-walk variant with $\theta_{tj}=\Theta$. EMVS calculations are initialized at previous solutions (warm starts), as designated by the  $^\star$ sign. FD stands for False Discoveries (noise variables which were identified as active at least once), FN stands for False Nondiscoveries (number of true variables which were removed from the model at all time points), DIM is estimated number of covariates identified as active at least once. }
\vspace{0.5cm}
\label{tab_synth2}
\scalebox{0.75}{\begin{tabular}{l l l  r  rr   rr  rr     r r r }
\hline\hline
                    &      &  &    & \multicolumn{2}{c}{$x_{1:50}$}   & \multicolumn{2}{c}{$x_{1:4}$}   & \multicolumn{2}{c}{$x_{5:50}$} & &&\\
$p=200$       &       &  & Time (s) &   SSE         & Ham.        & SSE         & Ham.       & SSE         & Ham.      &  FD & FN & DIM \\ \hline
LASSO            &&  & 29.8&1760.2&395&1744.7&197.1&15.5&197.9&34.1&0.1&38  \\
\multicolumn{3}{l}{Dynamic EMVS (Gaussian) }    &                        &   &                        &   &            &       &&      \\
\hline
$\lambda_1=.1$ && $\Theta=1$ (DLM) 			& 16.2&2253.3&19708&2209.9&108&43.4&19600&196&0&200\\
$\lambda_1=.1$ & $\lambda_0=.01$ & $\Theta=.99^\star$  & 113.5&555.9&580.6&479&91.6&76.8&489&26.2&0.2&30\\
$\lambda_1=.1$ & $\lambda_0=.01$ &  $\Theta=.9^\star$ &50.9&469.8&153.7&422.4&85.2&47.4&68.5&4.4&0.3&8.1\\
$\lambda_1=.1$ & $\lambda_0=.01$ &  $\Theta=.5^\star$ &43.8&500.3&154.5&447.3&89.5&53&65&3.9&0.3&7.6\\
$\lambda_1=.1$ & $\lambda_0=.01$ & $\Theta=.1^\star$ & 51.9&534.2&151.1&473.6&94&60.6&57.1&3.4&0.3&7.1\\
$\lambda_1=.1$ &$\lambda_0=.01$&$\Theta=.1^\star$ (RW)  &34.2&550&122.9&502.9&91.3&47.1&31.6&2.7&0.3&6.4\\
%$\lambda_1=.01$ &$\lambda_0=.001$ &$\Theta=.1$ & 18.2&2345.3&286.4&2345.3&286.4&0&0   \\
%\hline
%\multicolumn{3}{l}{Dynamic EMVS (Laplace) }    &            &             &   &           &             &                         \\
%\hline
%$\lambda_1=.1$ &$\lambda_0=1$ &$\Theta=.5$ & 51.5&434.6&1677.1&425.4&68.7&9.1    \\
\hline\hline
\end{tabular}}
\end{table}

Now, we explore a far more challenging scenario, repeating the example with $p=200$ instead of 50.
The coefficients and data generating process are the same with $p=50$, but now instead of 46 noise regressors, we have 196.
This high regressor redundancy rate is representative of the ``$p>>n$" paradigm (``$p>>T$" for time series data) and can test the limits of any sparsity inducing procedure.
Note that the number of coefficients to estimate is $p\times T=20\,000$. This is a very challenging scenario where  we will be able to truly evaluate the efficacy of $DSS$  when there is a large number of predictors with sparse signals.

The results are collated in Table~\ref{tab_synth2}.  Across considered hyper-parameter settings, Dynamic EMVS  does extremely well also for $p=200$, dramatically reducing SSE over DLM and LASSO.
We have performed the warm start strategy for $\Theta\in\{1,0.99,0.9,0.5,0.1\}$. Moving from $\Theta=1$ to $\Theta=0.99$ already yields considerable improvements in terms of separating the signal  from noise. Reducing $\Theta$ even further, one obtains  reduced  Hamming distance for redundant covariates, i.e.  noise  is being absorbed inside the spike. 
%As in the example with $p=50$, we compare the sum of squared error (SSE) and the Hamming distance between the MAP estimate $\wh{\B}$ and the true series $\B_0$, again making the comparisons for all regressors, as well as the first four signal regressors, and the remaining noise regressors. We again use the default settings for NGAR (with $1000$ burn-in and $1000$ samples). The
While LASSO does perform well in terms of  the Hamming distance, it does not do so well in terms of SSE.
Because LASSO lacks dynamics, the pattern of sparsity is not smooth over time, leading to erratic coefficient evolutions.
Because of the smooth nature of its sparsity, $DSS$ harnesses the dynamics to discern signal from noise, improving in both SSE and Hamming distance.
We have also added global variable selection performance metrics: False Discoveries (FD), False Non-discoveries (FN) and Dimension (DIM).
FD is defined as the number of noise variables (out of the 196 noise predictors) which were included in the model at least once during the time $t=1,\dots, 100$.
Similarly, FN is the number of true signal variables (out of the 4 true predictors) which were left out of the model  at all times points $t=1,\dots, 100$. Finally, DIM is the estimated number of predictors
identified as active at least once. We can see that LASSO includes too many noise variables, while Dynamic EMVS effectively reduces the dimensionality. 
In this vein, Dynamic EMVS can be regarded as a fast screening rule which can be followed by a more thorough analysis using only a smaller subset of more meaningful predictors.

\section{Macroeconomic Data}\label{sec:high}

We further illustrate the effectiveness of $DSS$ through a  macroeconomic dataset analyzed  in \cite{kalli_griffin}. The data consists of quarterly measurements of the US inflation (the personal consumption expenditure (PCE) deflator) and 31 potential explanatory variables including previous lags of inflation, activity variables (such as economic growth rate or output gap), unemployment rate etc. The dataset was obtained from the FRED (Federal Reserve Bank of St. Louis) economic database, the consumer survey database of the University of Michigan, the Federal Reserve Bank of Philadelphia, and the Institute of Supply Management (see \cite{kalli_griffin} and Figure \ref{inflation_data} for more details).

%There are eight primary groups that comprise this dataset: Output and Income, Labor Market, Consumption and Orders, Orders and Inventories, Money and Credit, Interest Rate and Exchange Rates, Prices, and Stock Market.
%Details of the descriptions and transformations of the predictors can be found in the Appendix of \cite{mccracken}.

 {For this example, we will treat the US inflation (Figure~\ref{fig_inf}) as the dependent variable and infer its sources of covariation  with the other variables.
 Inflation forecasting has been of substantial interest within the macroeconomic literature \citep{stock1999forecasting,KoopKorobilis2012,groen,kalli_griffin,wright2009forecasting,stock2007has}. 
 The primary goal of our analysis is to retrospectively identify underlying economic indicators that are pertinent to inflation. In addition, we evaluate the one-step-ahead forecasting ability of our models.
% The underlying assumption is that the model relevant for forecasting can potentially change over time and  we deploy our dynamic shrinkage approach to capture this dynamics.
%The benefit of shrinkage for this example, as well as in other contexts, is to recover sparse signals from the underlying economic indicators that are pertinent to inflation.
Because the economy is dynamic, it is natural to assume certain indicators to be effective during a certain period but useless during another.
For example, one might expect financial indicators to play a significant role in the economy during a financial crisis.
 The necessity of capturing these dynamic trends have been discussed and explored in \cite{stock2007has}, who point out that forecasting inflation has become harder due to trend cycles and dynamic volatility processes.
Unlike \cite{stock2007has}, where they model this trend via an unobserved component trend-cycle, we explore this characteristic through dynamic sparsity in the covariate space. Similar inflation forecasting applications were considered by  many other authors including \cite{kalli_griffin,  KoopKorobilis2012}.
The dataset  has a long span (from  the second quarter of 1965 to the first quarter of 2011), capturing  oil shocks  in 1973 and 1979, mild recession in 1990, the dot-com bubble as well as  the Great Recession in  2007-2009.
We would expect our forecasting model to change during these periods.}

%\cite{ KoopKorobilis2012} analyzed an inflation dataset with a similar time span with a very similar set of predictors in their generalized Phillips curve model. 

%Their analysis suggests that for the PCE deflator with a short forecasting horizon, there is only one persisting predictor.

 \begin{figure}[t!]
\begin{center}
\includegraphics[width=4.5cm,height=4.5cm]{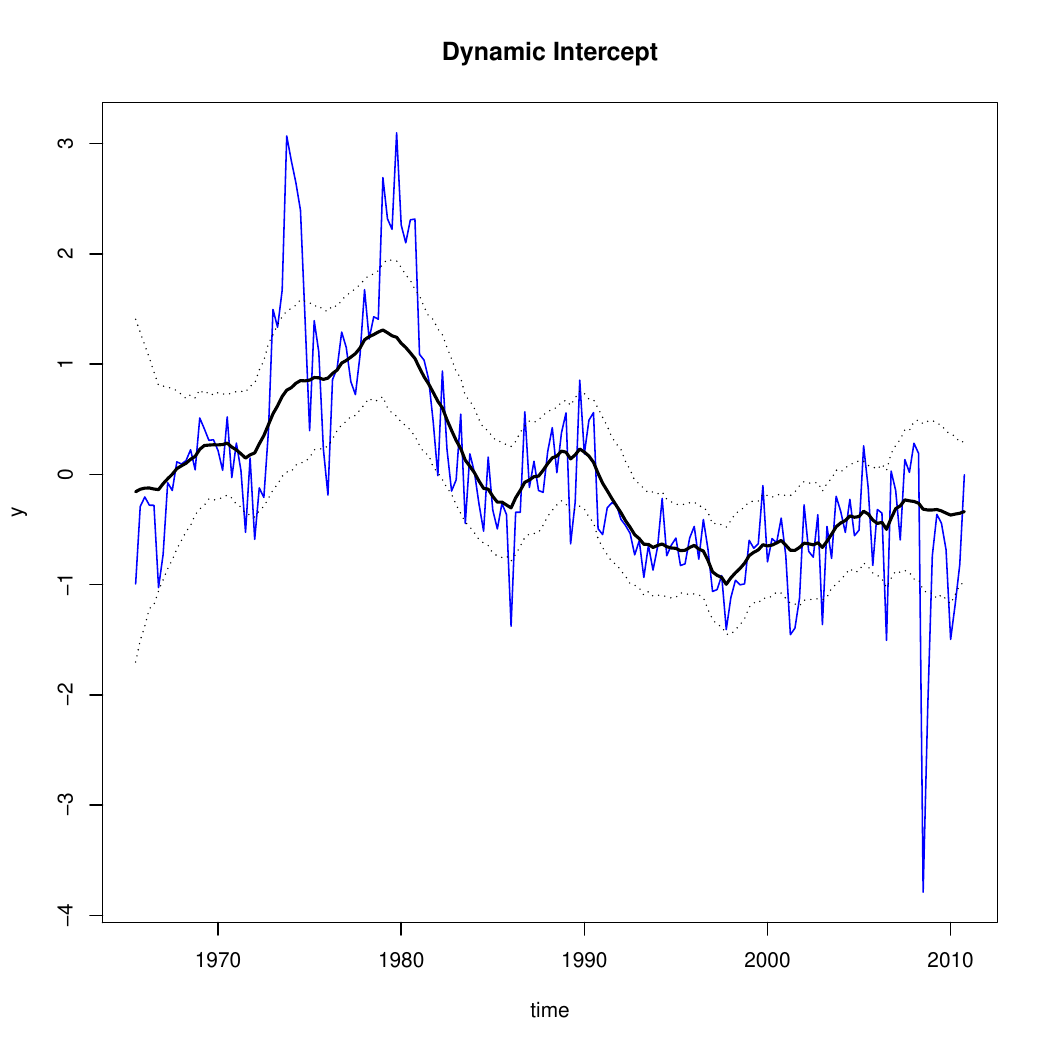}
\includegraphics[width=4.5cm,height=4.5cm]{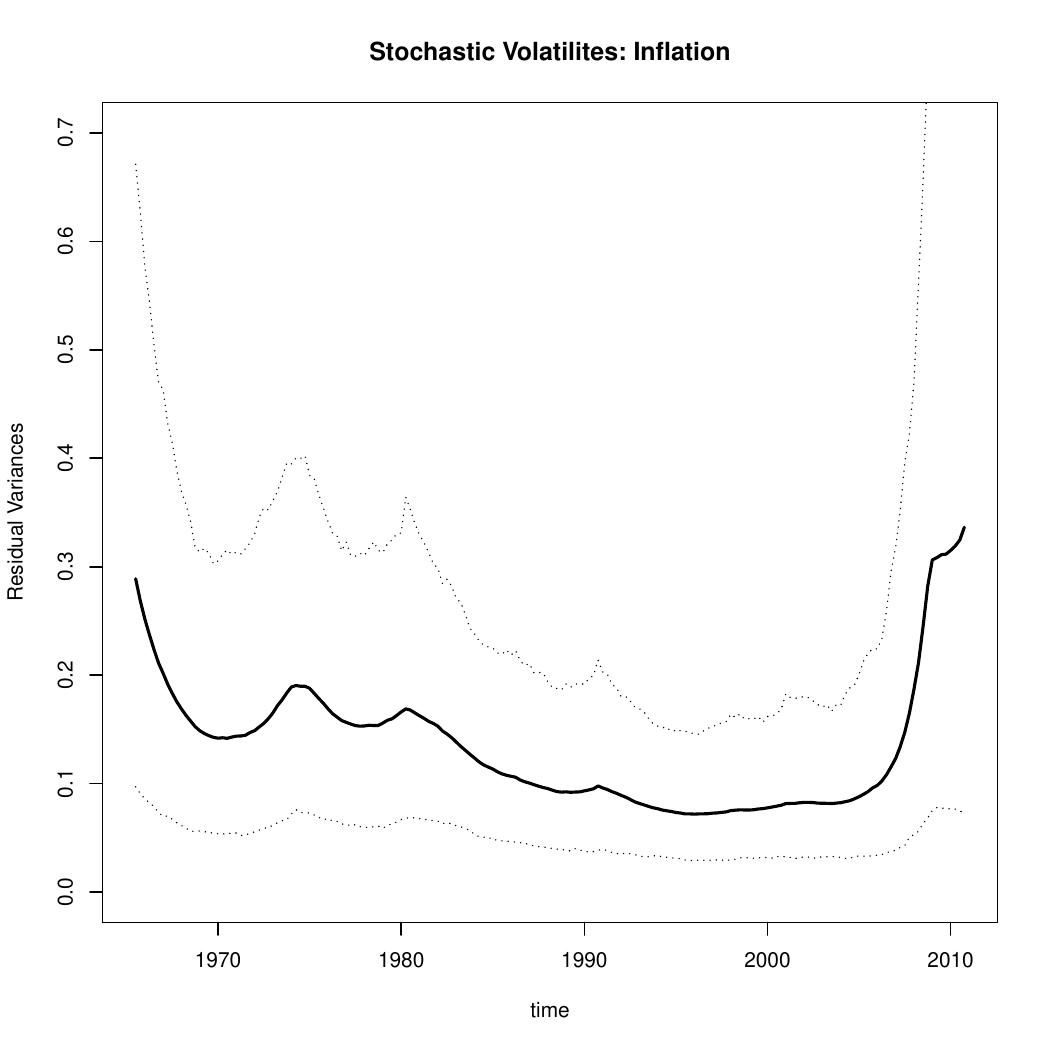}
\includegraphics[width=4.5cm,height=4.5cm]{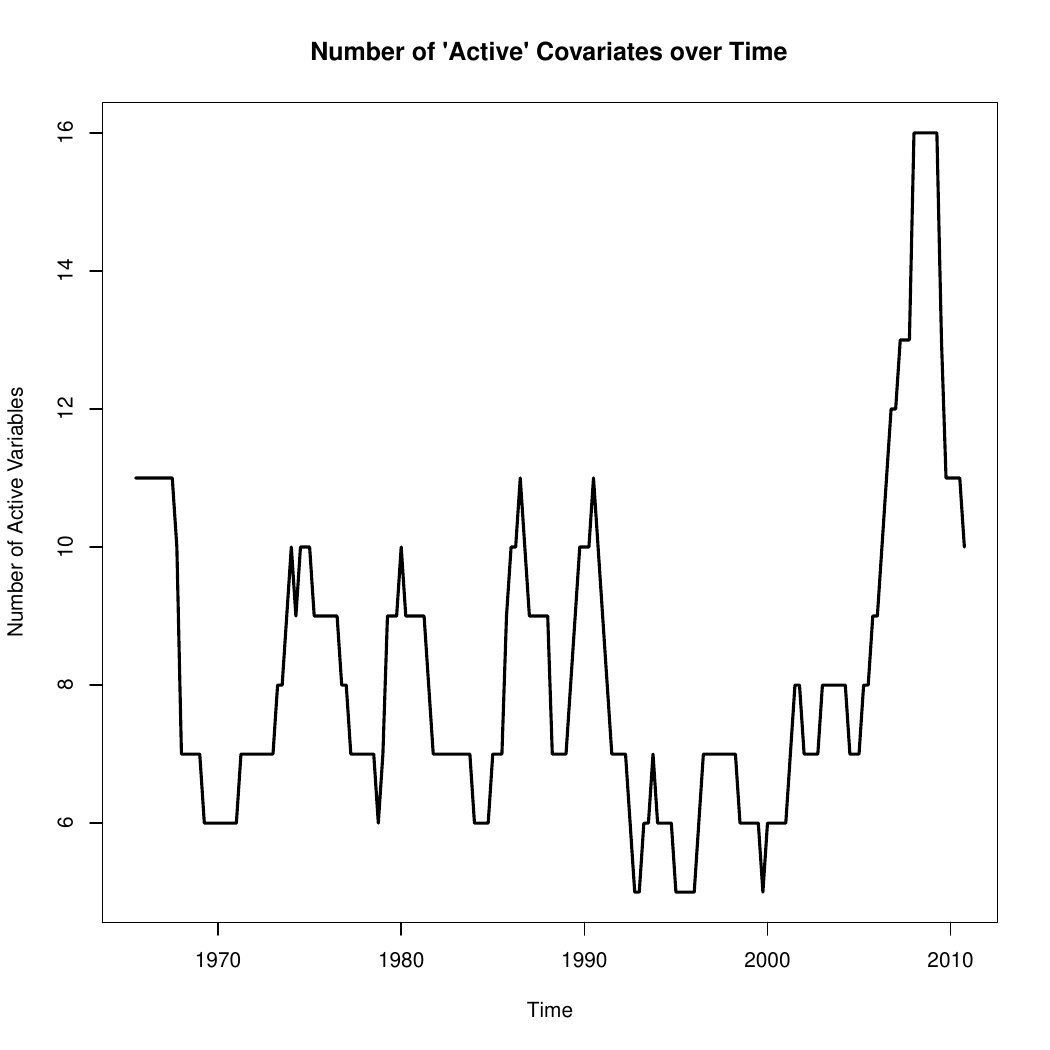}
\end{center}
\caption{\small (Left) Observed quarterly US inflation (recentered and rescaled) from 1965/2 to 2011/1 (blue time series). The black lines are the posterior mean of the dynamic intercept together with $95\%$ pointwise credible bands.  
(Middle)  Posterior means of residual variances (together with $95\%$ pointwise credible bands) under the discount stochastic volatility model.  (Right) Number of covariates with a posterior inclusion probability above 0.5}\label{fig_inf}
\end{figure}

To evaluate our method, we first measure its forecasting ability by conducting one month ahead point forecasts and computing the  mean squared cumulative forecast error (MSFE).
The analysis is done by cutting the data in half, training the methods using the first half of the data from 1965/2 to 1987/3.
We then sequentially update the forecasts through the second half from 1987/7 to 2011/1, updating and rerunning estimation to produce 1-month ahead forecasts every time we observe a new data at each $t$ (using data from $1{:}t$ to forecast $t+1$ for $t=1{:}T-1$, where $t=1$ is 1965/2, and $t=T-1$ is 2010/4). Namely, we refit the full  MCMC analysis (with $500$ MCMC iterations and $100$ burn-in) of each model to define the posterior based on $\y_{1:t}$ and to obtain forecasts $f_{t+1}=\x_{t+1}'\bm a_{t+1}$ (using the notation from Section \ref{sec:dynamic_SSVS}). We use conditional forecast densities as explained in the next paragraph. For Dynamic EMVS, we replace posterior means with modes in the forecast calculations.
%In the calculation of $\bm a_t$, we compute $\bm \Gamma_t$ based on the posterior distribution of $\b_{t-1}$ by noting that $\P[\gamma_{tj}=1\C \beta_{t-1j}]=\theta_{tj}$. Plugging in the posterior mean of $\beta_{t-1j}$ into $\theta_{tj}$, we replace $\bm \Gamma_t$ with $\phi_1\mathrm{diag}\{\theta_{t1},\dots,\theta_{tj}\}$. This forecast leverages the dynamic weights and incorporates the anticipation of sparsity. For the random-walk version, we do not have access to dynamic weights, where 
%$\theta_{tj}=\Theta$. We then simply  $\bm a_t$ as the posterior mean of $\b_{t-1}$
Out-of-sample forecasting is thus conducted and evaluated in a way that {no future information is used to analyze and evaluate the results}.
At the end of the analysis (2011/1), we estimate the retrospective coefficients throughout 1965/2 to 2011/1 in order to infer on the recovered signals, given all the data used in the analysis.
As with Section~\ref{sec:high}, we compare $DSS$ against the full DLM, null DLM (only intercept)  and LASSO (expanding window).
For $DSS$, we use multiple hyperparameters   to discern which combination produces best forecasts.
{The  initial condition for the SV variance is  $1/v_{0}\sim G(n_0/2,d_0/2)$ with $n_0=1$ and $d_0=1$. The discount factor is set to $0.9$.}

{ 
On the comparison of forecast ability  (Table \ref{tab:forecast}), it is curious that the null DLM model  actually performs better than the full DLM model. While the predictors have some explanatory power, the full DLM model is unable 
to tease out the signal and badly overfits, clearly hurting forecasts.
Dynamic EMVS (SSVS) is able to improve on the null model by capitalizing on the (albeit weak) signal hidden in the predictors.
$DSS$  thus significantly improves over the full DLM, which is unsurprising since the full DLM model is plagued with overfitting and false discoveries.
A surprising result is that Dynamic EMVS outperforms Dynamic SSVS in this example. This can be explained by the fact that  the posterior MAP trajectory is  sparser (shrunk towards zero more aggressively) and smoother than the posterior mean (which performs model averaging).  The added benefit of smoothing (in addition to sparsity) can be seen by comparing Dynamic EMVS to the  LASSO, which  achieves shrinkage, but does not capture the dynamics of signals.
The fact that the forecasting results of LASSO and Dynamic EMVS are similar suggests that the gains from shrinkage are similar.
$DSS$, capturing and capitalizing on both dynamics and shrinkage, achieves  forecast gains relative to just shrinkage (LASSO) or just dynamics (DLM).
The improved performance of Dynamic EMVS  is reassuring in the sense that the faster implementation can still yield  point forecasts that are very similar, if not better, 
to the ones obtained from the more time consuming MCMC.  

We also compare forecasting performance in terms of a metric that involves the entire predictive distribution (not just its mean), namely the sum of log-predictive likelihoods evaluated at observed values $y_{t+1}$ \citep{KoopKorobilis2012}. 
We use a conditional variant of the predictive likelihood $\pi(y_{t+1}\C \y_{1:t},\wh v_{t+1},\wh \bg_{t+1})$, where we condition on the posterior mean of the inclusion indicators and variances, i.e. $\wh \gamma_{t+1j}=\E[\gamma_{t+1j}\C \y_{1:t}]$ and $\wh v_{t+1}=\E[v_{t+1}\C \y_{1: t},\wh\b_{1:t}]$  with $\wh\beta_{tj}=\E[\beta_{tj}\C\y_{1:t}]$.

}

\begin{table}
\scalebox{0.8}{\begin{tabular}{l  |c c || l  |c c }
\hline\hline
 & MSFE & MAFE  &   & MSFE & MAFE\\
\hline
\bf Dynamic SSVS &    &  &\bf Dynamic EMVS &    &  \\
$\Theta=1,\lambda_1=0.01$ (DLM) Intercept               & 38.9   &37.03& LASSO                                                                         &        32.97    &34.38 \\
$\Theta=1,\lambda_1=0.01$ (DLM) Full  	                   & 51.01   &45.56 &$\Theta=1,\lambda_1=0.01$ (DLM) Full                       &46.59  &45.25\\
$\Theta=0.1,\lambda_1=0.01,\lambda_0=0.001$         & 36.32   &34.28   &$\Theta=0.1,\lambda_1=0.01,\lambda_0=0.001$         &31.61    &32.73\\
$\Theta=0.5,\lambda_1=0.01,\lambda_0=0.001$         & 41.31   &37.67 &$\Theta=0.5,\lambda_1=0.01,\lambda_0=0.001$         &31.69    &32.83\\
$\Theta=0.5,\lambda_1=0.01,\lambda_0=0.001$ (RV) & 42.69   &42.36 &$\Theta=0.5,\lambda_1=0.01,\lambda_0=0.001$ (RV) &41.37  &43.34\\
\hline\hline
\end{tabular}}
\caption{\small Mean squared (absolute) one-step-ahead forecast errors for Dynamic SSVS and Dynamic EMVS. RW stands for the random walk prior variant, ``Intercept" stands for a model with only a dynamic intercept and ``Full" 
stands for a full DLM model with no selection shrinkage. }\label{tab:forecast}
\end{table}

We now deploy $DSS$ priors using Dynamic SSVS ($2\,000$ posterior samples with a $500$ burn-in period) on the entire dataset to recover the series of regression coefficients. 
In order to capture   more subtle signals, we set $\lambda_1=0.01,\lambda_0=0.001$ and $\Theta=0.5$ so that the intersection point between the stationary spike-and-slab densities (i.e. our perceived selection threshold for practical significance) is $0.05$. We assume the discount stochastic volatility model with $\delta=0.9$ and $n_0=d_0=1$ and include an intercept term which is devoid of shrinkage (i.e. the intercept is  in the slab distribution at all times).
Out of the $31$ indicators (not including the intercept) only $12$ (\texttt{GDP, PCE, GPI, RGEGI,IMGS,NFP,M2,ENERGY,FOOD,MATERIALS,OUTPUT GAP, GS10}) had their posterior inclusion probability $\P[\gamma_{tj}=1\C \bm y_{1:T}]$ above $0.5$ at least $10$ times throughout the $182$ observations. We plot the number of ``active" covariates (i.e. with a posterior inclusion probability above 0.5) over time in Figure \ref{fig_inf} on the right. 
Note that the definition of an active coefficient is  ultimately tied to our choice of hyper-parameters and our practical significance threshold $0.05$.   More strict shrinkage priors would lead to fewer active covariates.
This plot presents evidence that the forecasting model is changing over time. More predictors are seen to contribute around the oil shocks, around   $1990$ and during the financial crisis, mirroring the inflation changes  during these periods.
A similar conclusion was also found in \cite{KoopKorobilis2012}.

The coefficient evolutions of the top $9$ predictors relevant for inflation are plotted in Figure  \ref{fig:evolutions}. Many of these predictors were also identified by \cite{kalli_griffin} with very similar estimated coefficient trajectories.
In particular, the explanatory power of \texttt{IMGS} (import of goods and services) growth is seen to peak around the oil shocks in the 1970's and around late 2000's. 
The two largest signals are the production growth indicator \texttt{GDP} and the consumption growth \texttt{PCE} with their coefficients largely stable with a marked increase during crisis in the late 2000's.
Interestingly, conventional indices of the labor market (including unemployment) are not recovered  with  a very strong signal (see Figure \ref{fig:evolutions2} in the Appendix).
The characteristics of these coefficients demonstrate how $DSS$ is successful in dynamically shrinking coefficients to zero during regime changes (recessions) as we would expect to happen. 
It is worth noting that the credible intervals absorb zero, indicating inherent sparsity/low signal of the contributing predictors. This is in line with earlier conclusions reached by \cite{KoopKorobilis2012} who found only very
 few predictors to be  relevant for (one-step ahead) inflation forecasting using dynamic model averaging.  From the plot of the evolution of dynamic intercept in Figure \ref{fig_inf} (on the left), we can see that the intercept itself is nicely tracking the data, leaving room for the other predictors to explain the two shocks in the 1970's and the drop around late 2000's. The plot of the estimated variances (Figure \ref{fig_inf} on the right) 
 identifies these structural shocks with increased estimated volatility, especially in the late 2000's. This companion plot indicates that structural changes affect both mean and variance.
 
%Our analysis suggests that the number of hours worked is the major indicator in the labor market as it is relevant to inflation.

 \begin{figure}[t!]
\begin{center}
\includegraphics[width=0.9\textwidth]{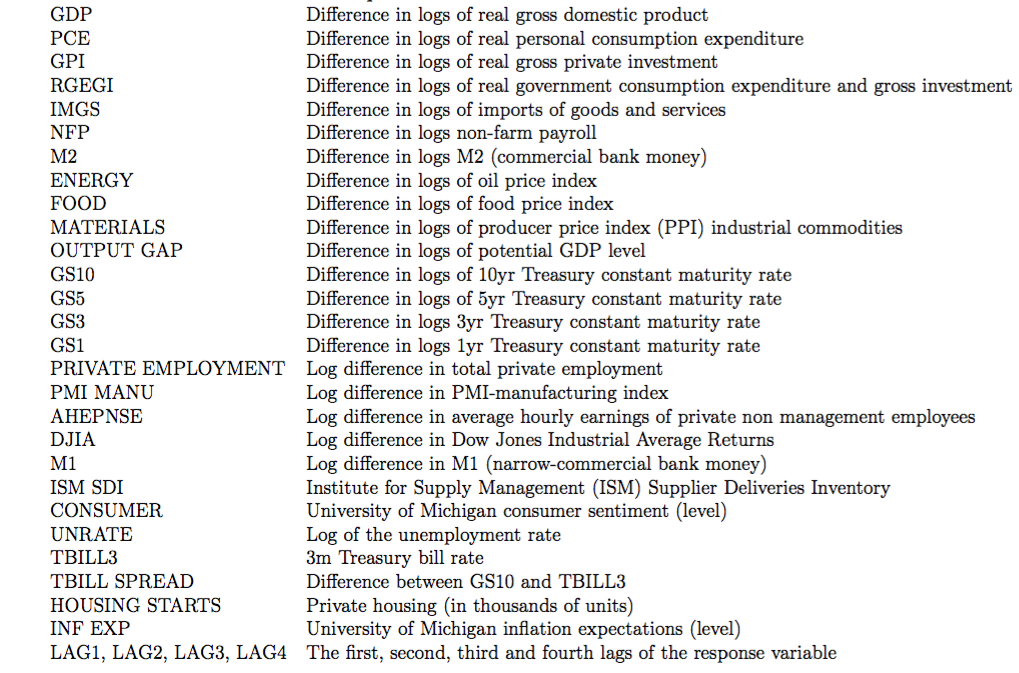}
\end{center}
\caption{A list of potential predictors for inflation forecasting (see \cite{kalli_griffin} for more details).}\label{inflation_data}
\end{figure}

\begin{figure}
\includegraphics[width=5cm,height=2cm]{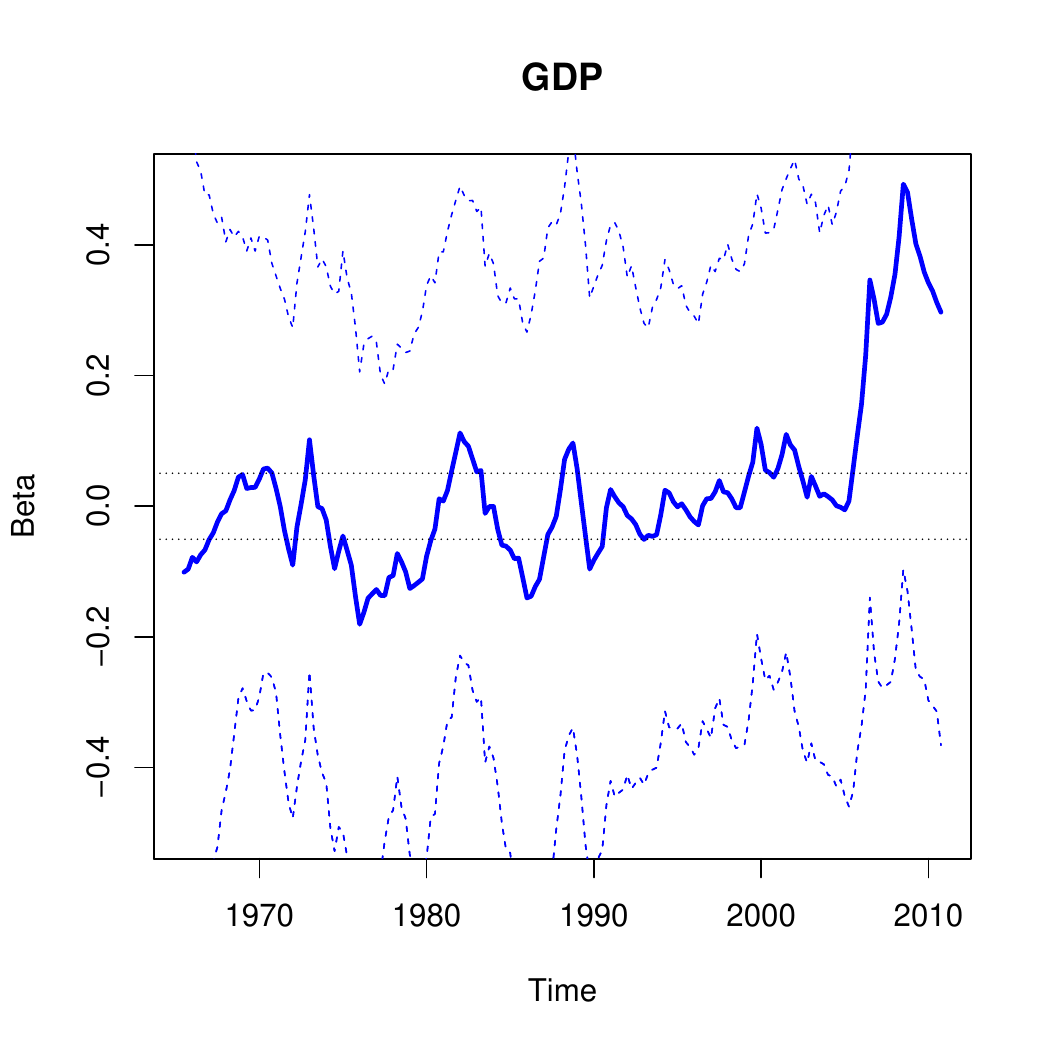}\includegraphics[width=5cm,height=2cm]{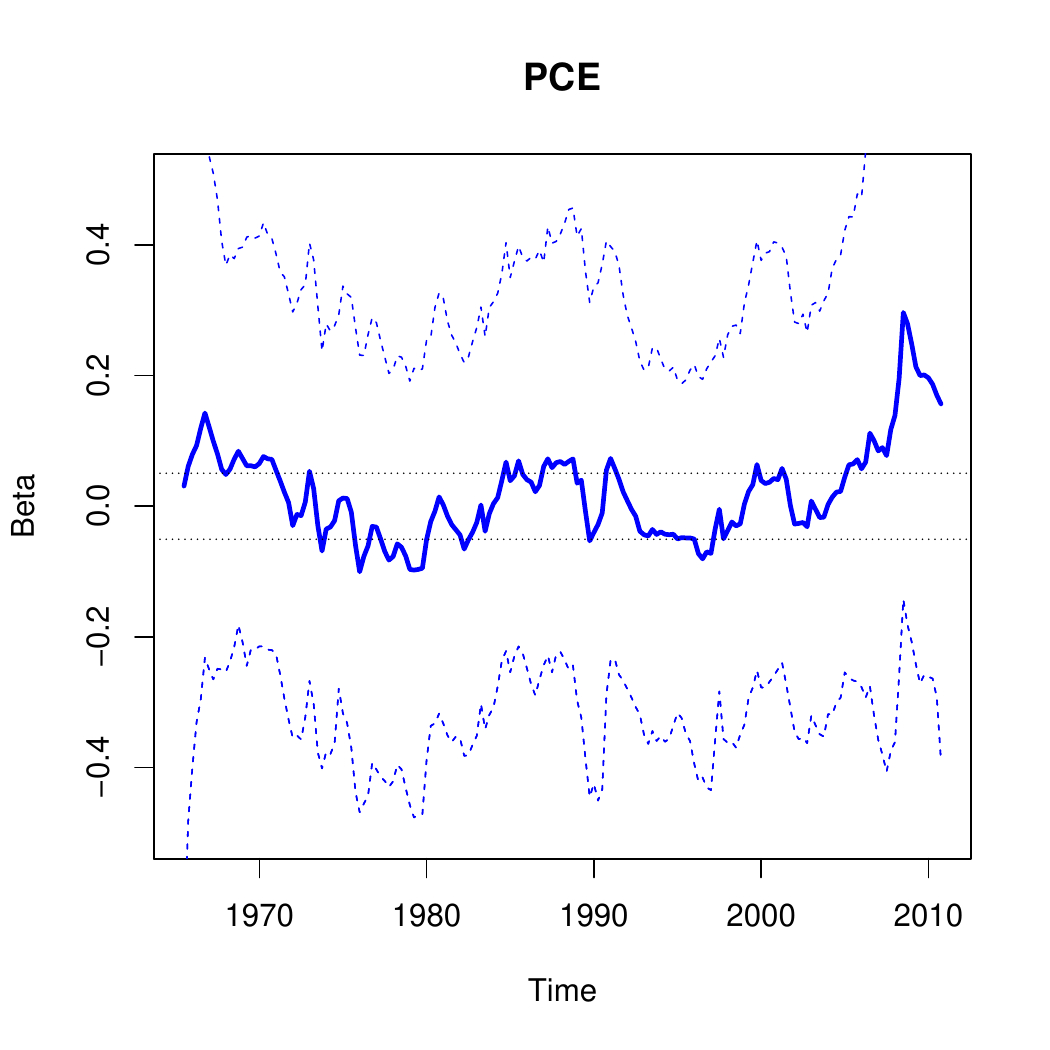}\includegraphics[width=5cm,height=2cm]{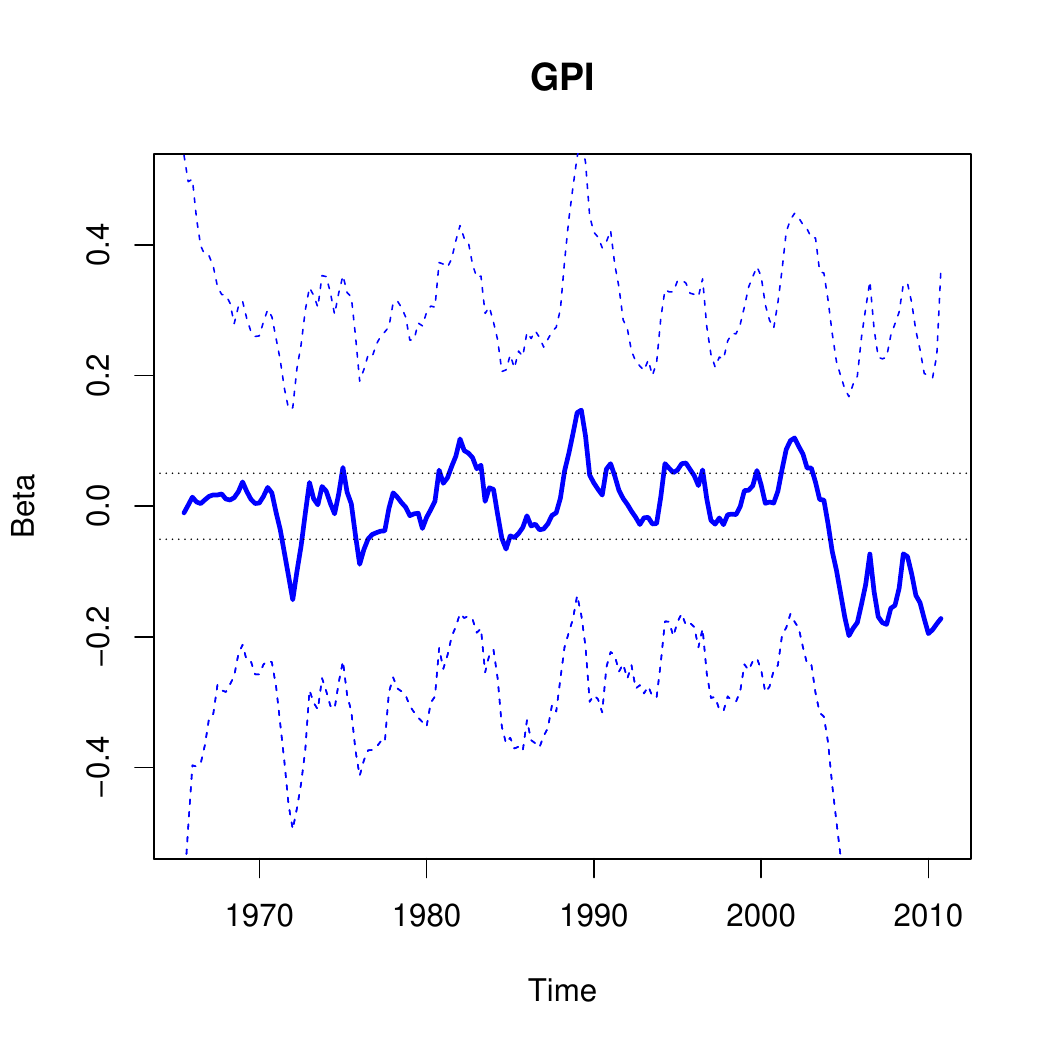}
\includegraphics[width=5cm,height=2cm]{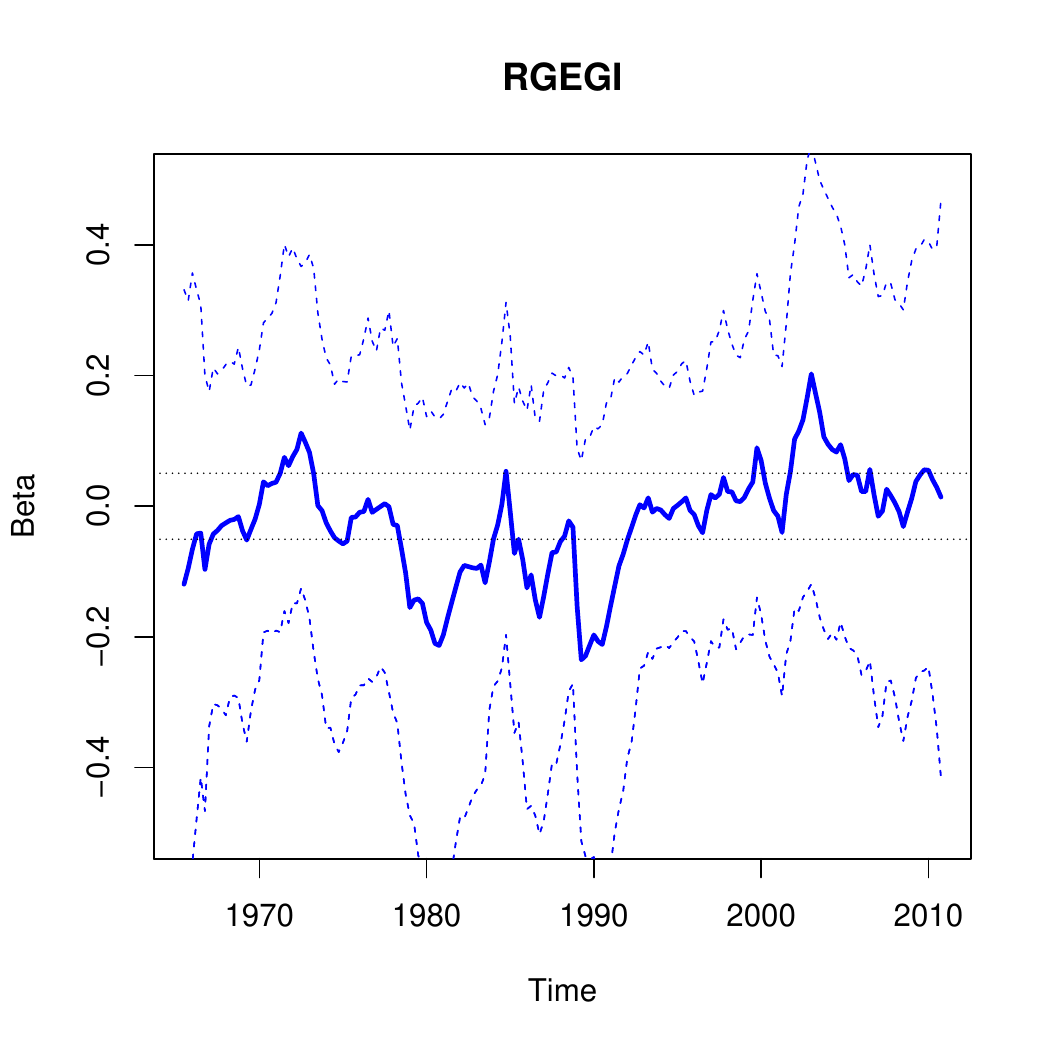}\includegraphics[width=5cm,height=2cm]{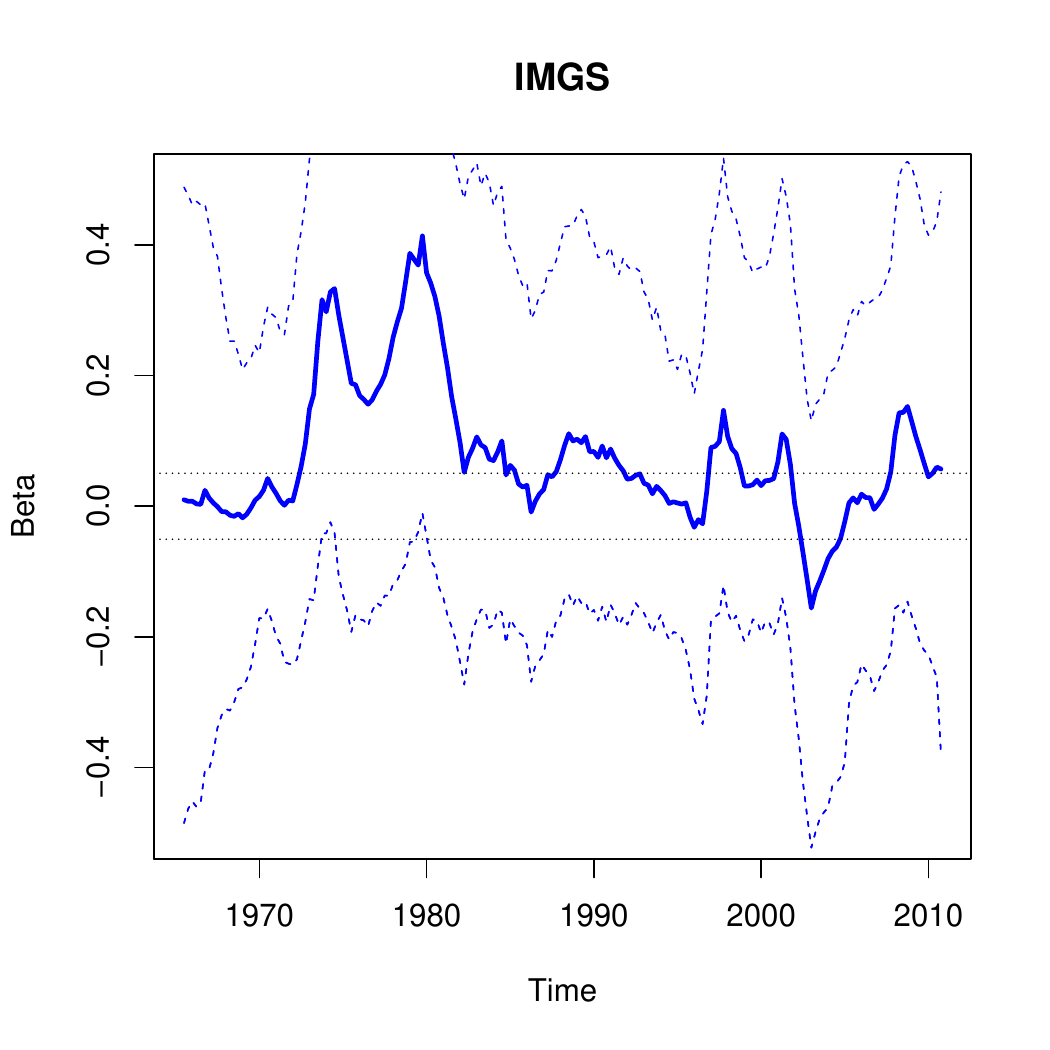}\includegraphics[width=5cm,height=2cm]{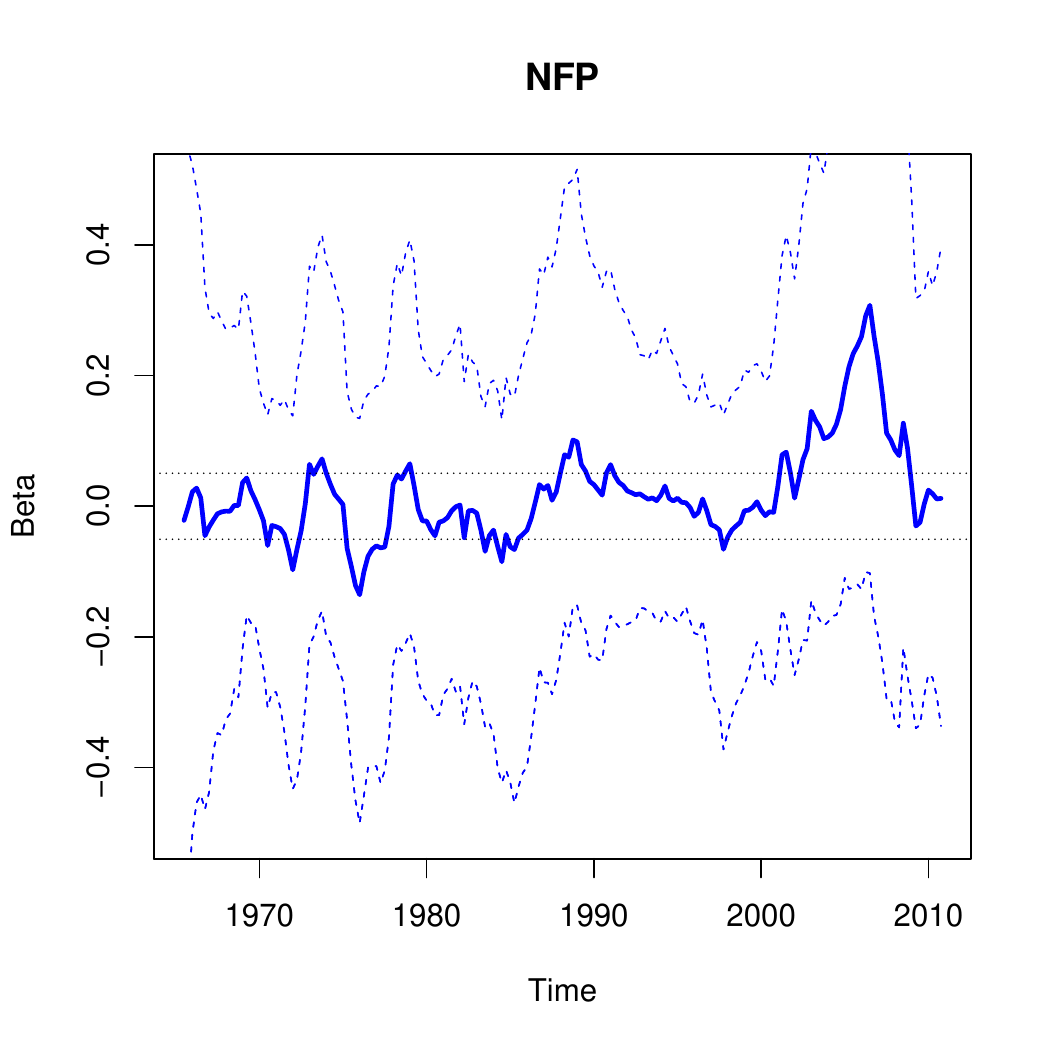}
\includegraphics[width=5cm,height=2cm]{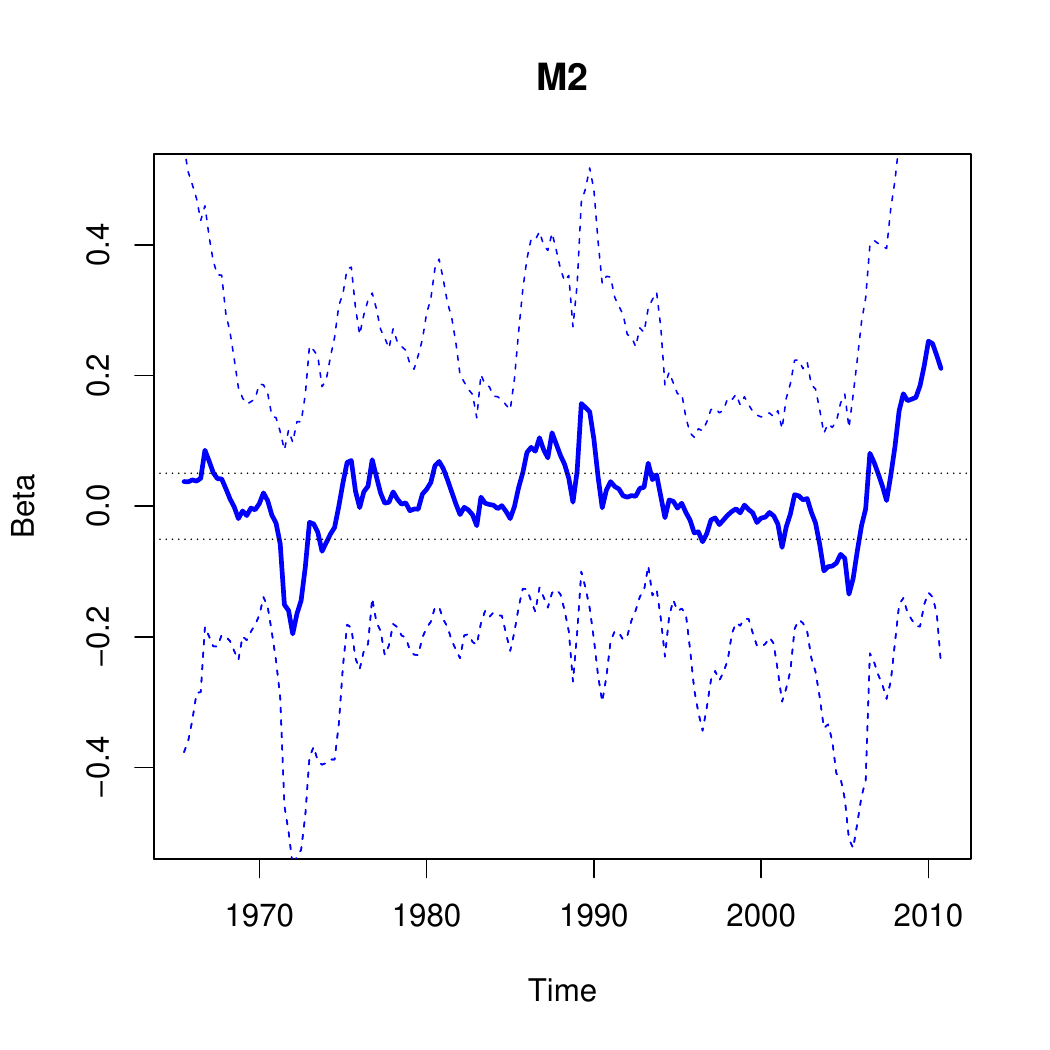}\includegraphics[width=5cm,height=2cm]{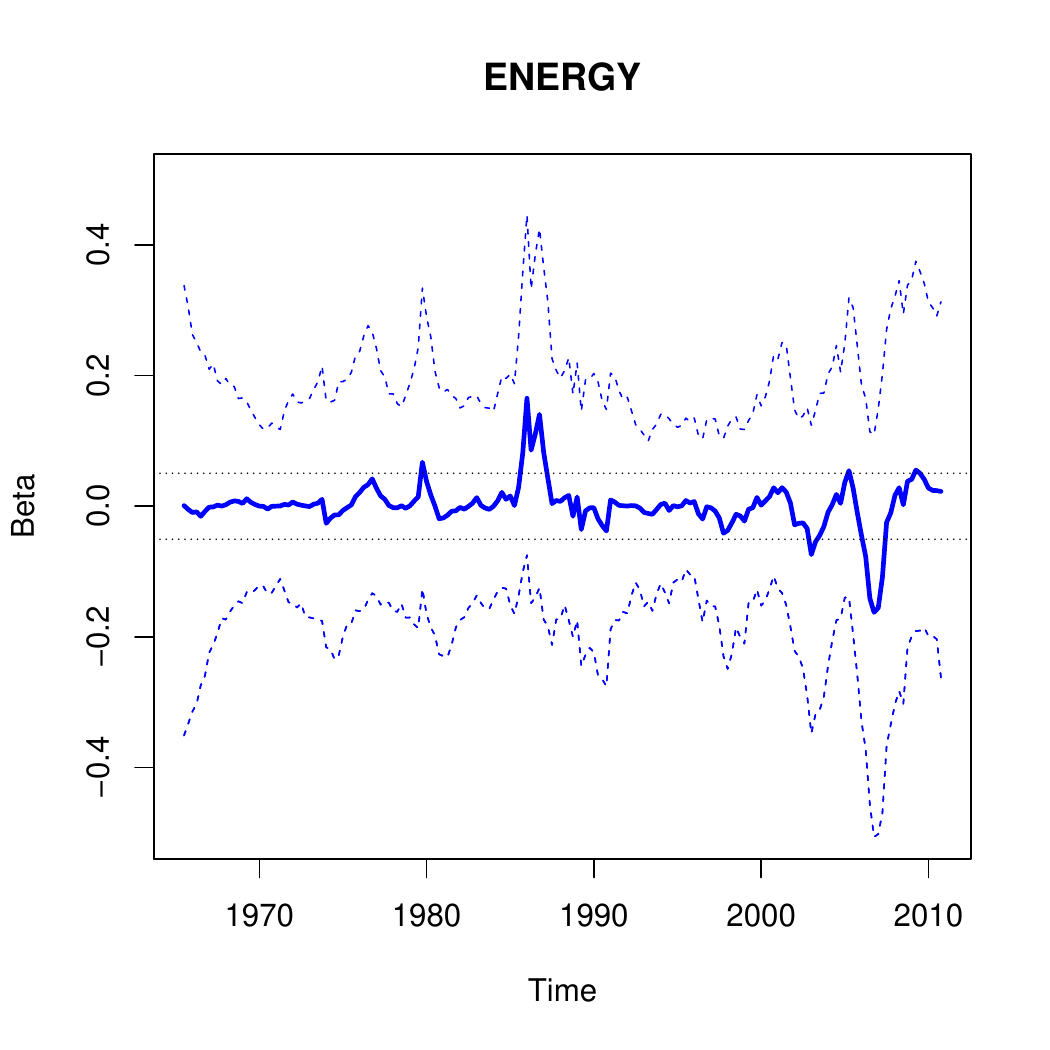}\includegraphics[width=5cm,height=2cm]{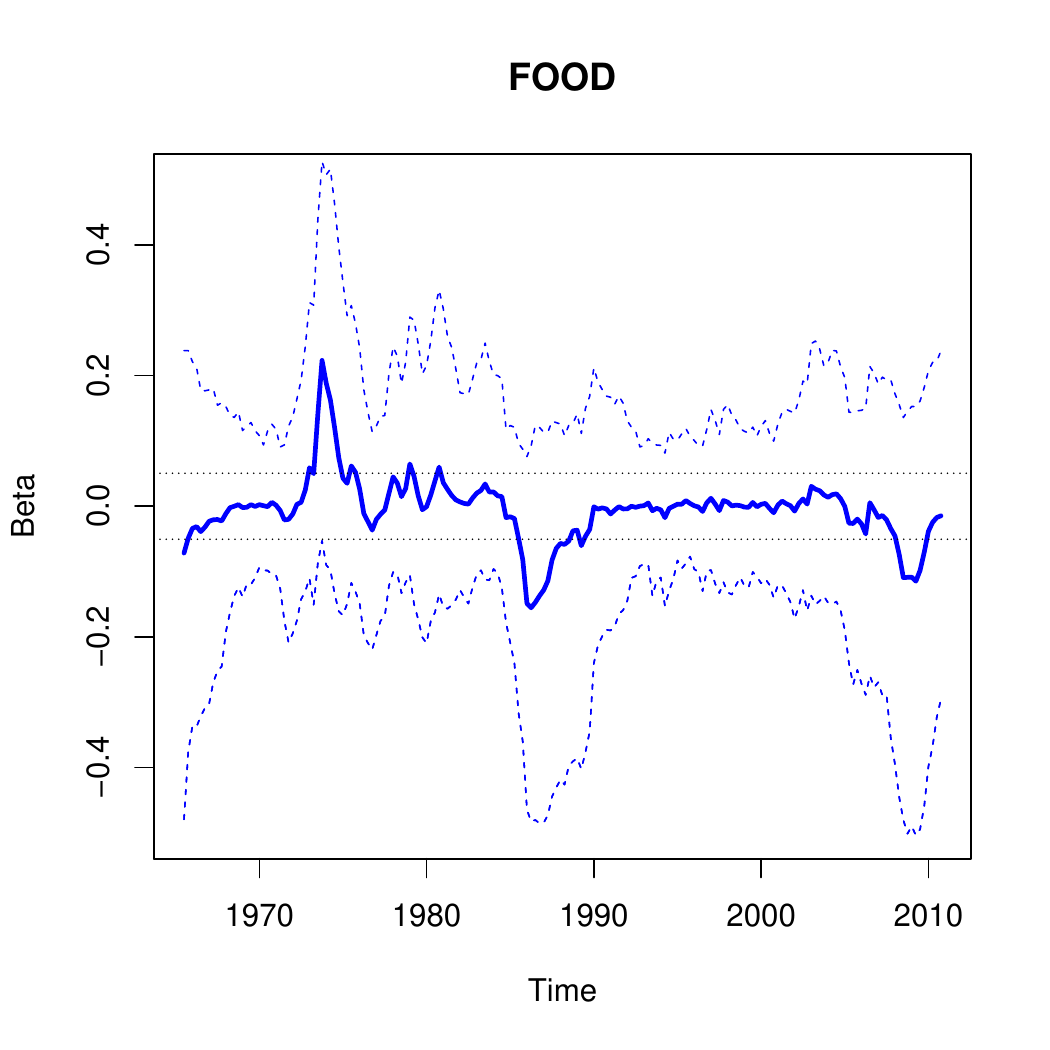}
\caption{Estimated coefficient evolutions (posterior means) of $9$ top predictors together with $95\%$ point-wise credible bands (dotted lines) using Dynamic SSVS. The horizontal lines correspond to the selection threshold $0.05$. }\label{fig:evolutions}
\end{figure}

Through this macroeconomic example, we were able to demonstrate the efficacy of $DSS$, in terms of forecast ability and interpretability, on a real, topical macroeconomic dataset.
With $DSS$, we were able to significantly outperform standard methods and successfully recover interesting signals.
MCMC benefits from the  added uncertainty statements as well as  full posterior/predictive distributions, leading to potentially better informed forecasts/decisions.

%Mirroring the results from the synthetic data, $DSS$ is able to be effective under datasets with a large number of predictors.

 \section{Discussion}\label{sec:dis}
 This paper introduces a new class of dynamic shrinkage priors, where  the stationary distribution is fully known and characterized by spike-and-slab marginals.  A key to obtaining this stabilizing property is the careful hierarchical construction of adaptive mixing weights that allows them  to depend on the lagged value of the process, thereby reflecting sparsity of past coefficients. 
 We propose  various versions of dynamic spike-and-slab ($DSS$) priors, using Laplace/Gaussian spike/slab distributions. For implementation, we resort to both optimization as well as posterior sampling.
 For Gaussian $DSS$ prior variants, we develop a Dynamic SSVS MCMC algorithm for posterior sampling.  For fast MAP smoothing, we develop a complementary procedure called Dynamic EMVS which can quickly  glean into the signal structure. For the Laplace spike variant,
  we implement a one-step-late EM algorithm for MAP estimation which iterates over one-site closed-form thresholding rules. Through simulation and a macroeconomic dataset, we demonstrate that $DSS$ are well suited for the dual purpose of dynamic variable selection (through thresholding to exact zero) and smoothing (through an autoregressive slab process) for forecasting and inferential goals. 
 
{
Many variants and extensions are possible for our $DSS$ prototype constructions. While our development has focused on stationary situations,  our priors can accommodate random walk evolutions (as we point out in our Remark \ref{remark:nonstat}). 
\cite{schotman} argue that  ``There is no need to look at the data from the specific viewpoint of stationarity or nonstationarity. Given the data one can determine which of the two is the most likely." We view stationarity as a modelling assumption which may be suitable for some data sets and less appropriate for others. For instance, \cite{schotman} discovered that for real exchange rate data, stationarity is a posteriori as probable as the random walk hypothesis. 
Recently,  \cite{lopes_mcc_tsay} propose a mixture prior which switches between stationary and non-stationary specifications and enables the quantification of  posterior plausibility of the unit root hypothesis.   An extension of our approach along these lines would be very interesting. We provide both stationary and non-stationary variants for the practitioners to choose from. Another interesting extension will be embedding our $DSS$ priors within the TVP-VAR models  \citep{Cogley2005,Primiceri2005,lopes_mcc_tsay,nakajima_west,
pettenuzzo2016bayesian,gefang2014bayesian,giannone2014short,korobilis2013var,banbura2010large}.

 % A modification of Theorem \ref{thm1}  follows, where the marginals are underpinned by the beta-binomial prior as opposed to the Bernoulli prior.
%The second very important extension will be treating the residual variances $\sigma^2$ as random and possibly time-varying. In this paper, we focused primarily on the priors on the regression coefficients. However, the  $DSS$ priors can be deployed in tandem with a stochastic process prior on $\{\sigma^2_t\}_{t=1}^T$ (as e.g. in \cite{kalli_griffin}). We have noted one possible variant in our MCMC algorithm in the Appendix.
}

%One could incorporate a similar prior as in \cite{lopes_mcc_tsay} within our MCMC scheme.

An R code is available  from the first author upon request.

\section*{Acknowledgments}
The authors would like to thank the Reviewers and the Associate Editor for providing thorough feedback which lead to substantial improvements of our paper.}

%\section*{Acknowledgments}
%This work emerged from interactions with the Optimization  Bayesian Analysis $\&$ Decisions (O,BAD) group of the SAMSI Optimization program 2016/17. 
%We would like to thank Mike West for providing useful comments and Ed George for reading through the paper and for providing generous suggestions and comments. This work was supported by the James S. Kemper Foundation Faculty Research Fund at the University of Chicago Booth School of Business  

\bibliographystyle{ba}

\clearpage
\setcounter{page}{1}
\begin{center}
{\bf\LARGE Dynamic Variable Selection \\with Spike-and-Slab Process Priors} 

\bigskip
{\Large Veronika Rockova \& Kenichiro McAlinn} 

\bigskip
{\large  Supplementary Material} 

\bigskip\bigskip
\end{center}

\appendix
\section{Additional Plots}
Figure \ref{fig:evolutions2} plots coefficient evolutions for additional variables in the inflation study.
\begin{figure}[!h]
\includegraphics[width=7cm,height=4cm]{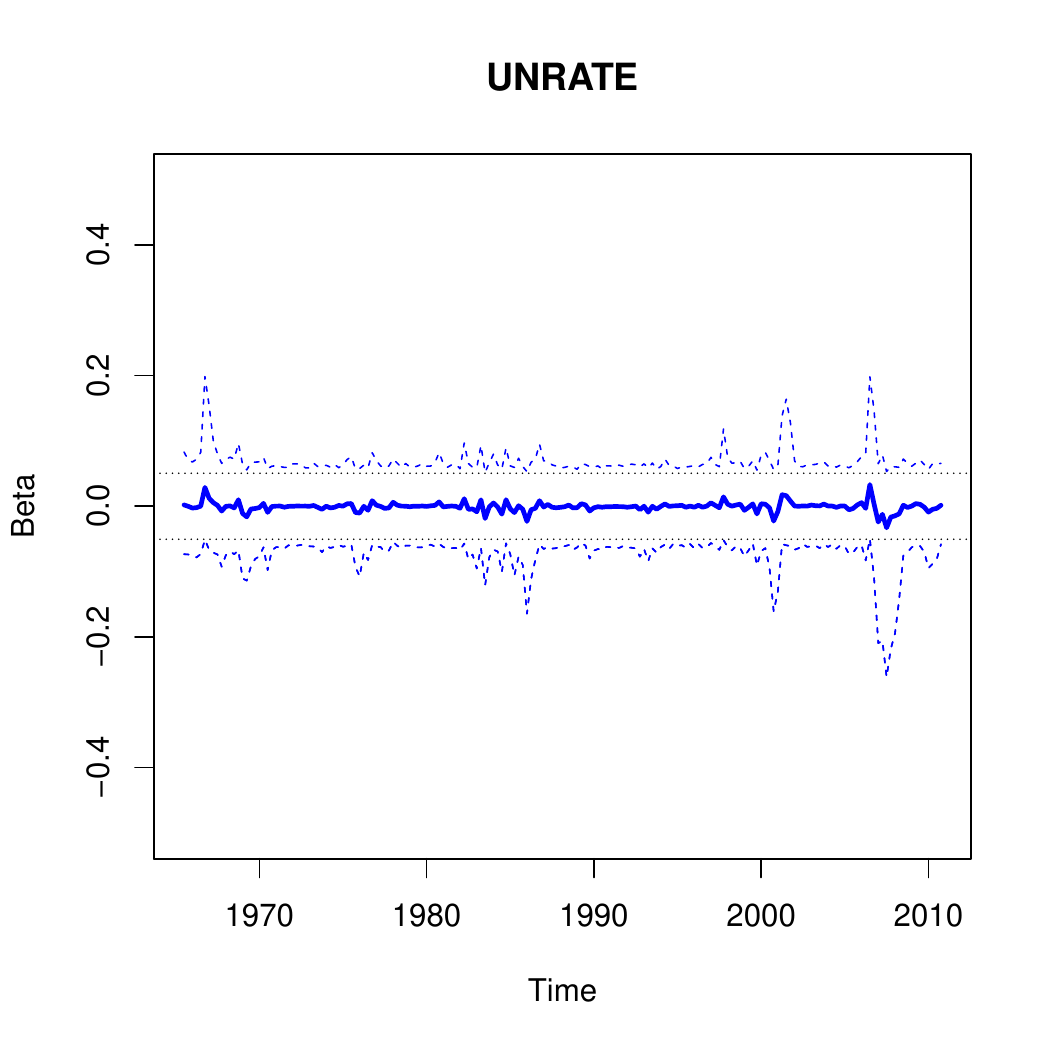}\includegraphics[width=7cm,height=4cm]{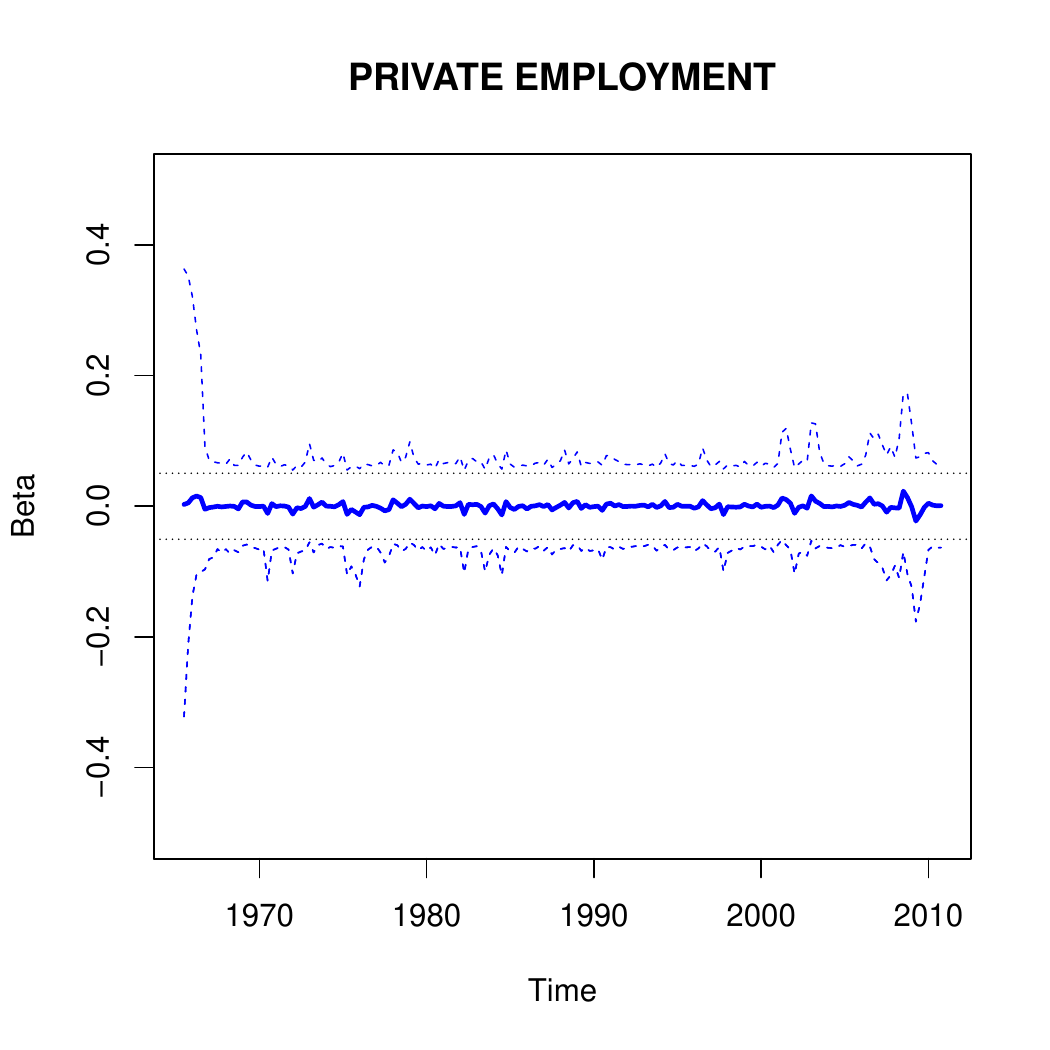}
\includegraphics[width=7cm,height=4cm]{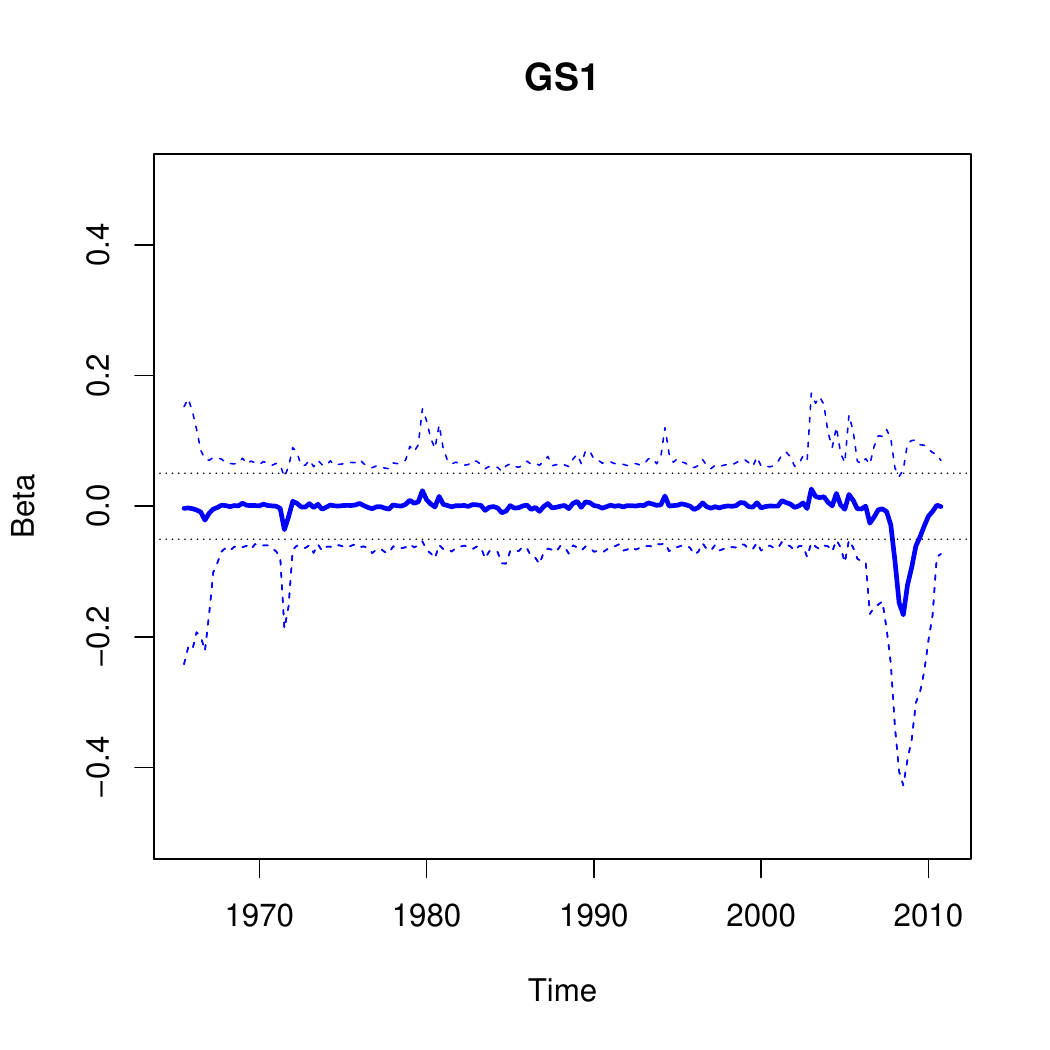} \includegraphics[width=7cm,height=4cm]{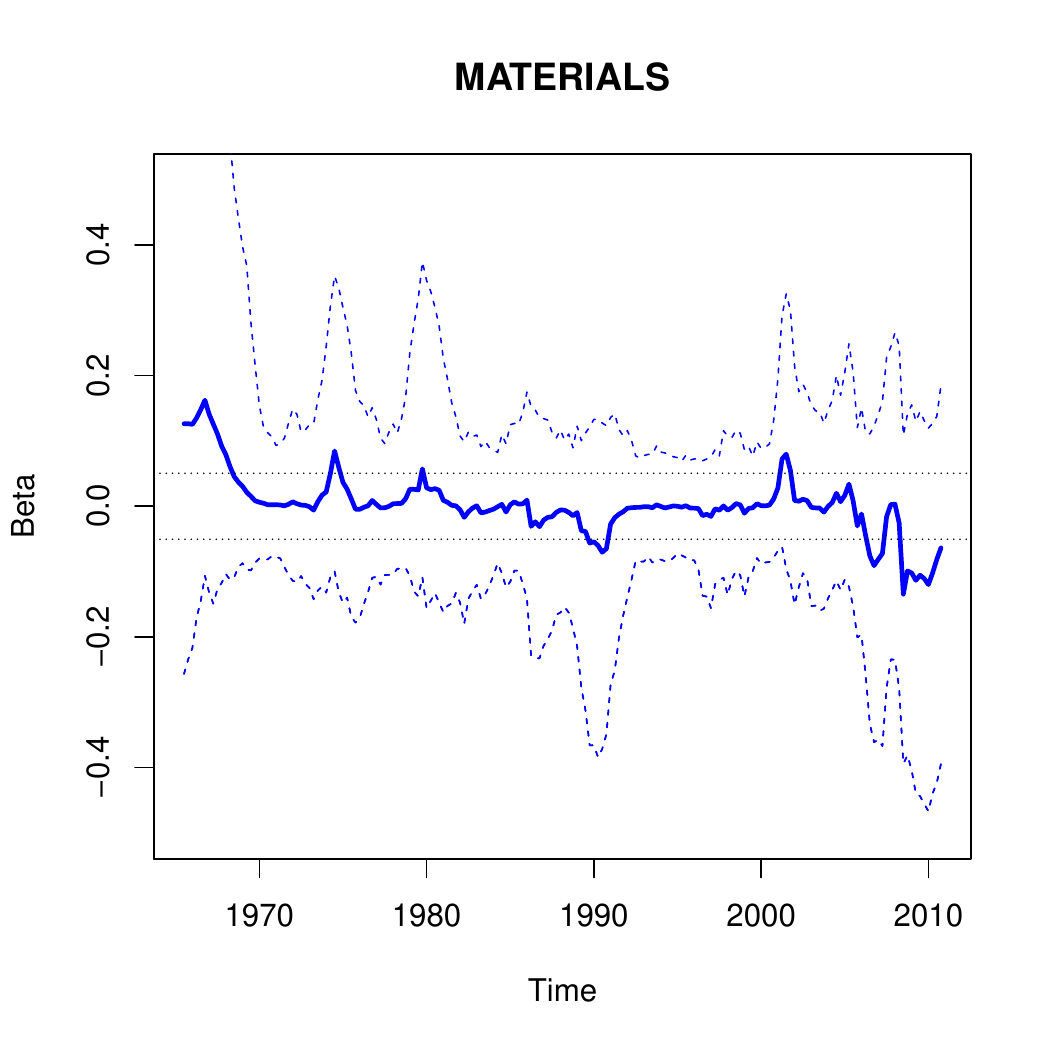}%\includegraphics[width=5cm,height=2cm]{pics/M2.pdf}\includegraphics[width=5cm,height=2cm]{pics/ENERGY.pdf}\includegraphics[width=5cm,height=2cm]{pics/FOOD.pdf}
\caption{Estimated coefficient evolutions (posterior means) together with $95\%$ point-wise credible bands (dotted lines) using Dynamic SSVS.}\label{fig:evolutions2}
\end{figure}

\section{Additional Simulations (Laplace spike) \label{supp:add}}

We perform sensitivity analysis for the Laplace EMVS version by studying the effect of tuning parameters $\phi_1,\lambda_1,\lambda_0$ and $\Theta$. We perform  $10$ new experiments, generating different responses and regressors using a similar set of coefficients as in Section 6. We compare the average sum of squared error (SSE) and average Hamming distance between the MAP estimate $\wh{\B}$ and the true series $\B_0$. Table \ref{tab_aux1} reports average performance metrics over the $10$ experiments.
The performance of $DSS$ is compared to the full DLM model \citep{WestHarrison1997book2} and LASSO.
%, NGAR \citep{kalli_griffin} and LASSO.
%For the DLM, we assume the state dynamics and observation innovations are known  (i.e. $\phi_0=0, \phi_1=0.98,\nu_{1:T}=\sqrt{0.25}$\footnote{what is $\nu$?}) with standard normal priors\footnote{I thought the variance was $10(1-\phi_1^2)$} on the coefficients.
%For DLM and NGAR, we use the same specifications as above.
%\begin{figure}[!t]
%\begin{center}
%\scalebox{0.3}{\includegraphics{pics/beta1}}\scalebox{0.3}{\includegraphics{pics/beta2}}\scalebox{0.3}{\includegraphics{pics/beta3}}
%\scalebox{0.3}{\includegraphics{pics/beta4}}\scalebox{0.3}{\includegraphics{pics/beta5}}\scalebox{0.3}{\includegraphics{pics/beta6}}
%\end{center}
%\caption{True and recovered time series of regression coefficients from the low-dimensional simulated example with $p=6$}\label{fig_simul}
%\end{figure}
For $DSS$, we now explore a multitude of combinations of hyper-parameters with $\phi_1=\{0.95,0.98\}$, $\lambda_0=\{0.7,0.9\}$,  $\lambda_1=\{10(1-\phi_1^2),25(1-\phi_1^2)\}$, and $\Theta=\{0.9,0.95,0.98\}$.  All these parameters are in the mild sparsity range; not over-emphasizing the spike.
We initialize the calculation with a zero matrix.

Looking at Table~\ref{tab_synth}, $DSS$ performs better in terms of both SSE and Hamming distance compared to DLM  and LASSO for the majority of the hyperparameters considered. To gain more insights, the table is divided into three blocks: overall performance on $\beta_{1:50}$, active coefficients $\beta_{1:4}$ and noise coefficients $\beta_{5:50}$.
The Hamming distance is reported in percentages. The number of false positives (and thereby the overall Hamming distance) is seen to increase with $\Theta$, where  large values of $\Theta$ have to be compensated with a larger spike parameter $\lambda_0$ to shrink the noisy coefficients to zero. The stationary slab variance $\lambda_1/(1-\phi_1^2)$ also affects variable selection, where larger values increase the selection threshold and produce less false discoveries. The reverse is true for the signal coefficients. In terms of SSE, $DSS$ outperforms the other two methods in estimating  $\beta_{5:50}$, demonstrating great success in suppressing unwanted parameters. 
Regarding the choice of $\phi_1$, larger values seem beneficial for  the signal coefficients, where borrowing more strength enhances stability in predictive periods.

Although there are some settings where $DSS$ underperforms, it is clear that $DSS$ has the potential to greatly improve over existing methods for a wide range of hyperparameters. In terms of SSE, the less well-performing settings are associated with large $\lambda_1$ (e.g. $\lambda_1=25/(1-\phi_1^2)$ and $\phi_1=0.95$), where the slab process is allowed to meander away from the previous value. The lack of stickiness (smaller $\phi_1)$ also provides an opportunity for the spike to threshold. The best performing setting for SSE is seen for a sticky prior ($\phi_1=0.98$) with a small slab variance ($\lambda_1=10/(1-0.98^2)$), a larger spike penalty ($\lambda_0=0.9$) and not excessively large $\Theta$. This combination seems to strike the right balance between selection and shrinkage.

\begin{table}[!t]
\centering
\caption{Performance evaluation of the methods compared for the high-dimensional simulated example with $p=50$. The results are split for the signal parameters ($x_{1:4}$) and noise parameters ($x_{5:50}$). Hamming distance is in percentages. Number in brackets for $DSS$ is using the hyperparameter set $\{\phi_1,\lambda_0,\lambda_1/(1-\phi_1^2),\Theta\}$. Best 5 results in $DSS$ are in bold.}
\label{tab_aux1}
\scalebox{0.8}{\begin{tabular}{lrlrrlrrlrr}
\hline\hline
                         &       &  & \multicolumn{2}{c}{$x_{1:50}$} &  & \multicolumn{2}{c}{$x_{1:4}$} &  & \multicolumn{2}{c}{$x_{5:50}$} \\
$p=50$                   & Time (s) &  & SSE         & Ham.      &  & SSE         & Ham.     &  & SSE         & Ham.      \\ \hline
DLM                      & 0.2   &   & 629.6       & 94.2        &   & 530.8      & 27.0        &   & 98.8       & 100       \\
%NGAR                     & 971.5 &   & 218.3       & 94.2        &   & 123.0      & 27.0        &   & 95.3       & 100       \\
LASSO                    & 10.2 &   & 552.4       & 19.8        &   & 342.9      & 52.0        &   & 219.8       & 17.0        \\
DSS                      &       &   &            &             &   &           &             &   &            &             \\
\hline
$\{$.95, .7, 10, .9$\}$  & 19.7  &   & 123.0       & 12.1        &   & 120.9      & 19.4        &   & 2.1       & 11.5        \\
$\{$.95, .7, 10, .95$\}$ & 19.3  &   & \textbf{ 89.5}       & 18.7        &   & \textbf{ 87.5}      & \textbf{ 15.6}        &   & \textbf{ 2.0}       & 19.0        \\
$\{$.95, .7, 10, .98$\}$ & 17.0  &   & 108.2       & {\color{black} 33.8}        &   & 104.0      & 17.2        &   & 4.2       & {\color{black} 35.3}        \\
$\{$.95, .9, 10, .9$\}$  & 19.9  &   & 206.3       & \textbf{ 7.8}         &   & 202.7      & 28.0        &   & 3.6       & 6.1         \\
$\{$.95, .9, 10, .95$\}$ & 19.5  &   & 127.6       & 10.1        &   & 125.0      & 19.9        &   & 2.6       & 9.3         \\
$\{$.95, .9, 10, .98$\}$ & 19.0  &   & 91.8       & 17.9        &   & 89.7      & \textbf{ 15.8}        &   & 2.1       & 18.0        \\
\hline
$\{$.95, .7, 25, .9$\}$  & 21.8  &   & {\color{black} 425.1}       & 7.9         &   & {\color{black} 413.6}      & {\color{black} 45.4}        &   & 11.5       & \textbf{ 4.7}         \\
$\{$.95, .7, 25, .95$\}$ & 20.1  &   & {\color{black} 341.8}       & 9.4         &   & {\color{black} 333.5}      & {\color{black} 40.3}        &   & 8.2       & 6.7         \\
$\{$.95, .7, 25, .98$\}$ & 21.0  &   & 233.9       & 12.6        &   & 229.1      & 30.8        &   & 4.7       & 11.0        \\
$\{$.95, .9, 25, .9$\}$  & 19.0  &   & {\color{black} 468.8}       & \textbf{ 6.1}         &   & {\color{black} 455.4}      & {\color{black} 48.1}        &   & {\color{black} 13.4}       & \textbf{ 2.4}         \\
$\{$.95, .9, 25, .95$\}$ & 17.4  &   & {\color{black} 404.4}       & \textbf{ 6.6}         &   & {\color{black} 392.8}      & {\color{black} 44.4}        &   & 11.5       & \textbf{ 3.3}         \\
$\{$.95, .9, 25, .98$\}$ & 16.6  &   & {\color{black} 303.7}       & 8.4         &   & {\color{black} 297.6}      & {\color{black} 36.8}        &   & 6.1       & \textbf{ 6.0}         \\
\hline
$\{$.98, .7, 10, .9$\}$  & 18.5  &   & 130.6       & {\color{black} 39.7}        &   & 115.2      & 21.7        &   & {\color{black} 15.4}       & {\color{black} 41.2}       \\
$\{$.98, .7, 10, .95$\}$ & 22.8  &   & 215.1       & {\color{black} 59.3}        &   & 174.0      & 26.1        &   & {\color{black} 41.1}       & {\color{black} 62.2}        \\
$\{$.98, .7, 10, .98$\}$ & 14.5  &   & 290.8       & {\color{black} 84.1}        &   & 239.1      & 26.7        &   & {\color{black} 51.7}       & {\color{black} 89.1}        \\
$\{$.98, .9, 10, .9$\}$  & 20.0  &   & \textbf{ 79.7}       & 17.2        &   & \textbf{ 77.7}      & \textbf{ 14.2}        &   & \textbf{ 2.0}       & 17.4        \\
$\{$.98, .9, 10, .95$\}$ & 20.9  &   & \textbf{ 81.5}       & 28.2        &   & \textbf{ 77.4}      & 17.5        &   & 4.1       & 29.1        \\
$\{$.98, .9, 10, .98$\}$ & 9.2   &   & 148.0       & {\color{black} 48.1}        &   & 132.8      & 22.7        &   & {\color{black} 15.2}       & {\color{black} 50.3}        \\
\hline
$\{$.98, .7, 25, .9$\}$  & 17.7  &   & 162.9       & 10.0        &   & 160.4      & 22.7        &   & \textbf{ 2.5}       & 8.9         \\
$\{$.98, .7, 25, .95$\}$ & 13.7  &   & 94.5       & 13.5        &   & 92.8      & \textbf{ 15.9}        &   & \textbf{ 1.6}       & 13.3        \\
$\{$.98, .7, 25, .98$\}$ & 17.1  &   & \textbf{ 75.2}       & 23.8        &   & \textbf{ 73.1}      & \textbf{ 15.5}        &   & 2.0       & 24.5        \\
$\{$.98, .9, 25, .9$\}$  & 12.3  &   & 254.1       & \textbf{ 7.2}        &   & 249.8      & 32.4        &   & 4.3       & \textbf{ 5.1}         \\
$\{$.98, .9, 25, .95$\}$ & 12.0  &   & 152.7       & \textbf{ 7.8}         &   & 150.0      & 21.9        &   & 2.7       & 6.6         \\
$\{$.98, .9, 25, .98$\}$ & 10.6  &   &\textbf{  91.3}       & 12.6        &   & \textbf{ 89.5}      & 16.0        &   & \textbf{ 1.8}       & 12.3       \\
\hline\hline
\end{tabular}}
\end{table}

Now, we explore a far more challenging scenario, repeating the example with $p=1000$ instead of 50.
The coefficients and data generating process are the same with $p=50$, but now instead of 46 noise regressors, we have 996.
This high regressor redundancy rate is representative of the ``$p>>n$" paradigm (``$p>>T$" for time series data) and can test the limits of any sparsity inducing procedure.
Under this setting, we will be able to truly evaluate the efficacy of $DSS$ and compare it to other methods when there is a large number of predictors with sparse signals.
The results are collated in Table~\ref{tab_synth1000}.  The same set of hyperparameters that performed best for $p=50$ also does extremely well  for $p=1000$, dramatically reducing SSE over DLM and LASSO.
%As in the example with $p=50$, we compare the sum of squared error (SSE) and the Hamming distance between the MAP estimate $\wh{\B}$ and the true series $\B_0$, again making the comparisons for all regressors, as well as the first four signal regressors, and the remaining noise regressors. We again use the default settings for NGAR (with $1000$ burn-in and $1000$ samples). The
We also note that, while LASSO does perform well in terms of  the Hamming distance, it does not do so well in terms of SSE.
Because LASSO lacks dynamics, the pattern of sparsity is not smooth over time, leading to erratic coefficient evolutions.
Because of the smooth nature of its sparsity, $DSS$ harnesses the dynamics to discern signal from noise, improving in both SSE and Hamming distance.

\begin{table}[!t]
\centering
\caption{Performance evaluation of the methods compared for the high-dimensional simulated example with $p=1000$. The results are split for the signal parameters ($x_{1:4}$) and noise parameters ($x_{5:1000}$). Hamming distance is in percentages. Number in brackets for $DSS$ is using the hyperparameter set $\{\phi_1,\lambda_0,\lambda_1/(1-\phi_1^2),\Theta\}$. Best 5 results in $DSS$ are in bold.}
\label{tab_synth1000}
\scalebox{0.8}{\begin{tabular}{lrlrrlrrlrr}
\hline\hline
                         &        &  & \multicolumn{2}{c}{$x_{1:1000}$} &  & \multicolumn{2}{c}{$x_{1:4}$} &  & \multicolumn{2}{c}{$x_{5:1000}$} \\
$p=1000$                 & Time   &  & SSE              & Ham.          &  & SSE             & Ham.        &  & SSE              & Ham.          \\ \hline
DLM                      & 6.2    &   & 949.5         & 99.7        &   & 936.9      & 27.0        &   & 12.6         & 100       \\
%NGAR                    &  5908.2  &   & 3153529.8     & 99.7         &   & 13313.4    & 27.0       &   & 3140216.4     & 100          \\
LASSO                     & 717.1 &   & 589.4         &     1.7    &   & 537.2      &    57.4      &   & 52.2         &     1.4  \\
DSS                      &        &   &              &             &   &           &             &   &              &             \\
\hline
$\{$.95, .7, 10, .9$\}$  & 854.2  &   & 294.7         & 2.1         &   & 293.6      & 36.5        &   & 1.1         & 2.0         \\
$\{$.95, .7, 10, .95$\}$ & 899.4  &   & 273.2         & 3.7         &   & 272.2      & 29.2        &   & 0.9         & 3.6         \\
$\{$.95, .7, 10, .98$\}$ & 782.8  &   & 302.3         & {\color{black} 6.7}         &   & 300.6      & \textbf{ 19.2}        &   & 1.7         & {\color{black} 6.7}         \\
$\{$.95, .9, 10, .9$\}$  & 658.7  &   & 376.0         & 1.1         &   & 373.9      & 43.0        &   & 2.1         & 0.9         \\
$\{$.95, .9, 10, .95$\}$ & 738.9  &   & 289.3         & 1.7         &   & 288.1      & 35.4        &   & 1.2         & 1.6         \\
$\{$.95, .9, 10, .98$\}$ & 785.4  &   & 288.9         & 3.6         &   & 287.8      & 28.6        &   & 1.1         & 3.5         \\
\hline
$\{$.95, .7, 25, .9$\}$  & 490.7  &   & {\color{black} 597.2}         & \textbf{ 0.7}         &   & {\color{black} 590.0}      & {\color{black} 54.6}        &   & {\color{black} 7.2}         & \textbf{ 0.5}         \\
$\{$.95, .7, 25, .95$\}$ & 630.9  &   & {\color{black} 511.7}         & 1.0         &   & {\color{black} 507.7}      & {\color{black} 50.2}        &   & 4.1         & 0.8         \\
$\{$.95, .7, 25, .98$\}$ & 814.4  &   & 375.8         & 1.7         &   & 373.5      & 43.5        &   & 2.4         & 1.5         \\
$\{$.95, .9, 25, .9$\}$  & 388.8  &   & {\color{black} 663.2}         & \textbf{ 0.5}         &   & {\color{black} 652.2}      & {\color{black} 56.7}        &   & {\color{black} 11.1}         & \textbf{ 0.2}         \\
$\{$.95, .9, 25, .95$\}$ & 527.4  &   & {\color{black} 574.8}         & \textbf{ 0.5}         &   & {\color{black} 567.3}      & {\color{black} 52.6}        &   & {\color{black} 7.5}         & \textbf{ 0.3}         \\
$\{$.95, .9, 25, .98$\}$ & 563.4  &   & 439.5        & \textbf{ 0.9}         &   & 436.0      & {\color{black} 46.8}        &   & 3.5         & \textbf{ 0.7}         \\
\hline
$\{$.98, .7, 10, .9$\}$  & 1018.1 &   & 344.2         & {\color{black} 10.4}        &   & 343.3      & 27.0        &   & 0.9         & {\color{black} 10.3}        \\
$\{$.98, .7, 10, .95$\}$ & 466.6  &   & 397.5         & {\color{black} 13.5}        &   & 389.2      & 28.6        &   & {\color{black} 8.3}         & {\color{black} 13.4}        \\
$\{$.98, .7, 10, .98$\}$ & 300.7  &   & {\color{black} 475.8}         & {\color{black} 16.5}        &   & {\color{black} 450.6}      & 27.5        &   & {\color{black} 25.2}         & {\color{black} 16.4}        \\
$\{$.98, .9, 10, .9$\}$  & 836.6  &   & \textbf{ 235.6}         & 3.4         &   & \textbf{ 235.0}      & \textbf{ 20.5}        &   & \textbf{ 0.6}         & 3.3         \\
$\{$.98, .9, 10, .95$\}$ & 740.9  &   & \textbf{ 251.4}         & 6.4         &   & \textbf{ 250.6}      & \textbf{ 23.7}        &   & \textbf{ 0.8}         & 6.3         \\
$\{$.98, .9, 10, .98$\}$ & 383.3  &   & 413.6         & {\color{black} 11.5}        &   & 409.0      & \textbf{ 26.8}        &   & 4.6         & {\color{black} 11.4}        \\
\hline
$\{$.98, .7, 25, .9$\}$  & 837.2  &   & 341.6         & 1.5         &   & 340.2      & 36.7        &   & 1.5         & 1.4         \\
$\{$.98, .7, 25, .95$\}$ & 953.6  &   & \textbf{ 242.6}         & 2.4         &   & \textbf{ 241.8}      & 28.8        &   & \textbf{ 0.8}         & 2.2         \\
$\{$.98, .7, 25, .98$\}$ & 787.0  &   & \textbf{ 218.1}         & 4.4         &   & \textbf{ 217.3}      & \textbf{ 19.8}        &   & \textbf{ 0.8}         & 4.3         \\
$\{$.98, .9, 25, .9$\}$  & 701.9  &   & 413.6         & \textbf{ 0.8}         &   & 410.9      & 44.7        &   & 2.7         & \textbf{ 0.7}         \\
$\{$.98, .9, 25, .95$\}$ & 760.3  &   & 320.4         & 1.2         &   & 318.9      & 37.5        &   & 1.5         & 1.1         \\
$\{$.98, .9, 25, .98$\}$ & 781.3  &   & \textbf{ 245.4}         & 2.3         &   & \textbf{ 244.5}      & 29.6        &   & \textbf{ 0.9}         & 2.1     \\  
\hline\hline 
\end{tabular}}
\end{table}

\end{document}